\documentclass[acmsmall,screen,nonacm]{acmart}

\usepackage{amsmath}
\usepackage{paralist}
\usepackage{enumitem} %
\usepackage{mathtools} %
\usepackage{graphicx}
\usepackage{multirow}
\usepackage{siunitx} %
\DeclareSIUnit{\rad}{rad}
\usepackage{algorithm} %
\usepackage[noend]{algpseudocode} %
\usepackage{proof} %
\usepackage{accents} %
\usepackage{makecell} %
\usepackage{forloop} %
\usepackage{ifthen} %
\usepackage{lipsum} %
\usepackage{xspace} %
\usepackage{listings} %
\usepackage{makecell}
\usepackage{color}

\AtBeginDocument{%
  }

\setcopyright{acmlicensed}
\acmDOI{10.1145/3720450}
\acmYear{2025}
\acmJournal{PACMPL}
\acmVolume{9}
\acmNumber{OOPSLA1}
\acmArticle{109}
\acmMonth{4}
\received{2024-10-16}
\received[accepted]{2025-02-18}

\newcommand{\EquationsFigureSize}{}

\newcommand{\Support}[1]{\text{supp}(#1)}
\newcommand{\Prob}{\mathbb{P}}
\newcommand{\Expectation}{\mathbb{E}}
\DeclareMathOperator{\Variance}{\text{Var}}
\newcommand{\Distr}{\mathKeyword{Distr}}
\DeclareMathOperator{\Erf}{\text{erf}}
\DeclareMathOperator{\ErfInv}{\text{erf}^{\,-1}}
\newcommand{\SqrtI}[1]{\surd\bigl(#1\bigr)}

\newcommand{\Reals}{\mathbb{R}}
\newcommand{\Nats}{\mathbb{N}}

\newcommand{\Subsets}[1]{\mathcal{P}(#1)}
\newcommand{\Distrs}[1]{\mathcal{D}(#1)}

\newcommand{\List}[1]{\mathKeyword{List}(#1)}
\newcommand{\FDistr}[1]{\mathKeyword{FinDistr}(#1)}

\newcommand{\IsomorphTo}{\cong}
\DeclareMathOperator{\dom}{\text{dom}}
\DeclareMathOperator{\IdFun}{\text{id}}
\newcommand{\parto}{\rightharpoonup} %
\newcommand{\toI}{\mathrel{\!\to\!}} %

\newcommand{\tuple}[1]{\left\langle#1\right\rangle}

\DeclareMathOperator{\FV}{FV}

\newcommand{\scalesembrack}[1]{\scalebox{0.85}{\raisebox{0.6pt}{$#1$}}}
\newcommand{\sem}[1]{\scalesembrack{\llbracket} {#1} \scalesembrack{\rrbracket}}
\newcommand{\bigsem}[1]{\Bigl\llbracket {#1} \Bigr\rrbracket}
\newcommand{\Fvalid}{\vDash}

\newcommand{\FSubst}[3]{{#1}[#2 \backslash #3]}
\newcommand{\UpdateMap}[3]{{#1}[{#2} \gets {#3}]}

\newcommand{\EmptyValuation}{\{\}}
\newcommand{\CatValuation}{\cup}
\newcommand{\BigCatValuation}{\bigcup}

\newcommand{\limply}{\,\rightarrow\,}

\newcommand{\lforall}[1]{\forall #1 \,}
\newcommand{\lexists}[1]{\exists #1 \,}
\newcommand{\limplyI}{\rightarrow}
\newcommand{\True}{\mathsf{true}}
\newcommand{\False}{\mathsf{false}}

\newcommand*{\dif}{\mathop{}\!\mathrm{d}}
\newcommand{\dL}{\dif \text{L}}
\newcommand{\lbox}[2]{\left[#1\right]\,#2}
\newcommand{\lboxI}[2]{[#1]#2}
\newcommand{\ldiamond}[2]{\langle#1\rangle\,#2}
\newcommand{\ldiamondI}[2]{\langle#1\rangle#2}
\newcommand{\dlode}[1]{\{#1\}}
\DeclareMathOperator{\dldom}{\,\&\,}
\DeclareMathOperator{\dlassign}{\coloneq}

\DeclareMathOperator*{\seqI}{;}
\DeclareMathOperator*{\seq}{\,;\,}

\newcommand{\BnfDef}{\ \coloneqq\ \xspace}
\DeclareMathOperator*{\BnfOr}{\ |\ }
\newcommand{\BnfOpSep}{\,,}

\input{macros/symbols}

\newcommand{\UnitSet}{\mathbf{1}}
\newcommand{\Unit}{\raisebox{1pt}{\scalebox{0.7}{$\bullet$}}}
\newcommand{\SumLeft}{\mathsf{left}}
\newcommand{\SumRight}{\mathsf{right}}
\newcommand{\SumLeftInj}[1]{(\SumLeft, {#1})}
\newcommand{\SumRightInj}[1]{(\SumRight, {#1})}

\newcommand{\BoolAnd}[2]{{#1} \text{ and } {#2}}

\newcommand{\BoolTrue}{\text{true}}

\DeclareMathOperator{\InvCCDF}{\overline{\text{\small CDF}}^{\,\raisebox{-2pt}{\scriptsize -1}}}
\newcommand{\InvCCDFwith}[1]{\overline{\text{\small CDF}}^{\,\raisebox{-1pt}{\scriptsize -1}}_{#1}}

\newcommand{\OptNone}{\bot}
\newcommand{\OptReals}{\Reals \cup \{\OptNone\}}

\DeclareMathOperator*{\AnyBinOp}{\odot}
\newcommand{\Funs}{\mathKeyword{Fun}}
\newcommand{\VarsAll}{\mathKeyword{Var}}
\newcommand{\SymbsAll}{\mathKeyword{Symb}}
\newcommand{\StatesAll}{\mathKeyword{VState}}
\newcommand{\InterpsAll}{\mathKeyword{Interp}}

\newcommand{\TagWithIndex}[2]{{#1}^{(#2)}}
\newcommand{\ValSem}[1]{\sem{#1}}
\newcommand{\ProofSem}[1]{\mathcal{P}\sem{#1}}
\newcommand{\InfActionSem}[1]{\mathcal{A}\sem{#1}}

\newcommand{\VarDef}[1]{\mathcal{D}(#1)}
\newcommand{\InferenceStrategy}{\iota}

\newcommand{\Const}{\mathKeyword{Const}}
\newcommand{\Unknown}{\mathKeyword{Unknown}}
\newcommand{\Assum}{\mathKeyword{Assum}}
\newcommand{\StateVar}{\mathKeyword{StateVar}}
\newcommand{\Param}{\mathKeyword{Param}}
\newcommand{\Dir}{\mathKeyword{Dir}}
\newcommand{\UpperAnnot}{\mathKeyword{up}}
\newcommand{\LowerAnnot}{\mathKeyword{lo}}
\newcommand{\GlobalBound}{\mathKeyword{GBound}}
\newcommand{\Bound}{\mathKeyword{Bound}}
\newcommand{\Ctrl}{\mathKeyword{Ctrl}}
\newcommand{\Plant}{\mathKeyword{Plant}}
\newcommand{\Safe}{\mathKeyword{Safe}}
\newcommand{\Inv}{\mathKeyword{Inv}}
\newcommand{\Init}{\mathKeyword{Init}}
\newcommand{\NoiseVar}{\mathKeyword{NoiseVar}}
\newcommand{\Noise}{\mathKeyword{Noise}}
\newcommand{\ObsVar}{\mathKeyword{ObsVar}}
\newcommand{\Obs}{\mathKeyword{Obs}}
\newcommand{\Inference}{\mathKeyword{Infer}}

\newcommand{\BoundConj}[1]{\Bound[#1]}
\newcommand{\BoundWith}[2]{\Bound_{#1}[#2]}
\newcommand{\GlobalParam}{\mathKeyword{GParam}}
\newcommand{\LocalParam}{\mathKeyword{LParam}}
\newcommand{\ProgState}{\mathKeyword{State}}
\newcommand{\Inst}{\mathKeyword{Inst}}
\newcommand{\ObsAvail}{\mathKeyword{ObsAvail}}
\newcommand{\Hist}{\mathKeyword{Hist}}
\newcommand{\CtrlWith}[1]{\Ctrl[#1]}
\newcommand{\EpsType}{[0, 1]}

\newcommand{\ObsMeasurementFun}{\mu}
\newcommand{\ObsAvailFun}{\alpha}
\newcommand{\StateMapping}{\varphi_{\mathsf{s}}}
\newcommand{\ActionMapping}{\varphi_{\mathsf{a}}}
\newcommand{\RevActionMapping}{\varphi_{\mathsf{a}}^{-1}}

\newcommand{\Eps}{\varepsilon}

\newcommand{\EpsRem}{\varepsilon_{\mathKeyword{rem}}}
\newcommand{\ShieldedState}{\hat S}
\newcommand{\ShieldedAction}{\hat A}

\newcommand{\CtrlAc}{a_{\mkern1mu\mathKeyword{ctrl}}}
\newcommand{\InfAc}{a_{\mkern1mu\mathKeyword{inf}}}
\newcommand{\BSymb}{e}
\newcommand{\InfResult}{B}
\newcommand{\GlobBounds}{b_\textsf{\tiny G}}
\newcommand{\LocBounds}{b_\textsf{\tiny L}}

\newcommand{\AlgSample}[1]{\textbf{sample } #1}
\newcommand{\AlgAssert}[1]{\textbf{assert } #1}

\newcommand{\CtrlActionSem}[1]{\mathcal{A}\sem{#1}}
\newcommand{\CtrlMonitorSem}[1]{\mathcal{M}\sem{#1}}
\newcommand{\CtrlExeSem}[1]{\mathcal{E}\sem{#1}}
\newcommand{\CtrlFallbackSem}[1]{\mathcal{F}\sem{#1}}

\newcommand{\VarsVec}[1]{{\vec #1}}

\newcommand{\EmptyList}{\emptyset}
\newcommand{\ListComprehension}[2]{\left[{#1} : {#2}\right]}
\newcommand{\ListSingleton}[1]{\left[ #1 \right]}
\DeclareMathOperator{\ListCat}{\oplus}

\DeclareMathOperator{\WhenSymbBE}{\,\text{when}\,}

\newcommand{\SymbBESem}[1]{\sem{#1}}
\newcommand{\SymbBoundExpr}{\mathKeyword{SBInst}}

\newcommand{\slConfigIndentSize}{1.5em}

\newcommand{\slConfigListingSize}{}

\newcommand{\slkeyword}[1]{\textsc{#1}}

\newcommand{\sldecl}[1]{\slkeyword{#1}\ }
\DeclareMathOperator{\slconstant}{\sldecl{constant}}
\DeclareMathOperator{\slunknown}{\sldecl{unknown}}
\DeclareMathOperator{\slassume}{\sldecl{assume}}
\DeclareMathOperator{\slbound}{\sldecl{bound}}
\DeclareMathOperator{\slcontroller}{\sldecl{controller}}
\DeclareMathOperator{\slplant}{\sldecl{plant}}
\DeclareMathOperator{\slsafe}{\sldecl{safe}}
\DeclareMathOperator{\slinvariant}{\sldecl{invariant}}
\DeclareMathOperator{\slnoise}{\sldecl{noise}}
\DeclareMathOperator{\slobserve}{\sldecl{observe}}
\DeclareMathOperator{\slinfer}{\sldecl{infer}}

\DeclareMathOperator{\pbmonotonicity}{\sldecl{monotonicity}}\DeclareMathOperator{\pbmodel}{\sldecl{model}}
\DeclareMathOperator{\pbsafe}{\sldecl{safe}}
\DeclareMathOperator{\pbtotality}{\sldecl{totality}}
\DeclareMathOperator{\pbinference}{\sldecl{inference}}

\DeclareMathOperator{\slaggregate}{\slkeyword{aggregate}\,}
\DeclareMathOperator{\slbest}{\slkeyword{best}\,}
\DeclareMathOperator*{\slwhen}{\,\slkeyword{when}\,}

\DeclareMathOperator*{\slassign}{\,\coloneq\,}
\DeclareMathOperator*{\sland}{\,\slkeyword{and}\,}

\newcounter{slindCounter}
\newcommand{\slind}[1][1]{%
  \forloop{slindCounter}{0}{\value{slindCounter}<#1}{\hspace{\slConfigIndentSize}}}

\newenvironment{sllisting}{
  \renewcommand{\arraystretch}{\slConfigLineStretch}
  \slConfigListingSize
  \noindent
  \tabular{@{}>{$}l<{$}}
}{
  \endtabular
  \renewcommand{\arraystretch}{1}}

\newcommand{\slnormal}{\mathcal{N}}
\newcommand{\sluniform}{\mathcal{U}}
\newcommand{\slbernouilli}{\mathcal{B}}

\newcommand{\mathKeyword}[1]{\mathsf{#1}}

\newcommand{\makeMathKeywordText}{\expandafter\renewcommand{\mathKeyword}[1]{\textsf{##1}}}
\newcommand{\makeMathKeywordMath}{\expandafter\renewcommand{\mathKeyword}[1]{\mathsf{##1}}}

\DeclareMathOperator*{\DefEq}{\ = \ }
\newcommand{\commasep}{,\ }

\newcommand{\aligncolon}{\;:\;}

\newcommand{\Equiv}{\,\equiv\,}



\newcommand{\TrainThetaSlope}{\theta}
\newcommand{\TrainThetaBias}{\varphi}
\newcommand{\TrainThetaSlopeMax}{\bar\TrainThetaSlope}
\newcommand{\TrainThetaSlopeMin}{\underaccent{\bar}{\TrainThetaSlope}}
\newcommand{\TrainThetaBiasMax}{\bar\TrainThetaBias}
\newcommand{\TrainThetaBrakingRate}{\TrainThetaSlopeMin B - \TrainThetaBiasMax}
\newcommand{\TrainThetaAcc}{\TrainThetaSlopeMax u + \TrainThetaBiasMax}

\newcommand{\bdist}[2]{\textsf{Bdist}_{#1}(#2)}
\newcommand{\TrainFxMax}{\bar {f}}

\algnewcommand\algorithmicswitch{\textbf{switch}}
\algnewcommand\algorithmiccase{\textbf{case}}
\algnewcommand\algorithmicassert{\texttt{assert}}
\algnewcommand\Assert[1]{\State \algorithmicassert(#1)}%
\algdef{SE}[SWITCH]{Switch}{EndSwitch}[1]{\algorithmicswitch\ #1\ \algorithmicdo}{\algorithmicend\ \algorithmicswitch}%
\algdef{SE}[CASE]{Case}{EndCase}[1]{\algorithmiccase\ #1}{\algorithmicend\ \algorithmiccase}%
\algtext*{EndSwitch}%
\algtext*{EndCase}%

\begin{document}

\title{Adaptive Shielding via Parametric Safety Proofs}

\author{Yao Feng}
\authornote{Yao Feng contributed a majority of the case studies and experiments.}
\orcid{0000-0002-8213-5181}
\affiliation{%
  \institution{Tsinghua University}
  \city{Beijing}
  \country{China}
}
\email{y-feng23@mails.tsinghua.edu.cn}

\author{Jun Zhu}
\orcid{0000-0002-6254-2388}
\affiliation{%
  \institution{Tsinghua University}
  \city{Beijing}
  \country{China}
}
\email{dcszj@mail.tsinghua.edu.cn}

\author{André Platzer}
\orcid{0000-0001-7238-5710}
\affiliation{%
  \institution{Karlsruhe Institute of Technology (KIT)}
  \city{Karlsruhe}
  \country{Germany}
}
\email{platzer@kit.edu}

\author{Jonathan Laurent}
\authornote{Jonathan Laurent contributed a majority of the theoretical framework and writing.}
\orcid{0000-0002-8477-1560}
\affiliation{%
  \institution{KIT}
  \city{Karlsruhe}
  \country{Germany}
}
\affiliation{%
  \institution{Carnegie Mellon University}
  \city{Pittsburgh}
  \country{USA}
}
\email{jonathan.laurent@kit.edu}

\renewcommand{\shortauthors}{Feng, Zhu, Platzer and Laurent}

\begin{abstract}
  A major challenge to deploying cyber-physical systems with learning-enabled controllers is to ensure their safety, especially in the face of changing environments that necessitate runtime knowledge acquisition. Model-checking and automated reasoning have been successfully used for shielding, i.e., to monitor untrusted controllers and override potentially unsafe decisions, but only at the cost of hard tradeoffs in terms of expressivity, safety, adaptivity, precision and runtime efficiency. We propose a programming-language framework that allows experts to statically specify \emph{adaptive shields} for learning-enabled agents, which enforce a safe control envelope that gets more permissive as knowledge is gathered at runtime. A shield specification provides a safety model that is parametric in the current agent's knowledge. In addition, a nondeterministic inference strategy can be specified using a dedicated domain-specific language, enforcing that such knowledge parameters are inferred at runtime in a statistically-sound way. By leveraging language design and theorem proving, our proposed framework empowers experts to design adaptive shields with an unprecedented level of modeling flexibility, while providing rigorous, end-to-end probabilistic safety guarantees.
\end{abstract}



\maketitle

\section{Introduction}

Learning-based methods such as reinforcement learning have shown great promise in the fields of autonomous driving~\cite{DBLP:journals/corr/Shalev-ShwartzS16a,DBLP:conf/eccv/LiangWYX18} and robot control~\cite{DBLP:journals/jirs/RovedaMFABTP20,DBLP:conf/icra/HesterQS12,DBLP:conf/icra/SmartK02}. However, their deployment in the real-world has been held back by reliability and safety concerns. The use of formal methods has been suggested to guarantee the safety of learning-enabled systems, both \emph{after} and \emph{during} training~\cite{DBLP:conf/aaai/FultonP18,DBLP:conf/tacas/FultonP19,DBLP:conf/atal/Elsayed-AlyBAET21,DBLP:conf/hybrid/HuntFMHDS21}. Many different approaches have been proposed to do so, which share as a foundation the idea of \emph{sandboxing} or \emph{shielding}~\cite{DBLP:journals/fmsd/MitschP16,DBLP:conf/aaai/AlshiekhBEKNT18,DBLP:conf/aaai/FultonP18,DBLP:conf/icra/ThummA22}. In this framework, the intended actions of a learning-enabled agent are monitored at runtime and overridden by appropriate fallbacks whenever they cannot be proved safe with respect to a model of the environment.

Most existing approaches leverage automated techniques for proving safety such as reachability analysis~\cite{DBLP:conf/icra/ThummA22,DBLP:conf/cdc/KollerBT018}, LTL model-checking~\cite{DBLP:conf/aaai/AlshiekhBEKNT18,DBLP:conf/isola/KonighoferL0B20} or Hamilton-Jacobi solving~\cite{DBLP:journals/tac/FisacAZKGT19}. Such techniques can be applied \emph{offline} to precompute numerical, table-based control envelopes that indicate what actions are safe in every possible state of the system. Alternatively, they can be used at runtime to analyze the safety of different actions with respect to the current state. Such \emph{online} methods mostly offer only bounded-horizon guarantees and must typically make aggressive tradeoffs to regain efficiency at the cost of precision or generality. However, an advantage of online methods is that they allow the model to be updated at runtime as the agent gathers information about its environment. Indeed, in many situations, a fully-specified model of the environment is not known at design-time and using a necessarily conservative static model may lead to overly cautious behavior. In principle, adaptivity can also be offered by offline methods, by precomputing control envelopes for \emph{every} possible model that may be considered at runtime. However, doing so further compounds the scalability challenges of model-checking. Finally, existing approaches to adaptive shielding, whether \emph{online} or \emph{offline}, are either offering limited expressivity to encode model uncertainty (e.g. \emph{bounded disturbance terms}~\cite{althoff2014online} or \emph{finite model families}~\cite{DBLP:conf/tacas/FultonP19}) or failing to provide rigorous end-to-end guarantees under assumptions that can be easily validated (e.g. methods based on learning Gaussian processes~\cite{DBLP:conf/aaai/ChengOMB19,DBLP:conf/eucc/BerkenkampS15,DBLP:conf/nips/BerkenkampTS017,DBLP:journals/tac/FisacAZKGT19}).

This work introduces a programming-language framework for designing \emph{adaptive} safety shields, leveraging a minimal yet crucial amount of human insight to target a very general class of model families, while shifting the full burden of safety analysis offline. In this framework, control envelopes are described by nondeterministic controllers, whose safety is established using differential dynamic logic~\cite{DBLP:journals/jar/Platzer08,DBLP:journals/jar/Platzer17}. Environment models can be parametrized by unknown function symbols. Those function symbols are subject to parametric bounds, whose parameters are accessible to the controller and can be \emph{instantiated and refined at runtime}. Updating these parameters in a way that is sound and efficient is a great challenge in itself. We address this challenge by offering experts a domain-specific language to express nondeterministic inference strategies that are \emph{sound by construction}, and whose nondeterminism is resolved by an \emph{inference policy} that can be programmed or learned in the same way that the \emph{control policy} is.

We illustrate our framework on four case studies, demonstrating its safety and its ability to handle advanced uncertainty models that are beyond the reach of existing methods. We also showcase the possibility of having agents learn how to manage their own safety budget by learning inference policies, which -- to the best of our knowledge -- has no equivalent in the literature.

\section{Overview}

\newcommand{\JSCModel}{\mathKeyword{Model}}

In this section, we motivate and illustrate our framework on a series of closely-related examples.

\subsection{Extracting Shields from Verified Nondeterministic Controllers}

Our framework builds on the idea of extracting runtime controller monitors from provably safe, nondeterministic controllers~\cite{DBLP:journals/fmsd/MitschP16,DBLP:conf/aaai/FultonP18}. The safety of such nondeterministic controllers is established using \emph{differential dynamic logic} ($\dL$)~\cite{DBLP:journals/jar/Platzer08,DBLP:journals/jar/Platzer17}, which is designed specifically to verify hybrid systems with both discrete and continuous dynamics. In $\dL$, we use hybrid programs (HPs) to model both controllers and the differential equations of the physical environments they interact with. A summary of the syntax and semantics of hybrid programs can be found in Table~\ref{tab:hps}. Proving the safety of a cyber-physical system using $\dL$ typically comes down to proving a modal formula of the form $\Init \!\limplyI\! \lboxI{(\Ctrl \seqI \Plant)^*}{\Safe}$, where $\Init$ and $\Safe$ are logical formulas while $\Ctrl$ and $\Plant$ are hybrid programs. This means that under a certain initial condition ($\Init$), no matter how often we execute the discrete-time controller ($\Ctrl$) and let the continuous-time system evolve to the next control cycle $(\Plant)$, the safety condition ($\Safe$) is satisfied.

\begin{table}%
  \caption{Semantics of Hybrid Programs in $\dL$}\label{tab:hps}
  \vspace{-0.2cm}
  \centering
  \renewcommand{\arraystretch}{1.2}
  \begin{tabular}{p{0.18\linewidth}p{0.77\linewidth}}
  \toprule
  Syntax  &  Semantics\\
  \midrule
  $x \dlassign e$ & Assign the value of $e$ to variable $x$, leaving all other variables unchanged.\\
  $x \dlassign *$ & Assign variable $x$ nondeterministically to some real value. \\
  $?Q$ & If $Q$ is {true}, continue running; else abort. \\
  $x'=f(x) \dldom Q$ & Follow the system of differential equations $x'=f(x)$ for a certain (nondeterministic) amount of time while $Q$ holds true.\\
  $\alpha \cup \beta$ & Nondeterministically run either HP $\alpha$ or $\beta$.\\
  $\alpha \seqI \beta$ & Sequentially run $\beta$ after $\alpha$.\\
  $\alpha^*$ & Run $\alpha$ repeatedly for any $\geq0$ amount of iterations.\\
  \bottomrule
  \end{tabular}
\end{table}

For example, the $\dL$ model described in Figure~\ref{fig:overview-etcs} designates a nondeterministic train controller that must stop by the \emph{end of movement authority} $e$ located somewhere ahead of the train, as assigned by the train network scheduler~\cite{DBLP:conf/hybrid/PlatzerQ08}. Variable $x$ models the position of the train on its tracks. At every control cycle, the driver can choose between braking with acceleration $-B < 0$ or accelerating with acceleration $A > 0$. It then gives control to the environment until the next control cycle, which happens in at most $T$ seconds. However, the option to accelerate is only acceptable when a sufficient distance remains to accelerate and then brake safely. The safety of this nondeterministic train controller can be established by proving the validity of Formula~\ref{eq:example-model-def}, using the rules and axioms of $\dL$. This can be done interactively in a proof assistant such as KeYmaera X~\cite{DBLP:conf/cade/FultonMQVP15}. In particular, doing so requires finding a loop invariant $\Inv$ that holds initially, implies the safety property and is preserved by any run of hybrid program $\Ctrl \seqI \Plant$. Here, one can take $\Inv \equiv \Init$.

\begin{figure}
  \begin{center}
    \begin{align}
      \JSCModel &\Equiv \Init \limplyI \lbox{(\Ctrl \seqI \Plant)^*}{\Safe} \label{eq:example-model-def} \\
      \Init &\Equiv (A>0 \land B>0 \land T>0) \land (x + v^2/2B \le e) \nonumber \\
      \Ctrl &\Equiv (a \dlassign -B) \cup (?(x + vT + AT^2/2 + (v+AT)^2/2B \le e) \seq a \dlassign A) \label{eq:example-model-ctrl} \\
      \Plant &\Equiv t \dlassign 0 \;\,;\, \{ x'=v \commasep v'=a \commasep t'=1 \dldom t \le T \land v \ge 0 \} \nonumber \\
      \Safe &\Equiv  x \le e \nonumber
    \end{align}
  \end{center}
  \vspace{-0.2cm}
  \caption{A Simple $\dL$ Model of a Train Braking-Control System~\cite{DBLP:conf/hybrid/PlatzerQ08}}\label{fig:overview-etcs}
  \Description[]{}
\end{figure}

Crucially, one can extract a runtime monitor from the nondeterministic controller shown in Equation~\ref{eq:example-model-ctrl} and use it as a shield for a learning-enabled agent~\cite{DBLP:conf/aaai/FultonP18}. For example, one could use reinforcement learning to learn a deterministic train controller that optimizes for speed, energy efficiency and passenger comfort, while guaranteeing safety by overriding acceleration whenever the distance to the end-of-movement authority is insufficient (i.e. the guard in Equation~\ref{eq:example-model-ctrl} is violated).

\subsection{Adding Adaptivity via Parametric Bounds}

As illustrated in the previous section, previous work~\cite{DBLP:journals/fmsd/MitschP16,DBLP:conf/aaai/FultonP18} has shown how to extract shields from nondeterministic controllers whose safety is established using differential dynamic logic. However, doing so requires an accurate model of every safety-relevant aspect of the system. In cases where only a conservative model is available, the shield may force the system into an overly cautious behavior. One contribution of this work is to allow the design of \emph{adaptive} shields, whose behavior is refined at runtime as the agent gathers more information about its environment. To do so, we allow environment models to be parameterized by unknown quantities and control envelopes to be parameterized by bounds on these same quantities.

An example of such an adaptive shield is described in Figure~\ref{fig:overview-train-global}. It is specified using our proposed \emph{shield specification language}, which we define rigorously in Sections~\ref{sec:adaptive-shielding}~and~\ref{sec:inference-strategy-language}. This example is a variation of the train control system from Figure~\ref{fig:overview-etcs} with a continuous action space: instead of making a binary choice between \emph{accelerating} and \emph{braking}, the agent must select a \emph{commanded acceleration} $u \in [-B, A]$, which is expressed in $\dL$ using a nondeterministic assignment $u \dlassign *$ followed by a test $?(-B \!\le\! u \!\le\! A)$. In addition, the relationship between the commanded acceleration $u$ and the \emph{actual acceleration} $a$ of the train is governed by an unknown linear function: $a = \TrainThetaSlope u + \TrainThetaBias$ where $\TrainThetaSlope > 0$ and $\TrainThetaBias$ are unknown real parameters. At runtime, the actual acceleration of the train is measured repeatedly, allowing the agent to compute increasingly precise estimates of the systematic disturbance $\TrainThetaSlope$ and $\TrainThetaBias$ and thus allowing the shield to be increasingly permissive.

The shield specification in Figure~\ref{fig:overview-train-global} consists of eleven sections, each of them introduced by a distinct keyword. \textsc{constant} introduces real quantities that are known by the agent at runtime such as the maximum commanded acceleration $A$. \textsc{unknown} introduces quantities that are \emph{not} known and must be estimated at runtime, namely $\TrainThetaSlope$ and $\TrainThetaBias$. \textsc{assume} gathers global assumptions about constant and unknown symbols such as $A > 0$ and $\TrainThetaSlope > 0$. Such assumptions could be encoded as system invariants instead (see \textsc{invariant}) but separating them from state-dependent invariants leads to greater conceptual clarity and more concise proof obligations (global assumptions are preserved by definition). The \textsc{controller}, \textsc{plant}, \textsc{safe} and \textsc{invariant} sections define a $\dL$ model similar in shape to the one already studied in Figure~\ref{fig:overview-etcs}. However, the plant model can feature unknown symbols, which it does here via the $v' = \TrainThetaSlope u + \TrainThetaBias$ differential equation. Also the controller can depend on \emph{bound parameters} that constrain these unknowns and that are introduced in the \textsc{bound} section. Here, we introduce a lower bound $\TrainThetaSlopeMin$ on $\TrainThetaSlope$, an upper bound $\TrainThetaSlopeMax$ on $\TrainThetaSlope$, and an upper bound $\TrainThetaBiasMax$ on $\TrainThetaBias$. These parameters are instantiated and refined at runtime as the agent is gathering knowledge, using statistical inference. 
The controller is similar to Equation~\ref{eq:example-model-ctrl}, except that it offers a continuous action space and conservatively assumes a maximum achievable braking rate of $\TrainThetaBrakingRate$. Thus, the control envelope defined by the nondeterministic controller gets increasingly permissive as tighter bounds are obtained on $\TrainThetaSlope$ and $\TrainThetaBias$.

The invariant can depend on bound parameters but only monotonically, in the sense that it can only be made more permissive by a tightening of the bounds. This is an intuitive requirement: acquiring knowledge must never transition an agent from a state considered safe to a state considered unsafe. Here, this requirement clearly holds since the maximum guaranteed braking rate of $\TrainThetaBrakingRate$ can only get larger as $\TrainThetaSlopeMin$ and $\TrainThetaBiasMax$ tighten. Finally, \textsc{noise} and \textsc{observe} specify the kind of information that may become available at runtime to compute and refine such bounds. Here, we are assuming that at every control cycle, an estimate of the \emph{actual} acceleration of the system \emph{may} be measured, with some Gaussian noise of standard deviation $\sigma$. No guarantee is offered on \emph{when} and \emph{under what conditions} such an observation becomes available. $\textsc{infer}$ defines a family of provably-sound strategies for updating $\TrainThetaSlopeMin, \TrainThetaSlopeMax$ and $\TrainThetaBiasMax$ based on such observations, as we discuss shortly in Section~\ref{sec:overview-inference}.
Before we do so though, we introduce a last variant of our train control shield that illustrates how \emph{functional} unknowns can be handled, using the concept of a \emph{local bound}.

\begin{figure}
  \begin{center}
  \begin{sllisting}
    \slconstant A, B, T, \sigma \\
    \slunknown \TrainThetaSlope, \TrainThetaBias \\
    \slassume
      A > 0 \commasep
      B > 0 \commasep
      T > 0 \commasep
      \sigma > 0 \commasep
      \TrainThetaSlope > 0 \\
    \slbound
      \TrainThetaSlopeMin: \TrainThetaSlopeMin \le \TrainThetaSlope \commasep
      \TrainThetaSlopeMax: \TrainThetaSlopeMax \ge \TrainThetaSlope \commasep
      \TrainThetaBiasMax: \TrainThetaBiasMax \ge \TrainThetaBias
      \\
    \slcontroller \\
      \slind u \dlassign * \ ;\, ?(-B \!\le\! u \!\le\! A) \,;\, %
      ?(x + vT + (\TrainThetaAcc)T^2/2 + (v + (\TrainThetaAcc)T)^2/2(\TrainThetaBrakingRate) \le e) \\
    \slplant t \dlassign 0 \seq \dlode{x'=v, v'= \TrainThetaSlope u + \TrainThetaBias, t'=1 \dldom t \le T \land v \ge 0} \\
    \slsafe x \le e \\
    \slinvariant %
      (\TrainThetaBrakingRate > 0) \land (x + v^2/2(\TrainThetaBrakingRate) \le e)  \\
    \slnoise \eta \sim \slnormal(0, \sigma^2) \\
    \slobserve \omega = \TrainThetaSlope u + \TrainThetaBias - \eta \\
    \slinfer \\
      \slind \TrainThetaSlopeMin, \TrainThetaSlopeMax \slassign \slaggregate i, j: (\omega_j - \omega_i)/(u_j - u_i) \sland (\eta_j - \eta_i)/(u_j - u_i) \slwhen u_j > u_i \seq \\
      \slind \TrainThetaBiasMax \slassign \slaggregate i: \omega_i - \TrainThetaSlopeMax u_i \sland \eta_i \slwhen u_i \le 0
  \end{sllisting}
  \end{center}
  \caption{An adaptive shield for a train control system, where the relationship between the commanded and actual train acceleration is governed by an unknown linear function.}\label{fig:overview-train-global}
  \Description[]{}
\end{figure}

\subsection{Handling Functional Unknowns with Local Bounds}\label{sec:overview-train-local}

In our previous example from Figure~\ref{fig:overview-train-global}, two real-valued unknowns are estimated at runtime. However, our framework can handle a more powerful and flexible form of model-uncertainty in the form of \emph{functional unknowns}. Figure~\ref{fig:overview-train-local} provides an example, considering a train that evolves on tracks of unknown, varying slope. Variable $x$ denotes the train position on an arc-length parametrization of the tracks. A train at position $x$ is subject to an additional acceleration term of $-g\cdot\sin(\theta_x)$ where $g \approx \qty{9.81}{\metre\per\s\squared}$ is Earth's gravitational acceleration and $\theta_x$ the track angle at coordinate $x$. We model this influence by defining the train's kinematics as $v' = a + f(x)$, with $a$ the acceleration commanded by the controller and $f$ an unknown function of $x$. In addition, $f$ is assumed to be $k$-Lipschitz, globally lower-bounded by $-A$ (ensuring that the train will always move forward when instructed to accelerate at rate $A$) and upper-bounded by a known constant $F$. One could simply use this global bound to implement a conservative shield. Our challenge is to do better by estimating superior, local estimates of $f$ based on runtime observations.

\begin{figure}
  \begin{center}
  \begin{sllisting}
    \slconstant A, B, F, k, \sigma \\
    \slunknown f(*) \\
    \slassume \\
      \slind A > 0 \commasep B > 0 \commasep T > 0 \commasep k > 0 \commasep \sigma > 0 \commasep F < B \commasep A + F > 0 \commasep \\
      \slind (\lforall x -A \le f(x) \le F) \commasep (\lforall x \lforall y |f(x) - f(y)| \le k |x - y|) \\
    \slbound %
      \TrainFxMax: f(x) \le \TrainFxMax
    \\
    \slcontroller \\
      \slind y \dlassign \min(y, \TrainFxMax) \seq \\
      \slind ((a \dlassign -B) \ \cup \\
      \slind \phantom{(}(?(x + vT + \frac{1}{2}(A + F)T^2 + \bdist{v+(A+F)T}{B - \min(F, \, y + k(vT + \frac{1}{2}(A + F)T^2) \ + \\
      \slind[2] k \cdot \bdist{v + (A+F)T}{B-F})} \le e) \seq a \dlassign A) \\
    \slplant t \dlassign 0 \seq \dlode{x'=v, v'=a+f(x), y'=kv, t'=1 \dldom t \le T \land v \ge 0} \\
    \slsafe x \le e \\
    \slinvariant \,(v \ge 0) \land
      (y \ge f(x)) \land (x + \bdist{v}{B - \min(F, \,y + k \!\cdot\! \bdist{v}{B-F})} \le e)
       \\
    \slnoise \eta \sim \slnormal(0, \sigma^2) \\
    \slobserve \omega = f(x) - \eta \\
    \slinfer
      \TrainFxMax \slassign F \seq
      \TrainFxMax \slassign \slbest i: \TrainFxMax_i + k|x - x_i| \seq
      \TrainFxMax \slassign \slaggregate i: \omega_i + k|x - x_i| \sland \eta_i \\
  \end{sllisting}
  \end{center}
  \caption{An adaptive shield for a train control system where the railway tracks are assumed to follow an unkown, space-varying slope function. We write $\bdist{v}{b} \equiv v^2/2b$. For simplicity, the train faces a \emph{binary} choice between \emph{acceleration} and \emph{braking} (as in Figure~\ref{fig:overview-etcs} and unlike in Figure~\ref{fig:overview-train-global}).}\label{fig:overview-train-local}
  \Description[]{}
\end{figure}

We do so by introducing the concept of a \emph{local bound}. Figure~\ref{fig:overview-train-local} defines such a  bound, namely $\TrainFxMax \ge f(x)$. As opposed to the \emph{global} bounds used in the previous section, this bound involves quantity $f(x)$ that is state-dependent. The guarantee that comes with such a definition is that before every control cycle, the inference module must provide a value $\TrainFxMax$ that is an upper bound on the value of $f(x)$, as evaluated in the current state. The value of $\TrainFxMax$ can be used in the controller. However, $\TrainFxMax$ \emph{cannot} be mentioned in the invariant since $\TrainFxMax \ge f(x)$ may not hold anymore as the plant executes and the train moves further on the tracks. A local bound is only guaranteed to be valid at discrete points in time, after each run of the inference module and right before the controller executes. Thus, we also maintain an upper bound \emph{y} on $f(x)$ that holds \emph{throughout} the plant and can therefore be used in the invariant. Variable $y$ is updated before the controller executes with the value of $\TrainFxMax$ whenever the latter provides a tighter bound. Then, it is evolved via the differential equation $y'=kv$, degrading the precision of our bound at a rate proportional to the Lipschitz constant of $f$ to ensure its preservation by the plant ($\dif f/\dif t \le \dif f / \dif x \cdot \dif x / \dif t \le k \cdot v$).

We can now derive our invariant, from which the controller's acceleration guard follows. Given a particular state and provided a bound $y \ge f(x)$, we wonder whether the train can be kept safe indefinitely from the current state by fully engaging the brakes. Given a constant, {effective} braking rate of $b$ and starting with speed $v$, it can be shown that the distance needed for the train to brake to a full stop is $\bdist{v}{b} \equiv v^2/2b$. %
Thus, a sufficient condition for the train to be safe indefinitely is $x + \bdist{v}{B-F} \le e$, since we can guarantee an effective braking rate of at least $B-F$ using the global bound $F$ on $f$. However, a stronger guarantee may result from combining the local bound $y \ge f(x)$ and the fact that $f$ is $k$-Lipschitz. Indeed, since we already know from our naive estimate that the train can stop within distance $\bdist{v}{B-F}$, we also know that the value of $f$ along this trajectory must be upper-bounded by $y + k \cdot \bdist{v}{B-F}$. This gives us an effective braking rate along the stopping trajectory of $B - (y + k \cdot \bdist{v}{B-F})$. In turn, a new stopping distance can be computed from this estimate, yielding the invariant shown in Figure~\ref{fig:overview-train-local}. The controller's acceleration guard follows from the invariant, as it simply ensures that the invariant will still hold after executing the plant for time $T$.

So far, we have seen how adaptive shields can be extracted from nondeterministic, parametric controllers. Our next step is to show how inference modules can also be synthesized that are guaranteed to instantiate such parameters from runtime observations in a statistically sound way.

\subsection{Inferring Statistically-Sound Bound Parameters}\label{sec:overview-inference}

In this section, we delve into our proposed \emph{inference strategy language}, which is used to specify nondeterministic inference strategies for instantiating bound parameters at runtime. Inference strategies are introduced by the \textsc{infer} keyword. We consider the strategy from Figure~\ref{fig:overview-train-local} as an example (the strategy from Figure~\ref{fig:overview-train-global} is analyzed in Appendix~\ref{ap:train-global-inf-strategy}). An inference strategy consists of a sequence of \emph{inference assignments}, where each assignment computes a value for a particular bound parameter and updates this parameter whenever the new value is tighter than the old one. Our language supports three forms of assignments, which are all represented in the strategy from Figure~\ref{fig:overview-train-local}. Each assignment yields a proof obligation that establishes its soundness.

The first assignment $\TrainFxMax \slassign F$ indicates that the global bound $F$ can always be used as a local bound on $f(x)$. Its soundness is justified by the validity of the following proof obligation, which is automatically generated by our framework and to be proved by the user: \[\cdots \land(\lforall{x}  (-A \le f(x) \le F)) \land \cdots \limply f(x) \le F.\] Here, the right-hand-side of the implication is simply the definition of $\TrainFxMax$, where $\TrainFxMax$ has been substituted into the assignment's right-hand-side. The left-hand-side consists of the global assumptions, from which only the relevant ones are shown here. The second assignment $\TrainFxMax \slassign \slbest i: \TrainFxMax_i + k|x - x_i|$ indicates that any previously established local bound $\TrainFxMax_i$ induces a new bound $\TrainFxMax_i + k|x-x_i|$ for the current state. The associated proof obligation is: \[  (\lforall x \lforall y |f(x) - f(y)| \le k |x - y|) \land (f(x_i) \le \TrainFxMax_i) \limply f(x) \le \TrainFxMax_i + k|x - x_i|, \] where all irrelevant assumptions have been removed for clarity. To a first approximation, the semantics of this assignment is to compute one such bound for every element in the agent's history and assign the tightest one to $\TrainFxMax$ if it beats its current value. The third assignment $\TrainFxMax \slassign \slaggregate i: \omega_i + k|x - x_i| \sland \eta_i$ is most interesting and the one that performs statistical inference from observations. To understand it better, let us start from the associated proof obligation:
\begin{equation}\label{eq:train-local-aggregate-obligation}
(\lforall x \lforall y |f(x) - f(y)| \le k |x - y|) \land (\omega_i = f(x_i) - \eta_i) \limply (f(x) \le \omega_i + k|x-x_i| + \eta_i).
\end{equation}
The validity of this obligation establishes the inequality $f(x) \le \omega_i + k|x-x_i| + \eta_i$ for every triple $(x_i, \omega_i, \eta_i)$ in the agent's history. This does not directly give us a usable bound since $\omega_i$ and $\eta_i$ are random variables, only the first one of which is observed. However, given a tolerance-level of $\Eps$, we can use the standard tail bound on Gaussians to establish a concrete upper-bound on $f(x)$ that holds with probability at least $1 - \Eps$. Here, we get the upper-bound $b_i \equiv \omega_i + k|x-x_i| + \sigma \cdot z_\Eps$, where  $z_\Eps \equiv \sqrt{2} \cdot \ErfInv(1 - 2\Eps)$ and $\ErfInv$ is the inverse of the \emph{error function}~\cite{bertsekas2008introduction}. Indeed:
\[ \Prob\{f(x) > b_i \} \,\le\, \Prob\{ \omega_i + k|x-x_i| + \eta_i > b_i \} \,=\, \Prob \{ \eta_i > \sigma \cdot z_\Eps \} \,=\, \Eps. \]
The first inequality above is a consequence of Equation~\ref{eq:train-local-aggregate-obligation}, the middle equality holds by substituting the definition of $b_i$, and the last equality holds by definition of the error function.

\newcommand{\OverviewAggBound}{b_{\lambda, \Eps}}

Such a probabilistic bound may be acceptable in cases where measurement noise is small. However, when $\sigma$ is large, it may be unacceptably conservative. Fortunately, we can get tighter bounds by aggregating multiple independent observations together. More precisely, consider some nonnegative coefficients $\lambda_i$ that sum up to one. Then, one can show that the following is a bound on $f(x)$ with probability at least $1 - \Eps$:
\begin{equation}\label{eq:overview-train-local-aggr}
  \OverviewAggBound \ \equiv \ \sum_i \lambda_i \cdot (\omega_i + k|x-x_i|) \ + \ \sqrt{\sum_i \lambda_i^2} \cdot \sigma \cdot z_\Eps.
\end{equation}
The proof is similar to the one for Equation~\ref{eq:train-local-aggregate-obligation}. We have $f(x) > b \limplyI \omega_i + k|x-x_i| + \eta_i > b$ for all $i$ and $b$. Thus, taking a convex combination, we have $f(x) > b \limplyI \sum_i \lambda_i (\omega_i + k|x-x_i| + \eta_i) > b$ for all $b$. As a consequence:
\begin{equation}\label{eq:overview-train-local-aggr-proof}
  \Prob\{f(x) > \OverviewAggBound \} \ \le\  \Prob\,\biggl\{  \,\sum_i \lambda_i (\omega_i + k|x-x_i| + \eta_i) > \OverviewAggBound \, \biggr\} \ = \  \Prob\,\biggl\{ \,\sum_i \lambda_i\eta_i > \sigma' z_\Eps\, \biggr\} \,=\, \Eps,
\end{equation}
where $\sigma' \!\equiv\! \sigma \sqrt{\sum_i \lambda_i^2}$ is the standard deviation of $\sum_i \lambda_i \eta_i$, which is a zero-mean Gaussian provided that the choice of $\lambda$ is uninformed by -- and thus \emph{independent} from -- the values of $\eta_i$ and $\omega_i$.

The bound obtained in Equation~\ref{eq:overview-train-local-aggr} is the sum of two terms. Minimizing the first term requires putting all the weight on the \emph{closest} measurement, where $k|x-x_i|$ is smallest. On the other hand, minimizing the second term requires spreading the weights uniformly between measurements, resulting in a value of $(\sigma \cdot z_\Eps) / \sqrt{n}$, where $n$ is the number of considered measurements. Finding the right tradeoff can be subtle and situation-dependent. Similarly, the value chosen for $\Eps$ at every control cycle must be substracted from a global \emph{safety budget}, whose proper management constitutes an important challenge. Fortunately, our framework does not force such choices on the shield designers. Rather, it enables the specification of nondeterministic inference strategies, where the choice of $\Eps$ and $\lambda$ is determined at each control cycle by a non-soundness-critical \emph{inference policy} that can be learned similarly to the \emph{control policy}.

The \textsc{aggregate} construct from our inference strategy language allows generalizing the reasoning above to a large class of bounds and distributions. Provided a symbolic bound that can be expressed as the sum of an \emph{observable} component and of a \emph{noise} component (syntactically separated using the \textsc{and} keyword), \textsc{aggregate} computes a weighted average of multiples instances of such bounds, using probabilistic tail inequalities to handle the aggregated noise term. More generally, a strategy written in our proposed inference languages can be compiled into one proof obligation per inference assignment, along with an inference module that integrates with an inference policy at runtime. %

\subsection{System Overview}

\begin{figure}
  \includegraphics[width=0.95\textwidth]{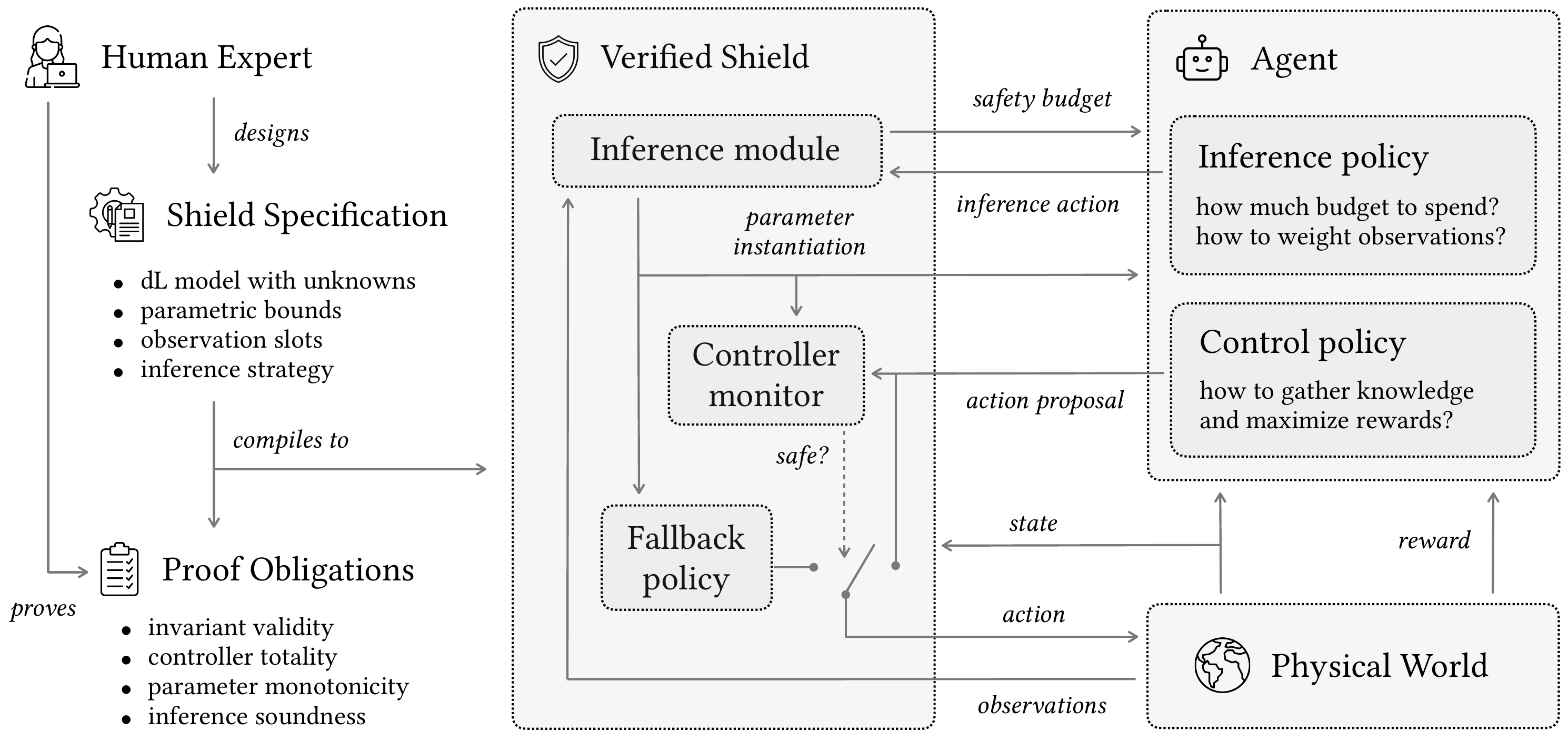}
  \caption{Adaptive Shielding Overview Diagram}\label{fig:overview-diagram}
  \Description[]{}
\end{figure}

Figure~\ref{fig:overview-diagram} provides a diagrammatic summary of our proposed framework for designing adaptive shields. In this framework, human experts are offered a domain-specific language for specifying shields in the form of parametric $\dL$ models coupled with nondeterministic inference strategies. A specification written in this language can be automatically compiled into a shield and into a series of proof obligations that are sufficient in establishing its soundness. All obligations are $\dL$ formulas, which can be discharged automatically in many cases but also proved interactively in a proof assistant such as KeYmaera X~\cite{DBLP:conf/cade/FultonMQVP15}. Obligations include the following:
\begin{itemize}
  \item The proposed invariant must be preserved when running the controller and the plant in sequence, assuming that all assumptions and bounds hold initially.
  \item The proposed invariant and the global bounds must imply the postcondition.
  \item The nondeterministic controller must be \emph{total}, in the sense that at least one action should be available in any state where the invariant and bounds are true.
  \item All inference assignments must be sound.
\end{itemize}
In the particular case where no functional unknowns are used and all differential equations admit an analytical solution expressible as a multivariate polynomial, all generated proof obligations are decidable. In general, manual proof assistance may be required for certain obligations. Notably, $\dL$ has proved effective at reasoning about \emph{unsolvable} differential equations~\cite{platzer2018differential,kabra2022verified}.

The generated shield consists of three components. The first one is an \emph{inference module} that is extracted from the expert-defined, nondeterministic inference strategy. The other two are a \emph{controller monitor} and a \emph{fallback policy}, both of which are extracted from the nondeterministic $\dL$ controller. At runtime, an untrusted agent interacts with the physical world under the protection of the shield, which overrides any proposed action that is not validated by the controller monitor using its fallback policy. The inference module is tasked to instantiate the model's bound parameters, which the controller monitor and the fallback policy depend on. The inference module is guided by the agent's inference policy, which provides hints on how the safety budget should be spent and how observations should be weighted when building aggregates. The safety budget is initialized with the system's tolerated probability of failure, which is typically a small value such as $10^{-6}$. In cases where not enough budget remains to honor the inference policy's recommendation, the inference module skips the associated inference assignment.

Importantly, the inference policy must \emph{not} depend on the value of the observations processed by the inference module and must only base its decisions on the \emph{availability} of these observations and on features of their associated states. Intuitively, this restriction prevents the agent from $p$-hacking its way to unsafety~\cite{significant}. Indeed, it would be unsound for the agent to make several measurements of the same quantity and then discard all measurements but the most favorable one. Another requirement enforced by the inference module is that the same observation cannot be reused across control cycles. Indeed, allowing the reuse of observations across control cycles enables an indirect form of cherry-picking since the next state after each control cycle may be influenced by the observations made during this cycle and thus carry some amount of information about it.

\section{Background}

\subsection{Differential Dynamic Logic}\label{sec:dL-semantics}

\newcommand{\dLTermCompDefs}{\sim\  \in \{= \BnfOpSep < \BnfOpSep \le \BnfOpSep \ge \BnfOpSep >\}}
\newcommand{\dLBinopDefs}{\AnyBinOp \in \{+ \BnfOpSep \times \BnfOpSep \min \BnfOpSep \max\}}

\begin{figure}
  \EquationsFigureSize
  \begin{align*}
    \text{Hybrid Program:}&&
    \alpha,\beta &\BnfDef
      x \dlassign e \BnfOr
      x \dlassign * \BnfOr
      ?P \BnfOr
      x'=f(x) \dldom Q \BnfOr
      \alpha \cup \beta \BnfOr
      \alpha \seqI \beta \BnfOr
    \alpha^* \\
    \text{Formula:}&&
    P, Q &\BnfDef
      \theta_1 \sim \theta_2 \BnfOr
      \lforall x P \BnfOr
      \lexists x P \BnfOr
      \lboxI{\alpha}{P} \BnfOr
      \ldiamondI{\alpha}{P} \BnfOr
      \!\lnot P\! \BnfOr
      P \lor Q \BnfOr
      P \land Q \BnfOr
      P \limplyI Q \\
    \text{Term:}&&
    \theta &\BnfDef
      r \BnfOr x \BnfOr f(\theta_1,\dots,\theta_n) \BnfOr |\theta| \BnfOr
      \theta_1 \AnyBinOp \theta_2
  \end{align*}
  \begin{gather*}
  \text{Number literal: \ } r \in \Reals \qquad
  \text{Variable: \ } x \in \VarsAll \qquad
  \text{Function symbol: \ } f \in \SymbsAll \\
  \text{Term comparison: } \dLTermCompDefs \qquad
  \text{Arithmetic operator: \ } \dLBinopDefs
  \end{gather*}
  \begin{gather*}
    \Funs \,\equiv\, {\textstyle \bigcup_{n \in \Nats}} \, (\Reals^n \to \Reals) \quad
    \InterpsAll \,\equiv\, \SymbsAll \to \Funs \quad
    \StatesAll \,\equiv\, \VarsAll \to \Reals \\
    \sem{\alpha} : \InterpsAll \to \Subsets{\StatesAll \times \StatesAll} \quad
    \sem{P} : \Subsets{\InterpsAll \times \StatesAll} \quad
    \sem{\theta} : \InterpsAll \times \StatesAll \to \Reals
  \end{gather*}
  \vspace{-0.4cm}
  \caption{Syntax and semantics of $\dL$. $\Subsets{\cdot}$ denotes the powerset operator.}\label{fig:dl-syntax-semantics}
  \Description[]{}
\end{figure}

Differential dynamic logic ($\dL$)~\cite{DBLP:journals/jar/Platzer08,DBLP:journals/jar/Platzer17} is designed specifically to verify hybrid systems with both discrete and continuous dynamics.
The syntax and semantics of $\dL$ is summarized in Figure~\ref{fig:dl-syntax-semantics}. A \emph{state} is defined as a mapping from variable names to real values. An \emph{interpretation} is defined as a mapping from function symbols to real functions of proper arity (possibly zero). The semantics of a \emph{hybrid program} $\alpha$ can be defined by induction on $\alpha$. For $I$ an interpretation, $\sem{\alpha} (I)$ denotes the set of all pairs of states $(s_1, s_2)$ such that $s_2$ is reachable from $s_1$ by executing $\alpha$. We provide an intuitive summary of such semantics in Table~\ref{tab:hps} but formal definitions can be found in the literature~\cite{DBLP:journals/jar/Platzer17}. The semantics $\sem{P}$ of a \emph{formula} $P$ is defined as the set of all $(I, s)$ pairs such that $P$ is true in state $s$ and under interpretation $I$. It follows naturally from first-order logic. In addition, $[\alpha]P$ is true if and only if $P$ is true after all runs of the hybrid program $\alpha$ and $\langle \alpha \rangle P$ is true if and only if there exists a run of the hybrid program $\alpha$ after which $P$ is true. Atomic formulas consist in comparisons of \emph{terms}. The semantics $\sem{\theta}$ of a term $\theta$ maps a pair of an interpretation and a state to a real value. A formula $P$ is said to be \emph{valid} (written $\Fvalid P$) if it is true in all states and interpretations.

\subsection{Notations for Valuations}

Given a set of variables $V$, a \emph{valuation} of $V$ is defined as a function $\rho : V \to \Reals$ that maps each variable in $V$ to a value. Moreover, a \emph{partial valuation} of $V$ is defined as a {partial} function $\rho': V \parto \Reals$ that attributes values to a subset $\dom \rho'$ of variables from $V$ (note the special arrow symbol for partial functions). We denote $\EmptyValuation$ an empty valuation. Moreover, if $v$ and $v'$ are valuations of two disjoint subsets of variables $V$ and $V'$, we write $v \CatValuation v'$ the valuation of $V \cup V'$ induced by $v$ and $v'$.

\subsection{Reinforcement Learning}

\newcommand{\RLState}{{S}}
\newcommand{\RLAction}{{A}}

Reinforcement learning (RL) is about learning to make decisions in an environment in such a way to maximize rewards. Environments are typically represented using Markov Decision Processes (MDPs)~\cite{sutton2018reinforcement}. An MDP is defined by a tuple $\tuple{\RLState, \RLAction, P, R, \gamma}$ where $\RLState$ is the state space, $\RLAction$ is the action space, and $P: \RLState \times \RLAction \to \Distrs{\RLState}$ is the transition function, which maps any state-action pair to a probability distribution over new states. $R: \RLState \times \RLAction \times \RLState \to \Reals$ is the reward function, which describes the reward associated with a specific state-action pair and a specific next state. Finally, $\gamma\in (0, 1)$ is the factor by which future rewards are discounted at each time step. The goal of RL is to find an optimal policy $\pi:\RLState \to \Distrs{\RLAction}$ for the agent to maximize the expected discounted sum of future rewards: $\mathbb{E}_\pi\left[\sum_{i=0}^{\infty}\gamma^iR(s_i, a_i)\right]$. In the most common \emph{model-free} RL setting, $P$ and $R$ are unknown and the agent accesses samples of these by directly interacting with the environment.

\subsection{Safe Reinforcement Learning via Shielding}
\label{sec:safe-rl-via-jsc}

This work builds on the Justified Speculative Control framework (JSC)~\cite{DBLP:conf/aaai/FultonP18}, which leverages models written in $\dL$ to shield reinforcement learning agents by only allowing provably-safe control actions. JSC crucially relies on the ability to extract monitors~\cite{DBLP:journals/fmsd/MitschP16} from nondeterministic $\dL$ controllers. To define those monitors~\cite{DBLP:journals/fmsd/MitschP16}, we need to bridge $\dL$ models with reinforcement learning environments. Thus, consider a $\dL$ model of the kind defined in Equation~\ref{eq:example-model-def}. Such a model induces a state space $\ProgState \equiv \FV(\JSCModel) \to \Reals$, where a state maps each free variable occuring in the model to a real value. In addition, the controller $\Ctrl$ induces an action space that can be defined inductively:
\begin{gather*}
    \CtrlActionSem{\alpha \cup \beta} \!\equiv\! \CtrlActionSem{\alpha} \uplus \CtrlActionSem{\beta}, \ \
    \CtrlActionSem{\alpha \seqI \beta} \!\equiv\! \CtrlActionSem{\alpha} \times \CtrlActionSem{\beta}, \ \
    \CtrlActionSem{x \dlassign *} \!\equiv\! \Reals, \ \
    \CtrlActionSem{x \dlassign e} \!\equiv\! \CtrlActionSem{?Q} \!\equiv\! \UnitSet.
\end{gather*}
In this definition, $X \uplus Y \equiv \{(\SumLeft, x) \!:\! x \!\in\! X\} \cup \{(\SumRight, y) \!:\! y \!\in\! Y\}$ denotes the \emph{disjoint union} of sets $X$ and $Y$ and $\UnitSet \equiv \{\Unit\}$ denotes a set with a single element. Intuitively, an action records a sequence of choices that uniquely characterizes a run of the controller.
For example, the discrete controller defined in Equation~\ref{eq:example-model-ctrl} has an action space $\CtrlActionSem{\Ctrl} = \UnitSet \uplus (\UnitSet \times \UnitSet)$, which is isomorphic to a 2-element set containing a \textit{braking} and an \textit{accelerating} action.
The continuous controller used in Figure~\ref{fig:overview-train-global} has an action space $\CtrlActionSem{\Ctrl} = \Reals \times \UnitSet$.
For any state $s \in \ProgState$ and action $a \in \CtrlActionSem{\Ctrl}$, we write $\CtrlExeSem{\Ctrl}(s, a)$ for the state that results from executing $\Ctrl$ while performing the sequence of choices encoded in $a$. A formal definition is provided in Appendix~\ref{ap:defining-action-ctrl}. From the structure of a nondeterministic controller, one can extract a \emph{controller monitor} that takes as an input a state and an action and returns a Boolean value indicating whether or not the action is permissible.

\begin{lemma}[Existence of a controller monitor and fallback policy] \label{lem:ctrl-monitor} Let $\Ctrl$ be a nondeterministic controller in the form of a $\dL$ hybrid program that is free of loops, modalities, differential equations, quantifiers and function symbols of nonzero arity. Then, there exists a \emph{computable} function $\CtrlMonitorSem{\Ctrl} : (\FV(\Ctrl) \toI \Reals) \times \CtrlActionSem{\Ctrl} \to \{\True, \False\}$ that we call \emph{controller monitor} and such that for all state $s$ and action $a$, $\CtrlMonitorSem{\Ctrl}(s, a) = \True$ if and only if $(s, \,\CtrlExeSem{\Ctrl}(s, a)) \in \sem{\Ctrl}(\EmptyValuation)$. In addition, given any formula $\Inv$ such that $\Fvalid \Inv \limplyI \ldiamondI{\Ctrl}{\,\True}$ (i.e. there exists an outgoing controller transition for all states in $\Inv$), there exists a computable function $\CtrlFallbackSem{\Ctrl} : (\FV(\Ctrl) \toI \Reals) \to \CtrlActionSem{\Ctrl}$ that we call \emph{fallback policy} and such that for all state $s$ such that $(\EmptyValuation, s) \in \sem{\Inv}$, $(s, \,\CtrlExeSem{\Ctrl}(s, \,\CtrlFallbackSem{\Ctrl}(s)) \in \sem{\Ctrl}(\EmptyValuation)$.
\end{lemma}

In our train example from Figure~\ref{fig:overview-train-global} (with a \emph{continuous} control input), the action space $\CtrlActionSem{\Ctrl}$ is isomorphic to $\mathbb{R}$. An action is a commanded acceleration $u$ and the runtime monitor $\CtrlMonitorSem{\Ctrl}$ simply evaluates the two tests featured in the controller definition for the selected value of $u$: it checks that both formula $(-B \!\le\! u \!\le\! A)$ and formula \( (x + vT + (\TrainThetaAcc)T^2/2 + (v + (\TrainThetaAcc)T)^2/2(\TrainThetaBrakingRate) \le e) \) evaluate to \emph{true} in the current state and for the current values of bound parameters $\TrainThetaSlopeMin$, $\TrainThetaSlopeMax$ and $\TrainThetaBiasMax$.
In our train example from Figure~\ref{fig:overview-train-local} (with a \emph{binary} control input), the action space is isomorphic to $\{\textsf{brake}, \textsf{accelerate}\}$. The braking action is always allowed. To determine whether or not to allow acceleration, the monitor evaluates the test from the associated branch of the nondeterministic controller from Figure~\ref{fig:overview-train-local}. In general, controller monitors can be derived from shield specifications via a simple syntactic transformation. We formalize this transformation and provide a proof of Lemma~\ref{lem:ctrl-monitor} in Appendix~\ref{ap:ctrl-monitor}. The KeYmaera X toolchain features a tool called ModelPlex~\cite{DBLP:journals/fmsd/MitschP16} for automatically deriving executable controller monitors from $\dL$ models.

The existence of a fallback policy is guaranteed by the assumption $\Fvalid \Inv \!\!\limply\!\! \ldiamondI{\Ctrl}{\,\True}$, which indicates that there exists \emph{at least} one way to resolve the nondeterministic choices in $\Ctrl$ in such a way to pass all tests (i.e. not trigger any monitor alert). The corresponding action can be computed systematically by solving a decidable SMT problem, hence the existence of a decidable fallback policy (Appendix~\ref{ap:generating-fallbacks}). However, for the sake of monitoring efficiency, tools like ModelPlex require the user to specify an \emph{explicit} fallback policy on a case-by-case basis. Doing so is never a problem in practice since proving $\Fvalid \Inv \!\!\limply\!\! \ldiamondI{\Ctrl}{\,\True}$ requires characterizing such a policy anyway. In particular, an explicit fallback policy can be automatically extracted from a constructive proof of this formula~\cite{bohrer2020constructive}. In our train examples, a possible fallback policy is to brake with maximal force, which corresponds to action $-B$ in the continuous case and action ``\textsf{brake}'' in the discrete case.

\section{Adaptive Shielding}\label{sec:adaptive-shielding}

\subsection{Shield Specifications}\label{sec:shield-spec}

We define a language to specify adaptive shields via parametric safety models. Models can assume a set of parametric bounds on a number of unknown quantities. The derived runtime controller is parametric in these same bounds, which are maintained and improved at runtime by a derived inference module.  Examples of \emph{shield specifications} are available in Figures~\ref{fig:overview-train-global}~and~\ref{fig:overview-train-local}. Formally, a \emph{shield specification} is defined by a tuple of the following elements:

\makeMathKeywordText
\begin{description}
  \item[$\Const$:] a set of \emph{constants}, which are $\dL$ symbols of arity zero whose value is known at runtime. %
  \item[$\Unknown$:] a set of $unknowns$, which are $\dL$ symbols of fixed arity, whose value is \emph{not} known at runtime and can be \emph{estimated} by the inference module using observations.
  \item[$\Assum$:] a set of \emph{assumptions}, which are $\dL$ formulas representing global assumptions about the constants and unknowns. Assumptions can feature symbols but no free variables. When clear from the context, we also write $\Assum$ for the \emph{conjunction} of all assumptions.
  \item[$\StateVar$:] a set of \emph{state variables}, which are $\dL$ variables used to represent model state. In our concrete syntax, the set of state variables is inferred as the set of all mentioned free variables that are not parameters, observation variables or noise variables (see below).
  \item[$\Param, \Dir, \Bound$:] a set of \emph{parametric bounds}. $\Param$ is a set of $\dL$ variables called \emph{bound parameters} (or \emph{parameters} for short). Each parameter $p$ is mapped to a $\dL$ formula $\Bound_p$ that is either \emph{monotone} or \emph{anti-monotone} with respect to $p$, depending on \emph{bound direction} $\Dir_p \in \{ \UpperAnnot, \LowerAnnot \}$. In our concrete syntax, the bound direction can be omitted whenever easily inferrable. For $\theta$ a $\dL$ term, we write $\BoundWith{p}{\theta}$ for the result of substituting $\theta$ for $p$ in $\Bound_p$. If $\Dir_p = \UpperAnnot$, we call $\Bound_p$ an \emph{upper-bound} and we mandate that it is monotone, meaning that $\Fvalid \forall x y \; (x  \le y \limplyI \BoundWith{p}{x} \limplyI \BoundWith{p}{y})$. If $\Dir_p = \LowerAnnot$, we call $\Bound_p$ a \emph{lower-bound} and we mandate that it is anti-monotone: $\Fvalid \forall x y \; (x \le y \limplyI \BoundWith{p}{y} \limplyI \BoundWith{p}{x})$. $\Bound_p$ can feature constants, unknowns and state variables. It cannot feature any parameter other than $p$. It is said to be \emph{local} if it contains a free occurence of a state variable. Otherwise, it is said to be \emph{global}. The associated parameter $p$ is also called \emph{local} or \emph{global} accordingly. We write $\LocalParam$ the set of \emph{local} parameters and $\GlobalParam$ the set of \emph{global} parameters. We write $\GlobalBound \equiv \bigwedge_{p \in \GlobalParam} \Bound_p$ for the conjunction of {all} global bounds. Finally, when clear from the context, we use the notational shortcut $\Bound \equiv \bigwedge_{p \in \Param} \Bound_p$ for the conjunction of all bounds. For $b \in \Param \parto \Reals$, we also write $\BoundConj{b} \equiv \bigwedge_{p \in \dom b} \BoundWith{p}{b(p)}$. %
  \item[$\Ctrl$:] a \emph{controller}, in the form of a (nondeterministic) $\dL$ hybrid program that is free of loops, modalities, differential equations and quantifiers. The controller can involve constants, parameters and state variables but \emph{no unknowns}. Only state variables can be modified. For $c \in \Const \toI \Reals$ and $b \in \Param \toI \Reals$, we write $\CtrlWith{c, b}$ for the result of substituting all constants and parameters in $\Ctrl$ by their value according to $c$ and $b$. Thus, $\CtrlWith{c, b}$ has no symbols and its free variables are all in $\StateVar$. For any $(c, b)$ pair, Lemma~\ref{lem:ctrl-monitor} guarantees that a controller monitor and a fallback policy can be extracted from $\CtrlWith{c, b}$.
  \item[$\Plant$:] a $\dL$ hybrid program that represents the \emph{physical environment}. As opposed to the controller, it can feature unknowns but no parameters. Only state variables can be modified.
  \item[$\Safe$:] a $\dL$ formula that represents the \emph{safety constraint} under which the system must operate. It can involve constants, unknowns and state variables but no parameters, since the safety condition cannot change over time as knowledge is gathered.
  \item[$\Inv$:] an \emph{invariant}, in the form of a $\dL$ formula. The invariant can feature constants, state variables and unknowns but only \emph{global} parameters. Local parameters are not allowed in the invariant. The following proof obligations are generated that involve the invariant:%
    \begin{enumerate}
      \item \label{obl:inv-implies-post} The invariant must imply the postcondition:
        $\Fvalid \Assum \wedge \GlobalBound \wedge \Inv \limplyI \Safe$.
      \item \label{obl:inv-preserved} The invariant must be preserved:
          $\Fvalid \Assum \wedge \Bound \wedge \Inv \limplyI \lbox{\Ctrl \seqI \Plant}{\Inv}$.
      \item \label{obl:ctrl-total} There always exists a safe controller action: $\Fvalid \Assum \wedge \Bound \wedge \Inv \limplyI \ldiamond{\Ctrl}{\True}$.
      \item \label{obl:inv-monotone} The invariant must be \emph{monotone} with respect to upper-bound parameters and \emph{anti-monotone} with respect to lower-bound parameters.
    \end{enumerate}
  \item[$\NoiseVar, \Noise$:] a set of \emph{noise variable declarations}. $\NoiseVar$ is a set of $\dL$ variables that can be used in the definition of observations and that model samples from \emph{mutually independent} random variables. Each  {noise variable} $\eta \in \NoiseVar$ is mapped to an expression $\Noise_\eta$ that denotes a probability distribution. $\NoiseVar_\eta$ can be either $\slnormal(\theta_1, \theta_2)$ for a normal distribution, $\sluniform(\theta_1, \theta_2)$ for a uniform distribution or $\slbernouilli(\theta_1)$ for a Bernouilli distribution. In all cases, $\theta_1$ and $\theta_2$ are $\dL$ terms that can feature constants and state variables. We write $\sem{\Noise} : (\Const \toI \Reals) \to \NoiseVar \to \Distr(\Reals)$ for the semantic denotation of $\Noise$.
  \item[$\ObsVar, \Obs$:] a set of \emph{observations}. $\ObsVar$ is a set of $\dL$ variables called \emph{observation variables}. Each observation variable $\omega$ is mapped to a defining $\dL$ term $\Obs_\omega$ that may feature constants, unknowns, state variables and noise variables (but no bound parameter).
 \item[$\Inference$:] an \emph{inference strategy}, which defines an inference module that issues concrete values for bound parameters at runtime, using an external \emph{inference policy} for guidance.
\end{description}
\makeMathKeywordMath
The inference strategy can be defined in a custom, domain-specific language that is explained in details in Section~\ref{sec:inference-strategy-language}. Before we dive into this language though, we characterize the abstract requirements that an inference module must fulfill in Section~\ref{sec:symbolic-inference-bounds}. The soundness of our proposed inference strategy language can then be established against these requirements.

\subsection{Symbolic Bound Instantiations}\label{sec:symbolic-inference-bounds}

An inference strategy $\Inference$ induces an action space $\InfActionSem{\Inference}$ for the agent's inference policy, which is defined rigorously in Section~\ref{sec:inference-strategy-language}. In particular, inference actions associated with a single $\textsc{aggregate}$ assignment are $(\Eps, \lambda)$ pairs, as seen in Section~\ref{sec:overview-inference}. An inference strategy denotes a function $\sem{\Inference} : \InfActionSem{\Inference} \to \List{\Param \times \SymbBoundExpr \times \EpsType}$ that maps an inference action to a list of {symbolic inference assignments}. A \emph{symbolic inference assignment} consists of a triple $(p, e, \Eps)$ of a bound parameter $p$, a \emph{symbolic bound instantiation} $e$ and a failure probability $\Eps$.

To illustrate the concept of a \emph{symbolic bound instantiation}, consider the following example statement from Section~\ref{sec:overview-inference}: ``$\TrainFxMax \slassign \slaggregate i: \omega_i + k|x - x_i| \sland \eta_i$''. Provided a $(\Eps, \lambda)$ pair and a history of states and observations, the right-hand-side of this statement can be evaluated into a concrete, numerical value proposal for $\TrainFxMax$. This is done in two stages, the first one of which being the computation of a \emph{symbolic bound instantiation} $e$ for $\TrainFxMax$. For example, if $\Eps=10^{-8}$, $\lambda_2 = 0.3$, $\lambda_5 = 0.7$ and $\lambda_i = 0$ for all $i \notin \{2, 4\}$, we obtain the following symbolic bound instantiation:
\begin{equation}\label{eq:symb-inst-example}
  e \,\equiv\, 0.3 \times (\omega_2 + k|x-x_2|) + 0.7 \times (\omega_5 + k|x-x_5|) + \InvCCDFwith{\eta_2,\eta_5 \sim \slnormal(0, \,\sigma^2)}(0.3 \eta_2+0.7\eta_5, 10^{-8}).
\end{equation}
In the expression above, the $\InvCCDF$ construct provides a symbolic representation of the \emph{inverse tail function} of a probability distribution, also called its \emph{inverse complementary cumulative distribution function} -- hence the notation. Given a random expression and a tolerance $\Eps$, $\InvCCDF$ returns the best upper-bound obtainable on this expression that is true with probability $1-\Eps$ at least. For example, we have $\InvCCDFwith{X \sim \sluniform(0, 1)}(X + 1, \Eps) = 2 - \Eps$ for all $\Eps \in [0, 1]$ since $\Prob_{X \sim \sluniform(0, 1)}\{X + 1 > 2 - \Eps\} = \Eps$. Details on how to compute or approximate the $\InvCCDF$ operator numerically are available in Appendix~\ref{sec:eval-tail-expressions}.

Provided some values for the constant symbols, for the current state and for past states and observations from the agent's history, the symbolic bound instantiation from Equation~\ref{eq:symb-inst-example} can be in turn evaluated into a real number. There are two reasons for performing such staging. First, it makes it very clear that $\sem{\Inference}$ has no access to the actual observation values when building $e$, which is critical for soundness (as previously discussed in Section~\ref{sec:overview-inference}). Second, this pushes the orthogonal challenge of evaluating $\InvCCDF$ outside of the core inference machinery.

\subsubsection{Formal syntax and semantics}

A sketch of a formal syntax and semantics for symbolic bound instantiations is provided in Figure~\ref{fig:tail-expressions}. A symbolic bound instantiation $e$ can be a $\dL$ term $\theta$, the sum of two other symbolic bound instantiations, an application of the $\InvCCDF$ operator or a \emph{guarded expression} $(e \WhenSymbBE P)$, where $P$ is a quantifier-free, modality-free $\dL$ formula. Importantly, the evaluation of a symbolic bound instantiation can fail and return a special value $\OptNone$. It does so when attempting to access the value of an undefined variable or symbol, or when evaluating $(e \WhenSymbBE P)$ while $P$ evaluates to $\False$. Support for $\OptNone$ is important since our framework does not mandate observations to be available at each cycle, making some observation variables possibly unassigned. In addition, ``$\WhenSymbBE$'' is necessary for our inference strategy language to support conditionally sound bounds, which we use in the example from Figure~\ref{fig:overview-train-global} and discuss further in Section~\ref{sec:inference-strategy-language}.

\subsubsection{Indexed variables}

Finally, symbolic bound instantiations can contain special $\dL$ variables, whose name contains an integer index (i.e. $\omega_2$). Semantically, there is nothing special about such variables and we can just regard them as normal variables following a naming convention. The only thing that matters is that for any $x \in V$ and $i \in \Nats$, $x_i \notin V$ where $V \equiv \StateVar \cup \Param \cup \ObsVar \cup \NoiseVar$. We also introduce some notations to ease the manipulation of indexed variables. For $v$ a partial valuation over $V$ and $i \in \Nats$, we write $\TagWithIndex{v}{i} \equiv \{x_i \mapsto r : v(x) = r \}$ the same valuation in which all domain variables are tagged with index $i$. For example, if $x, y \in V$ and $v = \{x \mapsto 0, y \mapsto 1\}$, then $\TagWithIndex{v}{2} = \{x_2 \mapsto 0, y_2 \mapsto 1\}$. In addition, for $P$ a formula with variables in $V$, we write $\TagWithIndex{P}{i}$ the formula that results from $P$ after tagging every free variable with index $i$. Finally, we also allow the use of symbolic names as indices. For example, in the inference assignment ``$\slaggregate i: \omega_i + k|x - x_i| \sland \eta_i$'', $\omega_i$ is a normal $\dL$ variable whose name follows a particular convention. We introduce a special syntactic transformation for ``instantiating'' symbolic indices in formulas. For example, if $P \equiv x_i + x_j > 0$, we write $\FSubst{P}{i,j}{2,3} \equiv x_2 + x_3 > 0$.

\begin{figure}
  \EquationsFigureSize
  \begin{gather*}
    \delta \BnfDef \slnormal(\theta_1, \theta_2) \BnfOr \sluniform(\theta_1, \theta_2) \BnfOr \slbernouilli(\theta) \qquad
    e \BnfDef \theta \BnfOr e_1 + e_2 \BnfOr e \WhenSymbBE P  \BnfOr \InvCCDFwith{\eta_1 \sim \delta_1 \cdots \eta_n \sim \delta_n}(\theta_1, \theta_2) %
  \end{gather*}
  \begin{align*}
    \SymbBESem{e} &\aligncolon (\SymbsAll \to \Funs) \times (\VarsAll \parto \Reals) \to \OptReals \\
    \SymbBESem{\InvCCDFwith{\eta_1 \sim \delta_1 \cdots \eta_n \sim \delta_n}(\theta_1, \theta_2)}(I, v) &\in
       \Bigl\{\,b : \Prob_{\eta_i \sim \sem{\delta_i}(I, v)}\bigl(\sem{\theta_1}(I, v \!\CatValuation\! (\CatValuation_i \eta_i)) > b \bigr) \le \sem{\theta_2}(I,v)\,\Bigr\} \cup \{\OptNone\}
  \end{align*}
  \vspace{-0.4cm}
  \caption{Syntax and semantics of symbolic bound instantiations. We only show the contract that $\InvCCDF$ must fulfill for soundness. Concrete implementations must optimize for small, real return values.  }\label{fig:tail-expressions}
  \Description[]{}
\end{figure}

\subsection{Soundness of Symbolic Inference Assignments}

As mentioned in the previous section, an inference strategy denotes a function $\sem{\Inference} : \InfActionSem{\Inference} \to \List{\Param \times \SymbBoundExpr \times \EpsType}$ that maps an inference action to a list of \emph{symbolic inference assignments}. Our proposed inference strategy language is \emph{sound} in the sense that any strategy expressed with it produces \emph{sound} assignments by construction. In turn, a symbolic inference assignment $(p, e, \Eps)$ is \emph{sound} if $e$ evaluates to a ``correct'' instantiation of $p$ with probability at least $1-\Eps$. To make this definition more precise, we need a couple of preliminary definitions:
\begin{itemize}
  \item A \emph{state} is defined as a valuation of the state variables: $\ProgState \equiv \StateVar \to \Reals$.
  \item A \emph{bound instantiation} is a partial valuation of the parameters: $\Inst \equiv \Param \parto \Reals$.
  \item An \emph{observation availability set} is a subset of observation variables indicating what observation were made at one point in time: $\ObsAvail \equiv \Subsets{\ObsVar}$.
\end{itemize}
We can now define the notion of \emph{soundness} for symbolic inference assignments. Intuitively, $(p, e, \Eps)$ is \emph{sound} if and only if for any possible history of states, observations and bound instantiations and for any intepretation of the model unknowns, updating an instantiation $b$ by assigning the value of $e$ to $p$ preserves the truth of $\BoundConj{b}$ with probability at least $1-\Eps$. The challenge in formalizing this intuition is to precisely characterize what counts as a \emph{possible} history. We do so in Definition~\ref{def:sound-inf-assignment}, building a valuation $v$ that stands for an arbitrary, consistent history.
\begin{definition}[Sound inference assignment]\label{def:sound-inf-assignment}
  A symbolic inference assignment $(p, \BSymb, \Eps)$ is said to be \emph{sound} if and only if for any $c \in \Const \toI \Reals$, $u \in \Unknown \toI \Funs$, $n \in \Nats$, $(s_1 \dots s_{n}) \in \ProgState^{n}$, $(z_1 \dots z_{n}) \in \ObsAvail^{n}$ and $(b_1 \dots b_{n}) \in \Inst^{n}$ such that $(c \CatValuation u, s_i) \in \sem{\Assum \land \Inv \land \BoundConj{b_i}}$ for all $i \le n$, we have:
  \[\Prob\left\{ (c \CatValuation u, s_{n}) \notin \sem{\BoundConj{\UpdateMap{b_n}{p}{r}}} \right\} \le \Eps \]
   where
    $r = \sem{e}(c, v)$,
    $v = \bigl(s_n \CatValuation b_n \bigr) \cup \bigl(\BigCatValuation_{1 \le i \le n} \TagWithIndex{s_i}{i} \CatValuation \TagWithIndex{b_i}{i} \CatValuation \TagWithIndex{o_i}{i}\bigr)$,
    $o_i = \{\omega \mapsto \sem{\Obs_\omega}(c \CatValuation u, s_i \CatValuation \eta_i) : \omega \in \Param \cap z_i \}$ for all $1 \le i \le n$,
    and $\eta_1 \cdots \eta_{n} \sim \sem{\Noise}(c)$ i.i.d.
\end{definition}

\subsection{Compatible Environments}

A shield acts as a buffer between an untrusted agent and a \emph{compatible} environment. By \emph{compatible}, we mean that the environment must conform to the assumptions made by the shield specification, without which no safety guarantee can be provided. The following definition formalizes this notion. %

\begin{definition}[Compatible environments]\label{def:compat-env}
  Consider a shield specification $\tuple{\Const, \dots, \Inference}$. A \emph{compatible environment} is a tuple $\tuple{\tuple{S, A, P, R, \gamma}, \StateMapping, \ActionMapping, \RevActionMapping, c, u, \ObsAvailFun, \ObsMeasurementFun}$ where:
  \begin{itemize}
    \item $\tuple{S, A, P, R, \gamma}$ is a Markov Decision Process.
    \item $\StateMapping: S \to \ProgState$ is a \emph{state mapping} function that maps environment states to shield states.
    \item $\ActionMapping: A \to \CtrlActionSem{\Ctrl}$ is a surjective \emph{action mapping}.
    \item $\RevActionMapping : \CtrlActionSem{\Ctrl} \to A$ is a reverse action mapping such that $\ActionMapping \circ \RevActionMapping = \IdFun$.
    \item $c : \Const \to \Reals$ is an interpretation of the shield's constant symbols.
    \item $u : \Unknown \to \Funs$ is an interpretation of the shield's unknown symbols.
    \item $\alpha : S \to \Subsets{\ObsVar}$ is an \emph{observation availability function}.
    \item $\mu : S \to \Distrs{\ObsVar \to \Reals}$ is an \emph{observation measurement function}.
  \end{itemize}
  In addition, the following should hold:
  \begin{enumerate}
    \item \label{item:correct-model} The shield's plant correctly models the environment's transition function: for all $s, s' \in S$, $a \in A$ and $b \in \Inst$, assuming that
    \begin{inparaenum}[\it (1)]
      \item $s' \in \Support{P(s, a)}$ (i.e $(s, a, s')$ is a valid transition),
      \item $(c \CatValuation u, \StateMapping(s) \CatValuation b) \in \sem{\Assum \land \Bound \land \Inv}$, and
      \item $(\StateMapping(s) \CatValuation b,  s'' \CatValuation b) \in \sem{\Ctrl}(c)$ where $s'' \equiv \CtrlExeSem{\CtrlWith{c, b}}(\StateMapping(s), \ActionMapping(a))$,
    \end{inparaenum} then we have $(s'', \StateMapping(s')) \in \sem{\Plant}(c \CatValuation u)$.
    \item \label{item:correct-observation} \label{} The observation measurement function behaves consistently with $\Noise$ and $\Obs$: for all $s \in S$ and $o \in \ObsVar \toI \Reals$, then $\Prob_{\eta \sim \sem{\Noise}(c)}\{ \sem{\Obs}(c \CatValuation u, \StateMapping(s) \CatValuation \eta) = o\} = \ObsMeasurementFun(s)(o)$.
  \end{enumerate}
\end{definition}

A compatible environment extends a standard RL environment with a formal mapping between this environment and a shield specification's underlying model. It also defines some ground-truth values for all unknown symbols, along with two observation functions for acquiring knowledge about those. We choose to have two separate functions, where $\alpha$ indicates what observations are available in any given state and $\mu$ performs an actual measurement of those observations. Such a separation is convenient since the inference policy is allowed to access information about observation availability but \emph{cannot} access actual observation values.

\subsection{Shielded Environments and Main Safety Theorem}\label{sec:shielded-env}

We define our main shielding algorithm by describing how a shield acts as a \emph{wrapper} around compatible environments, resulting into \emph{shielded environments} that are amenable to RL themselves and in which no unsafe state is reachable by construction.

\begin{algorithm}
  \caption{Sampling a transition in a \emph{shielded MDP}}\label{alg:shield}
  \begin{algorithmic}[1]
    \Procedure{shielded\_transition\,}{$\hat s, \hat a$}
      \State $s, h, \GlobBounds, \EpsRem \gets \hat s$ \Comment Decompose $\hat s \in \ShieldedState$
      \State $\CtrlAc, \InfAc \gets \hat a$ \Comment Decompose $\hat a \in \ShieldedAction$
      \State $\InfResult \gets \sem{\Inference}(\InfAc)$ \Comment Get list of symbolic bound assignments $\InfResult$ \label{line:inf-start} \label{line:algo-start}
      \State $\Omega = \{ \omega_i \in \FV(\BSymb) : (p, \BSymb, \Eps) \in \InfResult, \, \omega \in \ObsVar \}$ \Comment Determine observations to make
      \State $n \gets |h| + 1$ \Comment Get index of history element being built
      \State $v \gets \StateMapping(s) \CatValuation \TagWithIndex{\StateMapping(s)}{n} \CatValuation \GlobBounds \CatValuation \TagWithIndex{\GlobBounds}{n}$ \Comment Initialize a valuation for $\InfResult$
      \For{ $i \in \{i : \omega_i \in \Omega \}$ } \Comment For every relevant step $i$ in history
        \State $s', z', \LocBounds' \gets h_i$ %
        \State $o \gets \AlgSample{\mu(s')}$ \Comment Measure observations for step $i$ \label{line:lazy-measurement}
        \State $v \gets v \CatValuation \{ \omega_i \mapsto o(\omega) : \omega_i \in \Omega, \, \omega \in z' \}$ \Comment Update $v$ with \underline{accessible} measurements
        \State $v \gets v \CatValuation \TagWithIndex{\StateMapping(s')}{i} \CatValuation \TagWithIndex{\LocBounds'}{i}$ \Comment Update $v$ with history element $i$
        \State $h \gets \UpdateMap{h}{i}{(s', \EmptyValuation, \, \LocBounds')}$ \Comment Ensure observations are not reused
      \EndFor
      \For{$(p, e, \Eps) \in \InfResult$} \Comment For every symbolic bound assignment in $\InfResult$ \label{line:start-exec-inference}
        \If{$\Eps > \EpsRem$} \Comment{If not enough budget remains}
          \State \textbf{continue} \Comment{We skip the assignment}
        \EndIf
        \State $\EpsRem \gets \EpsRem - \Eps$ \Comment Update the remaining safety budget
        \State $r \gets \sem{e}(c, v)$ \Comment{Evaluate symbolic instantiation with $v$}
        \If{$r \ne \OptNone$ \textbf{and} ($p \notin \dom v$ \textbf{or} $r < v(p)$)} \Comment{Whenever a tighter value is obtained}
          \State $v \gets \UpdateMap{\UpdateMap{v}{p}{r}}{\TagWithIndex{p}{n}}{r}$ \Comment{Update $v$ with the new tighter bound}
        \EndIf
      \EndFor
      \State $\LocBounds \gets \{p \mapsto v(p) : p \in \LocalParam\}$ \
      \State $\GlobBounds \gets \{p \mapsto v(p) : p \in \GlobalParam\}$
      \State $\AlgAssert{\dom \LocBounds = \LocalParam}$ \Comment Ensure that all local parameters are set \label{line:assert-local}
      \State $h \gets h \cdot (s, \, \ObsAvailFun(s), \, \LocBounds)$ \Comment Add current state to the inference history \label{line:inf-end}
      \State $\text{safe} \gets \CtrlMonitorSem{\CtrlWith{c, \,\GlobBounds \CatValuation \LocBounds}}$ \Comment Instantiate the control monitor \label{line:mon-start}
      \State $\text{fallback} \gets \CtrlFallbackSem{\CtrlWith{c, \,\GlobBounds \CatValuation \LocBounds}}$ \Comment Instantiate the fallback policy
      \If{\textbf{not }$\text{safe}(\StateMapping(s), \, \ActionMapping(\CtrlAc))$} \Comment When the control action is deemed unsafe
        \State $\CtrlAc \gets \RevActionMapping\bigl(\text{fallback}(\StateMapping(s))\bigr)$ \Comment{Override it with a fallback action} \label{line:mon-end} \label{line:fallback}
      \EndIf
      \State $s \gets \AlgSample{P(s, \CtrlAc)}$ \Comment Perform the action in the underlying MDP \label{line:perform-underlying}
      \State $\Return \  (s, h, \GlobBounds, \EpsRem)$ \Comment Return a new value of $\hat s$
    \EndProcedure
    \medskip
  \end{algorithmic}
\end{algorithm}

\begin{definition}[Shielded environment]\label{def:shielded-env}
Consider a shield specification $\tuple{\Const, \dots, \Inference}$ and a compatible environment $\tuple{\tuple{S, A, P, R, \gamma}, \dots, \ObsMeasurementFun}$ (see Definition~\ref{def:compat-env}). We define the \emph{induced shielded environment} as an MDP $\tuple{\hat S, \hat A, \hat P, \hat R, \gamma}$ where:
\begin{itemize}
  \item A \emph{state} $\hat s = (s, h, \GlobBounds, \EpsRem)$ consists of a state $s$ from the original environment, a history $h$, an instantiation of global parameters $\GlobBounds$ and a remaining safety budget $\EpsRem \ge 0$. In turn, a history is a list of $(s', z, \LocBounds)$ triples where $s'$ is a state, $z$ is an observation availability set and $\LocBounds$ is a local parameter instantiation. Formally, $\ShieldedState \equiv S \times \Hist \times (\GlobalParam \to \Reals) \times [0, 1]$ and $\Hist \equiv \List{S \times \Subsets{\ObsVar} \times (\LocalParam \to \Reals)}$.
  \item An \emph{action} $\hat a = (\CtrlAc, \InfAc)$ is defined as a pair of a \emph{control action} and of an \emph{inference action}. Formally, $\ShieldedAction \equiv \CtrlActionSem{\Ctrl} \times \InfActionSem{\Inference}$.
  \item The \emph{transition function} $P : \hat S \times \hat A \to \Distrs{\hat S}$ is defined as in Algorithm~\ref{alg:shield}.
  \item The \emph{reward function} is lifted from the original MDP: if $\hat s = (s, h, \GlobBounds, \EpsRem)$, $\hat a = (\CtrlAc, \InfAc)$ and $\hat s' = (s', h', \GlobBounds', \EpsRem')$, then $\hat R(\hat s, \hat a, \hat s') = R(s, \CtrlAc, s')$.
\end{itemize}
\end{definition}

Note that in our definition, a state of the shielded MDP does not contain any information about past observation values. Keeping such knowledge from the agent is crucial for soundness since the agent could otherwise engage in cherry-picking. Instead, the history component of a state only indicates which past observations are available and not used already.

Algorithm~\ref{alg:shield} rigorously defines the process of sampling a transition from a shielded environment, indirectly defining the behavior of our proposed adaptive shield (see Figure~\ref{fig:overview-diagram} for a synthetic view). Lines~\ref{line:inf-start}~to~\ref{line:inf-end} describe the execution of the inference module, which computes valuations $\GlobBounds$ and $\LocBounds$ for global and local bound parameters respectively while updating the history $h$ and the remaining safety budget $\EpsRem$. Lines~\ref{line:mon-start}~to~\ref{line:mon-end} instantiate the controller monitor and fallback policy using $\GlobBounds$ and $\LocBounds$. If the action proposed by the agent is rejected by the controller monitor, it is overriden via the fallback policy. Finally, line~\ref{line:perform-underlying}, executes the resulting action in the real environment.

Note that we rely on a formalization trick where Algorithm~\ref{alg:shield} performs observation measurements lazily (Line~\ref{line:lazy-measurement}), sometimes on past states, rather than performing them eagerly and storing them. Although physically unrealistic, this formalization enables us to stay within the standard MDP framework that is dominant in RL. Concrete implementations would of course perform observations eagerly and ``secretely'' cache them, offering the same interface. An alternative choice would have been to define our shielding algorithm using framework of Partially-Observed Markov Decision Processes (POMDPs)~\cite{kochenderfer2015decision}. However, doing so would complicate our proofs and definitions for little gain. Finally, instantiating the controller monitor (Line~\ref{line:mon-start}) requires \emph{every} global and local parameter to be mapped to a real value. This trivially holds for global parameters, by definition of a shielded environment state (Definition~\ref{def:shielded-env}). However, this does not obviously hold for \emph{local} parameters, hence the assertion on Line~\ref{line:assert-local}. It is the inference strategy itself that must guarantee the validity of this assertion, which it does by statically enforcing that every local parameter is provided a \emph{default} value that is independent from observations or other local parameters. In the example from Figure~\ref{fig:overview-train-local}, the first assignment $\TrainFxMax \slassign F$ serves this role. We can now state our main safety theorem, for which a proof sketch is provided in Appendix~\ref{ap:main-theorem-proof}.

\begin{theorem}[Safety of shielded environments]\label{thm:main-safety}
  Consider a shield specification $\tuple{\Const, \dots}$ and a compatible environment $\tuple{\tuple{S, \dots}, \dots, \ObsMeasurementFun}$ (see Definition~\ref{def:compat-env}). Let $\hat E$ be the resulting shielded MDP (see Definition~\ref{def:shielded-env}). Then, given a safety budget $\Eps \in [0, 1]$, an initial global bound instantiation $b \in \GlobalParam \toI \Reals$ and an initial state $s \in S$ such that $(c \CatValuation u, \StateMapping(s) \CatValuation b) \in \sem{\Assum \land \Bound \land \Inv}$, any trace in $\hat E$ starting with state $\hat s_0 = (s, \EmptyList, b, \Eps)$ is \emph{safe} with probability at least $1 - \Eps$, meaning that $(c \CatValuation u, \StateMapping(s_t)) \in \sem{\Safe}$ for all state $\hat s_t = (s_t, \,\dots)$ in the trace.
\end{theorem}

\subsection{Learning in a Shielded Environment}\label{sec:learning-in-shielded-env}

Adaptive shields such as defined in Section~\ref{sec:shielded-env} can be used to protect \emph{any} agent from acting unsafely in its environment, whether or not it is capable of learning and regardless of how it is implemented. This is emphasized by our choice of formalizing shields as \emph{environment wrappers}, which are agent-agnostic by construction. Still, for the sake of concreteness, it is useful to illustrate the use of Algorithm~\ref{alg:shield} and Theorem~\ref{thm:main-safety} in two typical reinforcement-learning scenarios:

\begin{description}
  \item[Learning in a fixed environment.] \label{learning:fixed} A shielded agent is placed in a fixed, \emph{compatible} environment. It ignores the value of \emph{unknowns} but those are kept constant across training. A total safety budget is defined for the inference module, which must be small enough to make crashes unlikely (as per Theorem~\ref{thm:main-safety}). Initially, the parameter bounds provided by the inference module are very loose and so the shield's action monitor is equally conservative. However, at each control cycle, the inference module spends some safety budget to improve those bounds, thereby allowing the agent to take increasingly aggressive exploratory moves. When the safety budget is fully spent, the inference module is not allowed to update parameter bounds anymore. However, the agent is still guaranteed to remain safe indefinitely, and in particular as it finishes to optimize its control policy. Section~\ref{sec:experiments} demonstrates this setting by showcasing a train that learns to navigate on a \emph{fixed} circuit.
  
  \item[Meta-learning across a family of environments.] \label{learning:meta} In many practical scenarios, autonomous cyber-physical systems must be capable of \emph{quickly} adapting to changing environments without training a dedicated policy from scratch. It is thus useful to train an agent across a \emph{family} of environments so that it can generalize across them. Such a \emph{meta-learning} setting integrates particularly well with our shielding framework. Indeed, we can train a shielded agent through a series of episodes, each of which takes place in a possibly different environment that is nonetheless compatible with the agent's shield (for possibly different values of the unknowns). At the start of each episode, the inference module is reinitialized and a fixed, \emph{per-episode} safety budget is granted. In addition to the state information, the agent is given access to the current bound parameters and the remaining safety budget at all times. It must learn to gather suitable information about its environment and adapt its actions accordingly \emph{within each episode}. In particular, this allows a form of \emph{active sensing}, where the agent learns to act in such a way to receive useful information from the inference module. This also allows optimizing the behavior of the inference module itself by learning \emph{inference policies}, as we discuss in Section~\ref{sec:inference-strategy-language}. In the meta-learning setting, Theorem~\ref{thm:main-safety} is applied at \emph{each} episode and so a global safety bound can be obtained by multiplying the \emph{per-episode} safety budget by the number of training episodes. Section~\ref{sec:experiments} explores this setting in three different case studies, including one with a train controller that must learn to quickly adapt to different circuits with different slope profiles.
\end{description}

\section{The Inference Strategy Language}\label{sec:inference-strategy-language}

We define the syntax and semantics of our \emph{inference strategy language} in Figure~\ref{fig:inf-lang-syntax-semantics}. An inference strategy is defined as a sequence of inference assignments. Each kind of assignment defines its own action space and the action space for a whole strategy is the product of the action spaces of individual assignments.  Three kinds of assignments are supported.
Any assignment can be attached a \textsc{when}-clause that restricts its applicability, as illustrated in Figure~\ref{fig:overview-train-global} and Appendix~\ref{ap:train-global-inf-strategy}.
A direct assignment ``$p \slassign \theta$\,'' does not require any input from the inference policy and executing it costs no budget. It yields a single symbolic bound instantiation, which is $\theta$ itself. An assignment of the form ``$p \slassign \slbest i_1, \cdots, i_n : \theta$\,'' expects as an input a list of $n$-tuples of history indices and yields one bound instantiation $\FSubst{\theta}{i_1, \cdots, i_n\,}{\,j_1, \cdots, j_n}$ for each tuple $(j_1, \cdots, j_n)$. An alternate design would expect no input and consider \emph{every} possible combination of indices through the current history. However, the number of such combinations may grow intractably large, which explains our current pragmatic choice. %

\paragraph{Semantics of \textsc{aggregate}}
An assignment of the form ``$p \slassign \slaggregate i_1, \cdots, i_n : \theta_1 \sland \theta_2$'' takes as an input an $(\Eps, D)$ pair, where $\Eps$ indicates how much safety budget should be spent and $D$ is a finite distribution over all $n$-tuples of history indices. More precisely, $D$ is a sequence of $(\lambda, j)$ pairs where $j$ indexes a combination of measurements from history and $\lambda$ is a positive weight attached to it ($\sum_{\lambda, j \in D} \lambda = 1$). It yields a single bound instantiation that performs a weighted average over all instances of the \emph{observable bound component} $\theta_1$, while using the $\InvCCDF$ operator to handle the \emph{noise component} $\theta_2$. As a requirement, $\theta_1$ can contain observation variables but no noise variables, while $\theta_2$ can contain noise variables but no observation variables. The semantics of \textsc{aggregate} generalizes our reasoning from Section~\ref{sec:overview-inference} (see Equation~\ref{eq:overview-train-local-aggr-proof}). Indeed, let us assume that
$\BoundWith{p}{\FSubst{(\theta_1 + \theta_2)}{i}{j}}$ is true for any instantiation $j \equiv (j_1, \dots, j_n)$ of $i \equiv (i_1, \dots, i_n)$, as guaranteed by the generated proof obligation. Also, let us assume without loss of generality that $p$ is an \emph{upper}-bound. Since $\Bound_p$ is monotone in $p$ and thus convex, we know that $\BoundWith{p}{b^*}$ is true where
$ b^* \equiv \sum_{\lambda, j \in D} \lambda \FSubst{(\theta_1 + \theta_2)}{i}{j}. $
However, $b^*$ is not known at runtime since $\theta_2$ features noise variables that are not observed directly. Instead, our semantics yields instantiation $b \equiv \sum_{\lambda,j \in D} \lambda  \FSubst{\theta_1}{i}{j} + \delta$, where $\delta \equiv \InvCCDF(\sum_{\lambda, j \in D} \lambda \FSubst{\theta_2}{i}{j},\,\Eps).$
Since $\Bound_p$ is monotone, we have $b \ge b^* \limplyI \BoundWith{p}{b}.$ Contraposing, we get $\neg \BoundWith{p}{b} \limplyI b^* > b$. We can then establish the soundness of $(p, b, \Eps)$ with the exact same reasoning that was demonstrated in Equation~\ref{eq:overview-train-local-aggr-proof}:
\begin{equation*}
  \Prob(\neg \BoundWith{p}{b}) \,\le\, \Prob(b^* > b)
    \,=\, \Prob \Bigl( \,\sum_{i,\lambda \in D} \lambda \cdot \FSubst{\theta_2}{i}{j} \,>\, \delta\, \Bigl)
    \,\le\, \Eps.
\end{equation*}
Details on how to evaluate or approximate the $\InvCCDF$ operator are available in Appendix~\ref{sec:eval-tail-expressions}. In the particular case where $\theta_2$ is a linear combinations of Gaussian variables (with possibly state-dependent coefficients), closed-form solutions exist that involve the $\ErfInv$ function. Concentration inequalities (e.g. Chebyshev, Hoeffding, or Chernoff bounds) can be used to approximate it conservatively in the general case~\cite{bertsekas2008introduction}.

\begin{figure}
  \EquationsFigureSize
  \begin{align*}
    \InferenceStrategy \BnfDef\ & p_1 \dlassign e_1 \seq \dots \seq p_n \dlassign e_n \qquad
    (p_1, \dots, p_n \slassign e) \ \equiv\  (p_1 \slassign e \seq \dots \seq p_n \slassign e)
    \\
    e \BnfDef\ &
    \theta \slwhen G \BnfOr
    \slbest i_1, \dots, i_n : \theta \slwhen G \BnfOr
    \slaggregate i_1, \dots, i_n: \theta_1 \sland \theta_2 \slwhen G
  \end{align*}

  \begin{align*}
    \InfActionSem{p_1 \dlassign e_1 \seq \dots \seq p_n \dlassign e_n} &\DefEq  \InfActionSem{p_1 \dlassign e_1} \times \dots \times \InfActionSem{p_n \dlassign e_n} \\
    \InfActionSem{p \slassign e} &\DefEq \InfActionSem{e} \\
    \InfActionSem{\theta} &\DefEq \UnitSet \\
    \InfActionSem{\slbest i_1, \dots, i_n : \theta} &\DefEq \List{\Nats^n} \\
    \InfActionSem{\slaggregate i_1, \dots, i_n: \theta_1 \sland \theta_2} &\DefEq \EpsType \times \FDistr{\Nats^n} \\ \\
    \ValSem{\InferenceStrategy} &\;:\; \InfActionSem{\InferenceStrategy} \to \List{\Param \times \SymbBoundExpr \times
  \EpsType} \\
    \ValSem{p_1 \dlassign e_1 \seq \dots \seq p_n \dlassign e_n}((a_1, \dots, a_n)) &\DefEq
      \ValSem{p_1 \slassign e_1}(a_1) \ListCat \cdots \ListCat \ValSem{p_n \slassign e_n}(a_n) \\
    \ValSem{p \slassign e}(a) &\DefEq \ListComprehension{(p, e, \Eps)}{(e, \Eps) \in \ValSem{e}(a)} \\
    \ValSem{\theta \slwhen G}(\Unit) &\DefEq
      \ListSingleton{(\theta \WhenSymbBE G, \,0)} \\
    \ValSem{\slbest i_1, \dots, i_n : \theta \slwhen G}(J) &\DefEq
      \ListComprehension{(\FSubst{\theta}{i}{j} \WhenSymbBE \FSubst{G}{i}{j}, \, 0)}{j \in J}
  \end{align*}
  \begin{gather*}
    \ValSem{\slaggregate i_1, \dots, i_n: \theta_1 \sland \theta_2 \slwhen G}((\varepsilon, D)) \\[-2ex] \DefEq
    \ListSingleton{\
      \biggl(\,
        \biggl(\,
          \sum_{\lambda, j \,\in\, D} \! {\lambda \cdot \FSubst{\theta_1}{i}{j} }
          \ + \
          \InvCCDF\Bigl(\,
            \sum_{\lambda, j \,\in\, D} \! {\lambda \cdot \FSubst{\theta_2}{i}{j} }
            \,, \,\, \varepsilon \Bigl)
        \, \biggr)
        \, \WhenSymbBE
        \bigwedge_{\lambda, j \in D} \!\! \FSubst{G}{i}{j}
        \ , \ \ \Eps
      \,\biggr)
    \ }
  \end{gather*}
  \Description[]{}
  \caption{Syntax and semantics of the inference strategy language. Without loss of generality, all bounds are assumed to be upper-bounds (otherwise, substitute some parameter $p$ by $-p$). All \textsc{when}-clauses can be omitted when $G=\True$. We use the \emph{list concatenation} operator $\ListCat$ and the \emph{list comprehension} notation. $\FDistr{X}$ denotes the set of all distributions over set $X$ with finite support: $\FDistr{X} \equiv \{[(\lambda_1, x_1), \dots, (\lambda_n, x_n)] \,:\, n \in \Nats, \, x_1, \dots x_n \in X, \, \lambda_1, \dots, \lambda_n > 0, \, \sum_i \lambda_i = 1 \}$.}\label{fig:inf-lang-syntax-semantics}
\end{figure}

\paragraph{Proof obligations and soundness}
An inference strategy $\Inference$ induces a proof obligation in the form of a $\dL$ formula, $\ProofSem{\Inference}$, which is defined inductively in Figure~\ref{fig:inf-lang-obligations}. Provided that this obligation is valid, our language guarantees the soundness of all resulting inferences. This is expressed in Theorem~\ref{thm:inf-soundness}, for which a proof is provided in Appendix~\ref{ap:inf-lang-soundness-proof}.

\begin{figure}
  \EquationsFigureSize
  \begin{align*}
    \ProofSem{p_1 \dlassign e_1 \seq \dots \seq p_n \dlassign e_n} &\DefEq \ProofSem{p_1 \dlassign e_1} \wedge \dots \wedge \ProofSem{p_n \dlassign e_n} \\
    \ProofSem{p \slassign e} &\DefEq \Assum \wedge \Bigl(\bigwedge \VarDef{\FV(e)}\Bigr) \limply \ProofSem{e}(p) \\
    \ProofSem{\theta \slwhen G}(p) &\DefEq G \limply \BoundWith{p}{\theta} \\
    \ProofSem{\slbest i_1, \dots, i_n : \theta  \slwhen G}(p) &\DefEq G \limply \BoundWith{p}{\theta} \\
    \ProofSem{\slaggregate i_1, \dots, i_n: \theta_1 \sland \theta_2  \slwhen G}(p) &\DefEq G \limply \BoundWith{p}{\theta_1 + \theta_2}
  \end{align*}
  \begin{equation*}
    \VarDef{x_k} = \TagWithIndex{\VarDef{x}}{k} \quad
    \VarDef{p} = \Bound_p\  \quad
    \VarDef{\omega} = (\omega = \Obs_\omega) \quad
    \VarDef{v} = \Inv
  \end{equation*}
  \caption{Obligation induced by an inference strategy. Per convention, $p \in \Param$, $\omega \in \ObsVar$ and $v \in \StateVar$.}\label{fig:inf-lang-obligations}
  \Description[]{}
\end{figure}

\begin{theorem}[Soundness of the inference strategy language]\label{thm:inf-soundness}
  Let $\tuple{\Const,\dots,\Inference}$ a shield specification. If the formula $\ProofSem{\Inference}$ is valid, then $\sem{\Inference}(a)$ is a list of \emph{sound} inference assignments for any inference action $a \in \InfActionSem{\Inference}$.
\end{theorem}

\paragraph{Learning inference policies} A crucial feature of our inference strategy language is its nondeterminism: an inference strategy does not fully specify the behavior of the inference module and key decisions about what measurements to aggregate, how much budget to spend and when are left to a separate \emph{inference policy}. As illustrated in Section~\ref{sec:overview-inference}, such a policy must make subtle tradeoffs. However, it is not soundness-critical and amenable to learning. Learning an inference policy is possible within the \emph{meta-learning} setting described in Section~\ref{sec:learning-in-shielded-env}. A \emph{per-episode} safety-budget is fixed and the agent must learn a way to best manage this budget within an episode, in a way that generalizes across a family of environments and allows maximizing returns. Section~\ref{sec:experiments} illustrates this idea in three different case studies.

\section{Experimental Evaluation}\label{sec:experiments}

We evaluate our framework on a series of case studies and demonstrate the claims:

\begin{enumerate}[label=C\arabic*]
  \item \label{claim:modelling-flexibility} Our framework allows expressing diverse types of model uncertainty in a unified way.
  \item \label{claim:shield-safe} Shielded agents always remain safe.
  \item \label{claim:adaptive-shield-efficient} Adaptive shields allow more aggressive control strategies than non-adaptive shields.
  \item \label{claim:inference-learned} Inference policies can be learned that generalize across settings.
  \item \label{claim:control-info-learned} Control policies can be learned that actively seek information as a prerequisite for success.
\end{enumerate}

All case studies are summarized in Table~\ref{tab:experiments}. Each case study trains a shielded reinforcement learning agent, either in a fixed environment or using the meta-learning setting described in Section~\ref{sec:learning-in-shielded-env}. We demonstrate the qualitative claim~\ref{claim:modelling-flexibility} by having each case study illustrate distinct modelling concepts, all of which are listed in Table~\ref{tab:experiments}. Claims~\ref{claim:shield-safe} and \ref{claim:adaptive-shield-efficient} are evaluated on all case studies. As shown in Table~\ref{tab:crash-stats-short} and validating Theorem~\ref{thm:main-safety}, shielded agents never encounter crashes. Unshielded agents reach unsafe states during training, but also \emph{after} training in all but one environment. Furthermore, disabling the inference module leads to inferior policies in all examples. Claims~\ref{claim:inference-learned}~and~\ref{claim:control-info-learned} are relevant in the \emph{meta-learning} setting, which is explored in three case studies whose results we summarize below. Full details on our experimental protocol, results, and on the shield used for each case study are available in Appendix~\ref{ap:experiments}.

\paragraph{Versatile Train.} This case study uses the shield from Figure~\ref{fig:overview-train-local}, which is inspired by published models of train control systems~\cite{platzer2009european,kabra2022verified}, with added support for runtime track slope estimation. Our key results are summarized in Figures~\ref{fig:example2-1}~and~\ref{fig:example2-2}. We show how different inference policies work best for different kinds of train tracks and noise models. A policy that aggregates observations in \emph{small} batches works better in a setting with irregular tracks (large $k$) and low sensor noise (small $\sigma$), while a policy that aggregates observations in \emph{large} batches works better in settings with regular tracks and high sensor noise. Importantly, inference policies can be learned that are at least competitive with the best hardcoded solution in both cases. 

\paragraph{Crossing the River.} This case study illustrates the concept of \emph{active sensing} and demonstrates claim~\ref{claim:control-info-learned} in a minimal setting. The goal is for a robot equipped with a lamp to cross a bridge at night. The robot must find the position of the bridge, which translates into a model where the unknown is in the {safety property} instead of the plant. The robot only observes the bridge when it gets close enough and its lamp is on. In addition to an inference policy, it must learn a control policy that seeks knowledge \emph{explicitly}. We illustrate this learned control policy in Figure~\ref{fig:example3-2}.

\paragraph{Revisiting ACASX} This case study revisits a classic of cyberphysical-system verification: the next-generation Airborne Collision Avoidance System (ACAS X)~\cite{jeannin2015formally}. However, rather than assuming a \emph{known} intruder trajectory or a \emph{random-walk} intruder model~\cite{kochenderfer2012next}, we propose a new uncertainty model where intruders are either ACAS-compliant (in which case they are also actively trying to avoid collision) or not (in which case they must be treated as adversarial). The agent must attempt to infer the compliance of the intruder so as to avoid drastic maneuvers whenever possible. We illustrate the effectiveness of our learned inference policy in Figure~\ref{fig:example4-2}.

\newcommand{\ExpTablePropSep}{}
\newcommand{\ExpTablePropName}[1]{#1}
\newcommand{\ExpCaseStudyName}[1]{\textsc{#1}}

\newcommand{\ExpCaseColWidth}{20}
\newcommand{\ExpPropColWidth}{11}
\newcommand{\ExpDescColWidth}{\the\numexpr92-\ExpCaseColWidth-\ExpPropColWidth\relax}

\newcommand{\ExpTableFeatsLabel}{Concepts:}
\newcommand{\ExpTableDescrLabel}{Description:}
\newcommand{\ExpTableQuantLabel}{Quantitative:}
\newcommand{\ExpTableQualLabel}{Qualitative:}
\newcommand{\ExpTableResLabel}{Results:}

\newcommand{\MetaLearning}{\textsc{meta}}
\newcommand{\FixedSetting}{\textsc{fixed}}

\newcommand{\ExpTableEntry}[5]{%
  \multirow{4}{*}{\makecell[l]{\ExpCaseStudyName{#1}\\{\scriptsize#2}}}
  & \ExpTablePropName{\ExpTableFeatsLabel} &
  \emph{#3}
  \\ \ExpTablePropSep
  & \ExpTablePropName{\ExpTableDescrLabel} &
  #4
  \\ \ExpTablePropSep
  & \ExpTablePropName{\ExpTableResLabel} &
  #5
  \\
}

\begin{table}
  \small
  \renewcommand{\arraystretch}{1.1}
  \caption{Summary of our Experimental Evaluation. Each case study is annotated with references to the main experimental claims $\text{C}_i$ that it provides evidence for. A ``\FixedSetting'' annotation means that learning is conducted in a fixed environment, while a``\MetaLearning'' annotation means that the case study involves meta-learning (as defined in Section~\ref{sec:learning-in-shielded-env}). No shielded agent encountered an unsafe state in our experiments (Table~\ref{tab:crash-stats-short}).}\label{tab:experiments}
  \begin{tabular}{%
    p{0.\ExpCaseColWidth\linewidth}%
    p{0.\ExpPropColWidth\linewidth}%
    p{0.\ExpDescColWidth\linewidth}%
  }
    \toprule
    \textbf{Case Study}  &  \textbf{Property} & \textbf{} \\
    \midrule
    \ExpTableEntry{\hyperref[ap:sisyphean-train]{Sisyphean Train}}{%
     \ref{claim:modelling-flexibility},
      \ref{claim:shield-safe},
      \ref{claim:adaptive-shield-efficient},
      \FixedSetting
    }{Fixed environment, Concentration inequalities}{A train must go up-and-down a hill repeatedly to transport merchandise, protected by a variant of the shield from Figure~\ref{fig:overview-train-local} that assumes uniform measurement noise.}{Enabling the inference module yields a more efficient controller. Hoeffding bounds beat Chebyshev bounds for inference.}
    \midrule
    \ExpTableEntry{\hyperref[ap:versatile-train]{Versatile Train}}{%
      \ref{claim:modelling-flexibility},
      \ref{claim:shield-safe},
      \ref{claim:adaptive-shield-efficient},
      \ref{claim:inference-learned},
      \MetaLearning
    }{Meta-learning, Learned inference policy}{A train is exposed to a variety of track topologies and must learn to spend its safety budget wisely in unknown terrains.}{Different types of tracks call for different inference policies. Appropriate inference policies can be learned automatically.}
    \midrule
    \ExpTableEntry{\hyperref[ap:crossing-river]{Crossing the River}}{%
      \ref{claim:modelling-flexibility},
      \ref{claim:shield-safe},
      \ref{claim:adaptive-shield-efficient},
      \ref{claim:control-info-learned},
      \MetaLearning
    }{Mapping, Active sensing}{A robot equipped with a lamp must cross a bridge located at an unknown position along a river, at night. The lamp sheds light within a fixed radius.}{The robot successfully learns to approach the river, switch on its lamp and go along the river until it locates the bridge. Disabling the inference module prevents learning. A non-shielded agent frequently falls into the river, even \emph{after} training.}
    \midrule
    \ExpTableEntry{\hyperref[ap:acas]{Revisiting ACAS X}}{%
      \ref{claim:modelling-flexibility},
      \ref{claim:shield-safe},
      \ref{claim:adaptive-shield-efficient},
      \ref{claim:inference-learned},
      \MetaLearning
    }{Discrete model switching, Boolean inference}{%
    We consider a variant of the next-generation Airborne Collision Avoidance System (ACAS X)~\cite{jeannin2015formally}. The agent controls a plane and must learn to react to an intruder aircraft entering its airspace.%
    }{A provably-safe policy can be learned that attempts to infer whether or not the intruding aircraft flies in ACAS-compliant ways in order to avoid drastic maneuvers whenever possible. }
    \bottomrule
    \end{tabular}
    \renewcommand{\arraystretch}{1}
\end{table}

\begin{table}
  \small
  \caption{Return at testing time (average value of the last 100 episodes during 10,000 evaluation steps) and number of crashes during training and testing time, for different case studies and under different settings. The number of training steps and the maximum episode length for each setting are 80,000/100; 400,000/100; 400,000/100; 400,000/50; and 400,000/40 respectively. Agents are given a total safety budget of $10^{-3}$ in \emph{fixed-environment} settings (\textsc{Sisyphean Train}) and $10^{-7}$ per episode in \emph{meta-learning} settings (all others).}
  \label{tab:crash-stats-short}
  \begin{tabular}{lcccc}
      \toprule
      Setting & Method & Return & Crashes during training & Crashes during testing\\
      \midrule
      \multirow{3}{*}{\textsc{Sisyphean Train}} & Unshielded & $8.3\pm0.1$ & $110.7\pm21.5$ & $1.3\pm0.5$\\
 & Adaptive & $7.0\pm0.0$ & $0.0\pm0.0$ & $0.0\pm0.0$\\
 & Non-adaptive & $5.5\pm0.0$ & $0.0\pm0.0$ & $0.0\pm0.0$\\

      \midrule
      \multirow{3}{*}{\makecell{\textsc{Versatile Train}\\($k/\sigma$ large)}} & Unshielded & $10.0\pm0.7$ & $102.7\pm45.0$ & $0.0\pm0.0$\\
 & Adaptive & $9.1\pm0.3$ & $0.0\pm0.0$ & $0.0\pm0.0$\\
 & Non-adaptive & $5.8\pm0.0$ & $0.0\pm0.0$ & $0.0\pm0.0$\\

      \midrule
      \multirow{3}{*}{\makecell{\textsc{Versatile Train}\\($k/\sigma$ small)}} & Unshielded & $9.5\pm1.2$ & $94.3\pm6.2$ & $0.7\pm0.5$\\
 & Adaptive & $9.1\pm0.2$ & $0.0\pm0.0$ & $0.0\pm0.0$\\
 & Non-adaptive & $5.8\pm0.1$ & $0.0\pm0.0$ & $0.0\pm0.0$\\

      \midrule
      \multirow{3}{*}{\textsc{Crossing the River}} & Unshielded & $3.9\pm0.4$ & $325.7\pm81.8$ & $0.7\pm0.5$\\
 & Adaptive & $4.2\pm0.7$ & $0.0\pm0.0$ & $0.0\pm0.0$\\
 & Non-adaptive & $-5.0\pm0.0$ & $0.0\pm0.0$ & $0.0\pm0.0$\\

      \midrule
      \multirow{3}{*}{\textsc{Revisiting ACAS X}} & Unshielded & $6.7\pm0.5$ & $987.0\pm66.4$ & $25.0\pm6.7$\\
 & Adaptive & $6.6\pm0.2$ & $0.0\pm0.0$ & $0.0\pm0.0$\\
 & Non-adaptive & $4.5\pm0.0$ & $0.0\pm0.0$ & $0.0\pm0.0$\\

      \bottomrule
  \end{tabular}
\end{table}

\begin{figure}
  \vspace{0.2cm}
  \includegraphics{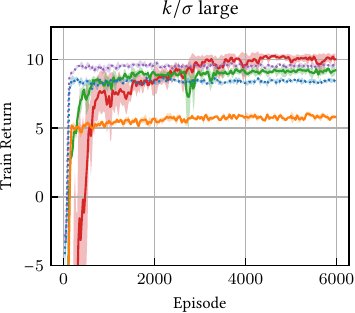}
  \hskip 15pt
  \includegraphics{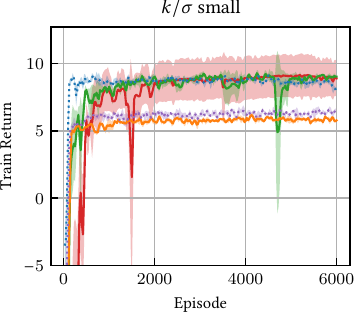}
  ~\\
  \includegraphics[trim=0 0 0 150,clip]{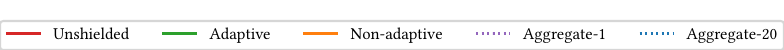}
  \caption{Returns for the \textsc{versatile train} case study, for two different combinations of $k$ and $\sigma$ (and averaged across three random seeds, with standard deviations represented as shaded areas around the mean). The train gets rewarded for reaching the station as quickly as possible (a fixed negative reward is issued at every step before the train reaches destination) and penalized for (unsafely) overshooting its target. Training returns are shown for an unshielded agent (red), an agent using our proposed adaptive shield (green), and a nonadaptive variant where the inference module is deactivated (orange). In addition, the \emph{Aggregate-$i$} agent is a variant of the \emph{Adaptive} agent that uses a hardcoded inference policy that spends a fixed budget to aggregate all available observations every $i$ steps. Different hardcoded inference policies perform best in different settings but the learned inference policy manages to match the best one in both cases.}
  \label{fig:example2-1}
  \Description[]{}
\end{figure}

\begin{figure}
  \includegraphics{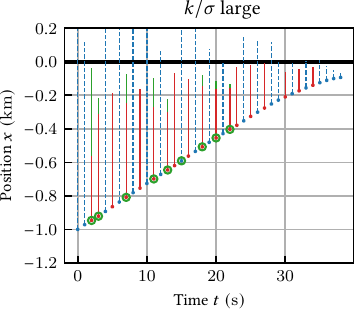}
  \hskip 15pt
  \includegraphics{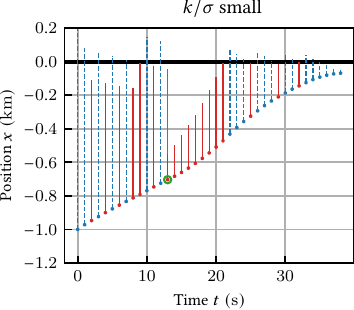}
  \caption{A visualization of the learned control and inference policies for the \textsc{versatile train} case study, on two random tracks and for two different combinations of $k$ and $\sigma$. We show the position of the train over time using dots. The train needs to stop before $x=0$ (bold black line). At each time step, a vertical segment indicates the estimated braking distance assuming that the train decides to accelerate for this time step. The segment is drawn in red if the train effectively decides to accelerate and in blue if it decides to brake. A green circle is shown whenever the inference module aggregates existing measurements to improve the local slope estimate. The resulting reduction in the shield's estimated braking distance is plotted as a green segment. In particular, we observe as expected that the learned inference policy aggregates measurements in bigger batches when $k/\sigma$ is small.}
  \label{fig:example2-2}
  \Description[]{}
\end{figure}

\begin{figure}
  \includegraphics{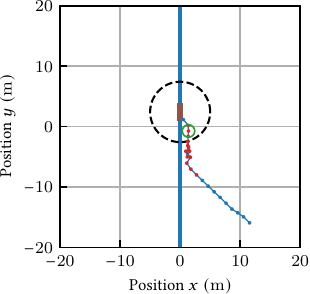}
  \hskip 15pt
  \includegraphics{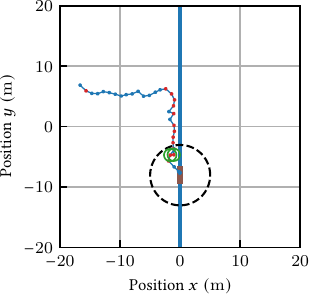}
  \caption{Visualization of the learned policy for the \textsc{crossing the river} case study. The river is drawn as a vertical blue line at $x=0$ and the bridge is drawn in brown. The range within which the agent can observe the bridge is plotted as a bold, black dashed circle. The agent's position is marked by a point for each control cycle, which is blue if the lamp is off and red if it is on. A green circle around the current position indicates some safety budget being spent. Regardless of its starting position, the agent learns to go near the river, activate the lamp when it gets close enough and then move along the river until it observes the bridge.}
  \label{fig:example3-2}
  \Description[]{}
\end{figure}

\begin{figure}
  \includegraphics{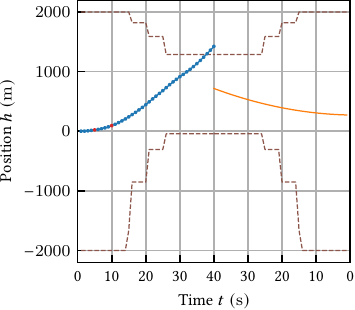}
  \hskip 15pt
  \includegraphics{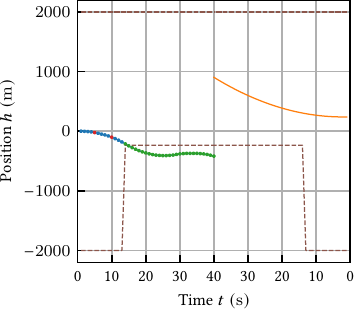}
  \caption{Learned policy for the \textsc{Revisiting ACAS X} case study, in two different scenarios. The scenario on the right features an ACAS-compliant intruder while the scenario on the left involves a non-compliant intruder. In both cases, the ownship's trajectory is shown on the left of the plot while the intruder's trajectory is shown on the right. The estimated upper and lower bounds for the intruder's final altitude at meeting time is shown for each time step as a dashed line. Red dots on the ownship's trajectory indicate the observation of information that is relevant to estimating the intruder's compliance. Green dots indicate that the intruder has been identified as compliant. When this happens, a less aggressive avoidance maneuver can be executed.}
  \label{fig:example4-2}
  \Description[]{}
\end{figure}

\section{Discussion}\label{sec:discussion}

\paragraph{Runtime overhead} Our framework pushes most of the burden of safety analysis offline via static proof obligations. Thus, it benefits from a small and predictable runtime overhead. This overhead results from running the inference module, monitoring the proposed control actions and executing the fallback policy when needed. The complexity of running a cycle of the inference module is linear in the size of the specification and in the size of the inference action being considered (e.g. the number of measurements being aggregated), assuming that the $\InvCCDF$ operator can be evaluated in linear time (which is true in particular in the Gaussian case or when using concentration inequalities, see Appendix~\ref{sec:eval-tail-expressions}). Running the controller monitor is inexpensive since it consists in evaluating a formula whose size does not exceed the size of the shield itself (see Appendix~\ref{ap:proof-obligations}). Finally, assuming that an explicit fallback policy is provided (see discussion in Section~\ref{sec:safe-rl-via-jsc}), the cost of computing fallback actions is small and predictable. In our experiments, the total overhead of shielding during training never exceeds 15\% (see Table~\ref{tab:overhead-stats}).

\paragraph{On the complexity of designing adaptive shields} There is an inherent complexity to the engineering task of designing adaptive monitors for rich, realistic model families, which typically necessitates the exploitation of domain-specific insights. Our framework allows designers to focus on this essential complexity, while eliminating the accidental complexity of enforcing end-to-end soundness guarantees that encompass statistical inference. For example, while our paper shows how tricky and error-prone statistical reasoning can be (e.g. reusing measurements is sound within an inference cycle but not across cycles), our framework exposes abstractions that fully protect users from such considerations and generates $\dL$ proof obligations that require no probabilistic reasoning.
Still, specifying formal models of hybrid systems and proving properties about them using interactive theorem proving is a nontrivial skill to acquire. Doing so is eased by the increasing availability of automation in tools such as KeYmaera X~\cite{DBLP:conf/cade/FultonMQVP15,DBLP:journals/fmsd/SogokonMTCP22}. We show the proof obligations generated by our tool for all case studies in Appendix~\ref{ap:proof-obligations}. Out of 32 obligations, 27 can be discharged fully automatically or with trivial human guidance (such as providing a single term for instantiating an existential variable). The others require more effort but can be proved in under an hour by an experienced KeYmaera X user. We provide an informal proof of the most challenging one (the proof of the train shield's invariant from Figure~\ref{fig:overview-train-local}) in Appendix~\ref{ap:train-invariant-proof}.

\section{Related Work}

\emph{Safe reinforcement learning} is widely studied~\cite{DBLP:journals/corr/abs-2205-10330,DBLP:journals/jmlr/GarciaF15,brunke2022safe}, including several approaches without formal verification. Several policy-search algorithms have been proposed that take into account safety constraints specified separately from the reward signal, but do not offer any guarantee that those constraints will not be violated during training or deployment~\cite{DBLP:conf/icml/AchiamHTA17,DBLP:conf/aaai/YangSTS21,DBLP:conf/iclr/TesslerMM19}. In contrast, this work is part of a family of approaches that are often called \emph{sandboxing} or \emph{shielding-based}~\cite{DBLP:conf/aaai/AlshiekhBEKNT18,DBLP:conf/aaai/FultonP18,DBLP:conf/icra/ThummA22}, and where the intended actions of an RL agent are monitored at runtime and overridden by appropriate fallbacks whenever they cannot be proved safe with respect to a model of the environment.

Virtually all existing approaches to shielding aim for full-automation of the shield computation, imposing hard tradeoffs in terms of safety, adaptivity, precision, expressivity and scalability. In circumstances where they run fast enough to be used online, methods based on reachability analysis~\cite{DBLP:conf/icra/ThummA22,DBLP:conf/cdc/KollerBT018,ivanov2019verisig} are naturally amenable to a form of adaptivity. However, they only offer bounded-horizon guarantees and precision comes at the expense of runtime efficiency. Methods based on LTL model-checking~\cite{DBLP:conf/aaai/AlshiekhBEKNT18,DBLP:conf/isola/KonighoferL0B20} or finite-MDP solving~\cite{pranger2021adaptive} can handle infinite time-horizons but require discretization of hybrid dynamics, often at a significant cost in terms of precision and safety. They also tend to scale poorly with increasing dimensionality. Methods based on Hamilton-Jacobi solving~\cite{DBLP:journals/tac/FisacAZKGT19} face similar challenges. In contrast, our proposed framework allows leveraging human ingenuity by extracting shields from nondeterministic, symbolic controllers that are proved safe using interactive theorem proving. It builds on the \emph{Justified Speculative Control} (JSC) framework~\cite{DBLP:conf/aaai/FultonP18}, which it generalizes in a crucial way to support adaptivity.

Our framework offers a uniquely expressive language for modelling environment uncertainty via arbitrarily constrained function symbols. For example, reachability analysis typically uses {bounded disturbance terms}~\cite{althoff2014online} to model environment uncertainty, which is sufficient to model the example from~Figure~\ref{fig:overview-train-global} but not those from Figures~\ref{fig:overview-train-local}~and~\ref{fig:acas}. To the best of our knowledge, these last two examples cannot be accommodated by any pre-existing approach. Another standard way of representing model uncertainty for the purpose of adaptive shielding is to use \emph{Gaussian processes} to model an unknown, state-dependent term added to the system's dynamics~\cite{DBLP:conf/aaai/ChengOMB19,DBLP:conf/eucc/BerkenkampS15,DBLP:conf/nips/BerkenkampTS017,DBLP:journals/tac/FisacAZKGT19}. This approach offers a different form of modelling flexibility, where assumptions about functional unknowns are implicitly encoded into prior kernels rather than hard logical constraints. However, this also makes the resulting safety guarantees harder to interpret and fundamentally dependent on assumptions that are nearly impossible to validate experimentally.

Another approach has been proposed to extend the JSC framework with a form of adaptivity, in which an agent starts with a finite set of plausible models and then progressively discards those that are found inconsistent with observations~\cite{DBLP:conf/tacas/FultonP19}. Our framework handles a more general form of parametric model uncertainty and additionally supports noisy observations and statistical reasoning. Finally, our idea of having experts define nondeterministic inference strategies that are sound by construction and refined via learning -- thereby making shielded agents in charge of their own safety budget -- has, to the best of our knowledge, no equivalent in the literature.

\section{Conclusion}

Our framework gives experts full access to the power of differential dynamic logic to build adaptive shields for hybrid systems. Its unique flexibility raises the equally unique challenge of performing statistical inference soundly and efficiently. We tackle this challenge with a mix of language design and learning, introducing the concepts of an \emph{inference strategy} and of an \emph{inference policy} respectively. Future work may explore the integration of reachability analysis and model-checking within our framework, allowing hybrid combinations of symbolic and numerical, offline and online proving.

\section*{Data-Availability Statement}
An artifact that contains our framework implementation, instructions to replicate all experiments, and formal proofs for all case studies is available at \url{https://doi.org/10.5281/zenodo.14916164}.

\begin{acks}
  This work was supported by the NSFC (Nos. 62276149, 92370124, 62350080, 92248303, U2341228, 62061136001), BNRist (BNR2022RC01006), Tsinghua Institute for Guo Qiang, and the High Performance Computing Center, Tsinghua University. J. Zhu was also supported by the XPlorer Prize. Finally, this work was supported by the Alexander von Humboldt Professorship program. We thank the anonymous reviewers for their helpful feedback.
\end{acks}

\bibliographystyle{ACM-Reference-Format}
\bibliography{main}


\begin{thebibliography}{44}


\ifx \showCODEN    \undefined \def \showCODEN     #1{\unskip}     \fi
\ifx \showDOI      \undefined \def \showDOI       #1{#1}\fi
\ifx \showISBNx    \undefined \def \showISBNx     #1{\unskip}     \fi
\ifx \showISBNxiii \undefined \def \showISBNxiii  #1{\unskip}     \fi
\ifx \showISSN     \undefined \def \showISSN      #1{\unskip}     \fi
\ifx \showLCCN     \undefined \def \showLCCN      #1{\unskip}     \fi
\ifx \shownote     \undefined \def \shownote      #1{#1}          \fi
\ifx \showarticletitle \undefined \def \showarticletitle #1{#1}   \fi
\ifx \showURL      \undefined \def \showURL       {\relax}        \fi
\providecommand\bibfield[2]{#2}
\providecommand\bibinfo[2]{#2}
\providecommand\natexlab[1]{#1}
\providecommand\showeprint[2][]{arXiv:#2}

\bibitem[Achiam et~al\mbox{.}(2017)]%
        {DBLP:conf/icml/AchiamHTA17}
\bibfield{author}{\bibinfo{person}{Joshua Achiam}, \bibinfo{person}{David Held}, \bibinfo{person}{Aviv Tamar}, {and} \bibinfo{person}{Pieter Abbeel}.} \bibinfo{year}{2017}\natexlab{}.
\newblock \showarticletitle{Constrained Policy Optimization}. In \bibinfo{booktitle}{\emph{Proceedings of the 34th International Conference on Machine Learning, {ICML} 2017, Sydney, NSW, Australia, 6-11 August 2017}} \emph{(\bibinfo{series}{Proceedings of Machine Learning Research}, Vol.~\bibinfo{volume}{70})}, \bibfield{editor}{\bibinfo{person}{Doina Precup} {and} \bibinfo{person}{Yee~Whye Teh}} (Eds.). \bibinfo{publisher}{{PMLR}}, \bibinfo{pages}{22--31}.
\newblock
\urldef\tempurl%
\url{http://proceedings.mlr.press/v70/achiam17a.html}
\showURL{%
\tempurl}


\bibitem[Alshiekh et~al\mbox{.}(2018)]%
        {DBLP:conf/aaai/AlshiekhBEKNT18}
\bibfield{author}{\bibinfo{person}{Mohammed Alshiekh}, \bibinfo{person}{Roderick Bloem}, \bibinfo{person}{R{\"{u}}diger Ehlers}, \bibinfo{person}{Bettina K{\"{o}}nighofer}, \bibinfo{person}{Scott Niekum}, {and} \bibinfo{person}{Ufuk Topcu}.} \bibinfo{year}{2018}\natexlab{}.
\newblock \showarticletitle{Safe Reinforcement Learning via Shielding}. In \bibinfo{booktitle}{\emph{Proceedings of the Thirty-Second {AAAI} Conference on Artificial Intelligence, (AAAI-18), the 30th innovative Applications of Artificial Intelligence (IAAI-18), and the 8th {AAAI} Symposium on Educational Advances in Artificial Intelligence (EAAI-18), New Orleans, Louisiana, USA, February 2-7, 2018}}, \bibfield{editor}{\bibinfo{person}{Sheila~A. McIlraith} {and} \bibinfo{person}{Kilian~Q. Weinberger}} (Eds.). \bibinfo{publisher}{{AAAI} Press}, \bibinfo{pages}{2669--2678}.
\newblock
\urldef\tempurl%
\url{https://doi.org/10.1609/AAAI.V32I1.11797}
\showDOI{\tempurl}


\bibitem[Althoff and Dolan(2014)]%
        {althoff2014online}
\bibfield{author}{\bibinfo{person}{Matthias Althoff} {and} \bibinfo{person}{John~M Dolan}.} \bibinfo{year}{2014}\natexlab{}.
\newblock \showarticletitle{Online verification of automated road vehicles using reachability analysis}.
\newblock \bibinfo{journal}{\emph{IEEE Transactions on Robotics}} \bibinfo{volume}{30}, \bibinfo{number}{4} (\bibinfo{year}{2014}), \bibinfo{pages}{903--918}.
\newblock
\urldef\tempurl%
\url{https://doi.org/10.1109/TRO.2014.2312453}
\showDOI{\tempurl}


\bibitem[Berkenkamp and Schoellig(2015)]%
        {DBLP:conf/eucc/BerkenkampS15}
\bibfield{author}{\bibinfo{person}{Felix Berkenkamp} {and} \bibinfo{person}{Angela~P. Schoellig}.} \bibinfo{year}{2015}\natexlab{}.
\newblock \showarticletitle{Safe and robust learning control with Gaussian processes}. In \bibinfo{booktitle}{\emph{14th European Control Conference, {ECC} 2015, Linz, Austria, July 15-17, 2015}}. \bibinfo{publisher}{{IEEE}}, \bibinfo{pages}{2496--2501}.
\newblock
\urldef\tempurl%
\url{https://doi.org/10.1109/ECC.2015.7330913}
\showDOI{\tempurl}


\bibitem[Berkenkamp et~al\mbox{.}(2017)]%
        {DBLP:conf/nips/BerkenkampTS017}
\bibfield{author}{\bibinfo{person}{Felix Berkenkamp}, \bibinfo{person}{Matteo Turchetta}, \bibinfo{person}{Angela~P. Schoellig}, {and} \bibinfo{person}{Andreas Krause}.} \bibinfo{year}{2017}\natexlab{}.
\newblock \showarticletitle{Safe Model-based Reinforcement Learning with Stability Guarantees}. In \bibinfo{booktitle}{\emph{Advances in Neural Information Processing Systems 30: Annual Conference on Neural Information Processing Systems 2017, December 4-9, 2017, Long Beach, CA, {USA}}}, \bibfield{editor}{\bibinfo{person}{Isabelle Guyon}, \bibinfo{person}{Ulrike von Luxburg}, \bibinfo{person}{Samy Bengio}, \bibinfo{person}{Hanna~M. Wallach}, \bibinfo{person}{Rob Fergus}, \bibinfo{person}{S.~V.~N. Vishwanathan}, {and} \bibinfo{person}{Roman Garnett}} (Eds.). \bibinfo{pages}{908--918}.
\newblock
\urldef\tempurl%
\url{https://proceedings.neurips.cc/paper/2017/hash/766ebcd59621e305170616ba3d3dac32-Abstract.html}
\showURL{%
\tempurl}


\bibitem[Bertsekas and Tsitsiklis(2008)]%
        {bertsekas2008introduction}
\bibfield{author}{\bibinfo{person}{Dimitri Bertsekas} {and} \bibinfo{person}{John~N Tsitsiklis}.} \bibinfo{year}{2008}\natexlab{}.
\newblock \bibinfo{booktitle}{\emph{Introduction to probability}}. Vol.~\bibinfo{volume}{1}.
\newblock \bibinfo{publisher}{Athena Scientific}.
\newblock


\bibitem[Bohrer and Platzer(2020)]%
        {bohrer2020constructive}
\bibfield{author}{\bibinfo{person}{Brandon Bohrer} {and} \bibinfo{person}{Andr{\'e} Platzer}.} \bibinfo{year}{2020}\natexlab{}.
\newblock \showarticletitle{Constructive game logic}. In \bibinfo{booktitle}{\emph{Programming Languages and Systems: 29th European Symposium on Programming, ESOP 2020, Held as Part of the European Joint Conferences on Theory and Practice of Software, ETAPS 2020, Dublin, Ireland, April 25--30, 2020, Proceedings 29}}. Springer International Publishing, \bibinfo{pages}{84--111}.
\newblock
\urldef\tempurl%
\url{https://doi.org/10.1007/978-3-030-44914-8\_4}
\showDOI{\tempurl}


\bibitem[Brunke et~al\mbox{.}(2022)]%
        {brunke2022safe}
\bibfield{author}{\bibinfo{person}{Lukas Brunke}, \bibinfo{person}{Melissa Greeff}, \bibinfo{person}{Adam~W Hall}, \bibinfo{person}{Zhaocong Yuan}, \bibinfo{person}{Siqi Zhou}, \bibinfo{person}{Jacopo Panerati}, {and} \bibinfo{person}{Angela~P Schoellig}.} \bibinfo{year}{2022}\natexlab{}.
\newblock \showarticletitle{Safe learning in robotics: From learning-based control to safe reinforcement learning}.
\newblock \bibinfo{journal}{\emph{Annual Review of Control, Robotics, and Autonomous Systems}}  \bibinfo{volume}{5} (\bibinfo{year}{2022}), \bibinfo{pages}{411--444}.
\newblock
\urldef\tempurl%
\url{https://doi.org/10.1146/ANNUREV-CONTROL-042920-020211}
\showDOI{\tempurl}


\bibitem[Cheng et~al\mbox{.}(2019)]%
        {DBLP:conf/aaai/ChengOMB19}
\bibfield{author}{\bibinfo{person}{Richard Cheng}, \bibinfo{person}{G{\'{a}}bor Orosz}, \bibinfo{person}{Richard~M. Murray}, {and} \bibinfo{person}{Joel~W. Burdick}.} \bibinfo{year}{2019}\natexlab{}.
\newblock \showarticletitle{End-to-End Safe Reinforcement Learning through Barrier Functions for Safety-Critical Continuous Control Tasks}. In \bibinfo{booktitle}{\emph{The Thirty-Third {AAAI} Conference on Artificial Intelligence, {AAAI} 2019, The Thirty-First Innovative Applications of Artificial Intelligence Conference, {IAAI} 2019, The Ninth {AAAI} Symposium on Educational Advances in Artificial Intelligence, {EAAI} 2019, Honolulu, Hawaii, USA, January 27 - February 1, 2019}}. \bibinfo{publisher}{{AAAI} Press}, \bibinfo{pages}{3387--3395}.
\newblock
\urldef\tempurl%
\url{https://doi.org/10.1609/aaai.v33i01.33013387}
\showDOI{\tempurl}


\bibitem[Elsayed{-}Aly et~al\mbox{.}(2021)]%
        {DBLP:conf/atal/Elsayed-AlyBAET21}
\bibfield{author}{\bibinfo{person}{Ingy Elsayed{-}Aly}, \bibinfo{person}{Suda Bharadwaj}, \bibinfo{person}{Christopher Amato}, \bibinfo{person}{R{\"{u}}diger Ehlers}, \bibinfo{person}{Ufuk Topcu}, {and} \bibinfo{person}{Lu Feng}.} \bibinfo{year}{2021}\natexlab{}.
\newblock \showarticletitle{Safe Multi-Agent Reinforcement Learning via Shielding}. In \bibinfo{booktitle}{\emph{{AAMAS} '21: 20th International Conference on Autonomous Agents and Multiagent Systems, Virtual Event, United Kingdom, May 3-7, 2021}}, \bibfield{editor}{\bibinfo{person}{Frank Dignum}, \bibinfo{person}{Alessio Lomuscio}, \bibinfo{person}{Ulle Endriss}, {and} \bibinfo{person}{Ann Now{\'{e}}}} (Eds.). \bibinfo{publisher}{{ACM}}, \bibinfo{pages}{483--491}.
\newblock
\urldef\tempurl%
\url{https://doi.org/10.5555/3463952.3464013}
\showDOI{\tempurl}


\bibitem[Fisac et~al\mbox{.}(2019)]%
        {DBLP:journals/tac/FisacAZKGT19}
\bibfield{author}{\bibinfo{person}{Jaime~F. Fisac}, \bibinfo{person}{Anayo~K. Akametalu}, \bibinfo{person}{Melanie~N. Zeilinger}, \bibinfo{person}{Shahab Kaynama}, \bibinfo{person}{Jeremy~H. Gillula}, {and} \bibinfo{person}{Claire~J. Tomlin}.} \bibinfo{year}{2019}\natexlab{}.
\newblock \showarticletitle{A General Safety Framework for Learning-Based Control in Uncertain Robotic Systems}.
\newblock \bibinfo{journal}{\emph{{IEEE} Trans. Autom. Control.}} \bibinfo{volume}{64}, \bibinfo{number}{7} (\bibinfo{year}{2019}), \bibinfo{pages}{2737--2752}.
\newblock
\urldef\tempurl%
\url{https://doi.org/10.1109/TAC.2018.2876389}
\showDOI{\tempurl}


\bibitem[Fulton et~al\mbox{.}(2015)]%
        {DBLP:conf/cade/FultonMQVP15}
\bibfield{author}{\bibinfo{person}{Nathan Fulton}, \bibinfo{person}{Stefan Mitsch}, \bibinfo{person}{Jan{-}David Quesel}, \bibinfo{person}{Marcus V{\"{o}}lp}, {and} \bibinfo{person}{Andr{\'{e}} Platzer}.} \bibinfo{year}{2015}\natexlab{}.
\newblock \showarticletitle{KeYmaera {X:} An Axiomatic Tactical Theorem Prover for Hybrid Systems}. In \bibinfo{booktitle}{\emph{Automated Deduction - {CADE-25} - 25th International Conference on Automated Deduction, Berlin, Germany, August 1-7, 2015, Proceedings}} \emph{(\bibinfo{series}{LNCS}, Vol.~\bibinfo{volume}{9195})}, \bibfield{editor}{\bibinfo{person}{Amy~P. Felty} {and} \bibinfo{person}{Aart Middeldorp}} (Eds.). \bibinfo{publisher}{Springer}, \bibinfo{pages}{527--538}.
\newblock
\urldef\tempurl%
\url{https://doi.org/10.1007/978-3-319-21401-6_36}
\showDOI{\tempurl}


\bibitem[Fulton and Platzer(2018)]%
        {DBLP:conf/aaai/FultonP18}
\bibfield{author}{\bibinfo{person}{Nathan Fulton} {and} \bibinfo{person}{Andr{\'{e}} Platzer}.} \bibinfo{year}{2018}\natexlab{}.
\newblock \showarticletitle{Safe Reinforcement Learning via Formal Methods: Toward Safe Control Through Proof and Learning}. In \bibinfo{booktitle}{\emph{8th {AAAI} Symposium on Educational Advances in Artificial Intelligence (EAAI-18), New Orleans, Louisiana, USA, February 2-7, 2018}}, \bibfield{editor}{\bibinfo{person}{Sheila~A. McIlraith} {and} \bibinfo{person}{Kilian~Q. Weinberger}} (Eds.). \bibinfo{publisher}{{AAAI} Press}, \bibinfo{pages}{6485--6492}.
\newblock
\urldef\tempurl%
\url{https://doi.org/10.1609/AAAI.V32I1.12107}
\showDOI{\tempurl}


\bibitem[Fulton and Platzer(2019)]%
        {DBLP:conf/tacas/FultonP19}
\bibfield{author}{\bibinfo{person}{Nathan Fulton} {and} \bibinfo{person}{Andr{\'{e}} Platzer}.} \bibinfo{year}{2019}\natexlab{}.
\newblock \showarticletitle{Verifiably Safe Off-Model Reinforcement Learning}. In \bibinfo{booktitle}{\emph{Tools and Algorithms for the Construction and Analysis of Systems - 25th International Conference, {TACAS} 2019, Held as Part of the European Joint Conferences on Theory and Practice of Software, {ETAPS} 2019, Prague, Czech Republic, April 6-11, 2019, Proceedings, Part {I}}} \emph{(\bibinfo{series}{LNCS}, Vol.~\bibinfo{volume}{11427})}, \bibfield{editor}{\bibinfo{person}{Tom{\'{a}}s Vojnar} {and} \bibinfo{person}{Lijun Zhang}} (Eds.). \bibinfo{publisher}{Springer}, \bibinfo{pages}{413--430}.
\newblock
\urldef\tempurl%
\url{https://doi.org/10.1007/978-3-030-17462-0_28}
\showDOI{\tempurl}


\bibitem[Garc{\'{\i}}a and Fern{\'{a}}ndez(2015)]%
        {DBLP:journals/jmlr/GarciaF15}
\bibfield{author}{\bibinfo{person}{Javier Garc{\'{\i}}a} {and} \bibinfo{person}{Fernando Fern{\'{a}}ndez}.} \bibinfo{year}{2015}\natexlab{}.
\newblock \showarticletitle{A comprehensive survey on safe reinforcement learning}.
\newblock \bibinfo{journal}{\emph{J. Mach. Learn. Res.}}  \bibinfo{volume}{16} (\bibinfo{year}{2015}), \bibinfo{pages}{1437--1480}.
\newblock
\urldef\tempurl%
\url{https://doi.org/10.5555/2789272.2886795}
\showDOI{\tempurl}


\bibitem[Gu et~al\mbox{.}(2022)]%
        {DBLP:journals/corr/abs-2205-10330}
\bibfield{author}{\bibinfo{person}{Shangding Gu}, \bibinfo{person}{Long Yang}, \bibinfo{person}{Yali Du}, \bibinfo{person}{Guang Chen}, \bibinfo{person}{Florian Walter}, \bibinfo{person}{Jun Wang}, \bibinfo{person}{Yaodong Yang}, {and} \bibinfo{person}{Alois~C. Knoll}.} \bibinfo{year}{2022}\natexlab{}.
\newblock \showarticletitle{A Review of Safe Reinforcement Learning: Methods, Theory and Applications}.
\newblock \bibinfo{journal}{\emph{CoRR}}  \bibinfo{volume}{abs/2205.10330} (\bibinfo{year}{2022}).
\newblock
\urldef\tempurl%
\url{https://doi.org/10.48550/arXiv.2205.10330}
\showDOI{\tempurl}
\showeprint[arXiv]{2205.10330}


\bibitem[Haarnoja et~al\mbox{.}(2018)]%
        {haarnoja2018soft}
\bibfield{author}{\bibinfo{person}{Tuomas Haarnoja}, \bibinfo{person}{Aurick Zhou}, \bibinfo{person}{Pieter Abbeel}, {and} \bibinfo{person}{Sergey Levine}.} \bibinfo{year}{2018}\natexlab{}.
\newblock \showarticletitle{Soft actor-critic: Off-policy maximum entropy deep reinforcement learning with a stochastic actor}. In \bibinfo{booktitle}{\emph{International conference on machine learning}}. PMLR, \bibinfo{pages}{1861--1870}.
\newblock


\bibitem[Hester et~al\mbox{.}(2012)]%
        {DBLP:conf/icra/HesterQS12}
\bibfield{author}{\bibinfo{person}{Todd Hester}, \bibinfo{person}{Michael~J. Quinlan}, {and} \bibinfo{person}{Peter Stone}.} \bibinfo{year}{2012}\natexlab{}.
\newblock \showarticletitle{{RTMBA:} {A} Real-Time Model-Based Reinforcement Learning Architecture for robot control}. In \bibinfo{booktitle}{\emph{{IEEE} International Conference on Robotics and Automation, {ICRA} 2012, 14-18 May, 2012, St. Paul, Minnesota, {USA}}}. \bibinfo{publisher}{{IEEE}}, \bibinfo{pages}{85--90}.
\newblock
\urldef\tempurl%
\url{https://doi.org/10.1109/ICRA.2012.6225072}
\showDOI{\tempurl}


\bibitem[Hunt et~al\mbox{.}(2021)]%
        {DBLP:conf/hybrid/HuntFMHDS21}
\bibfield{author}{\bibinfo{person}{Nathan Hunt}, \bibinfo{person}{Nathan Fulton}, \bibinfo{person}{Sara Magliacane}, \bibinfo{person}{Trong~Nghia Hoang}, \bibinfo{person}{Subhro Das}, {and} \bibinfo{person}{Armando Solar{-}Lezama}.} \bibinfo{year}{2021}\natexlab{}.
\newblock \showarticletitle{Verifiably safe exploration for end-to-end reinforcement learning}. In \bibinfo{booktitle}{\emph{{HSCC} '21: 24th {ACM} International Conference on Hybrid Systems: Computation and Control, Nashville, Tennessee, May 19-21, 2021}}, \bibfield{editor}{\bibinfo{person}{Sergiy Bogomolov} {and} \bibinfo{person}{Rapha{\"{e}}l~M. Jungers}} (Eds.). \bibinfo{publisher}{{ACM}}, \bibinfo{pages}{14:1--14:11}.
\newblock
\urldef\tempurl%
\url{https://doi.org/10.1145/3447928.3456653}
\showDOI{\tempurl}


\bibitem[Ivanov et~al\mbox{.}(2019)]%
        {ivanov2019verisig}
\bibfield{author}{\bibinfo{person}{Radoslav Ivanov}, \bibinfo{person}{James Weimer}, \bibinfo{person}{Rajeev Alur}, \bibinfo{person}{George~J Pappas}, {and} \bibinfo{person}{Insup Lee}.} \bibinfo{year}{2019}\natexlab{}.
\newblock \showarticletitle{Verisig: verifying safety properties of hybrid systems with neural network controllers}. In \bibinfo{booktitle}{\emph{Proceedings of the 22nd ACM International Conference on Hybrid Systems: Computation and Control}}. \bibinfo{pages}{169--178}.
\newblock
\urldef\tempurl%
\url{https://doi.org/10.1145/3302504.3311806}
\showDOI{\tempurl}


\bibitem[Jeannin et~al\mbox{.}(2015)]%
        {jeannin2015formally}
\bibfield{author}{\bibinfo{person}{Jean-Baptiste Jeannin}, \bibinfo{person}{Khalil Ghorbal}, \bibinfo{person}{Yanni Kouskoulas}, \bibinfo{person}{Ryan Gardner}, \bibinfo{person}{Aurora Schmidt}, \bibinfo{person}{Erik Zawadzki}, {and} \bibinfo{person}{Andr{\'e} Platzer}.} \bibinfo{year}{2015}\natexlab{}.
\newblock \showarticletitle{A formally verified hybrid system for the next-generation airborne collision avoidance system}. In \bibinfo{booktitle}{\emph{Tools and Algorithms for the Construction and Analysis of Systems: 21st International Conference, TACAS 2015, Held as Part of the European Joint Conferences on Theory and Practice of Software, ETAPS 2015, London, UK, April 11-18, 2015, Proceedings 21}}. Springer, \bibinfo{pages}{21--36}.
\newblock
\urldef\tempurl%
\url{https://doi.org/10.1007/978-3-662-46681-0\_2}
\showDOI{\tempurl}


\bibitem[Kabra et~al\mbox{.}(2022)]%
        {kabra2022verified}
\bibfield{author}{\bibinfo{person}{Aditi Kabra}, \bibinfo{person}{Stefan Mitsch}, {and} \bibinfo{person}{Andr{\'e} Platzer}.} \bibinfo{year}{2022}\natexlab{}.
\newblock \showarticletitle{Verified train controllers for the federal railroad administration train kinematics model: Balancing competing brake and track forces}.
\newblock \bibinfo{journal}{\emph{IEEE Transactions on Computer-Aided Design of Integrated Circuits and Systems}} \bibinfo{volume}{41}, \bibinfo{number}{11} (\bibinfo{year}{2022}), \bibinfo{pages}{4409--4420}.
\newblock
\urldef\tempurl%
\url{https://doi.org/10.1109/TCAD.2022.3197690}
\showDOI{\tempurl}


\bibitem[Kochenderfer(2015)]%
        {kochenderfer2015decision}
\bibfield{author}{\bibinfo{person}{Mykel~J Kochenderfer}.} \bibinfo{year}{2015}\natexlab{}.
\newblock \bibinfo{booktitle}{\emph{Decision making under uncertainty: theory and application}}.
\newblock \bibinfo{publisher}{MIT press}.
\newblock


\bibitem[Kochenderfer et~al\mbox{.}(2012)]%
        {kochenderfer2012next}
\bibfield{author}{\bibinfo{person}{Mykel~J Kochenderfer}, \bibinfo{person}{Jessica~E Holland}, {and} \bibinfo{person}{James~P Chryssanthacopoulos}.} \bibinfo{year}{2012}\natexlab{}.
\newblock \showarticletitle{Next generation airborne collision avoidance system}.
\newblock \bibinfo{journal}{\emph{Lincoln Laboratory Journal}} \bibinfo{volume}{19}, \bibinfo{number}{1} (\bibinfo{year}{2012}), \bibinfo{pages}{17--33}.
\newblock


\bibitem[Koller et~al\mbox{.}(2018)]%
        {DBLP:conf/cdc/KollerBT018}
\bibfield{author}{\bibinfo{person}{Torsten Koller}, \bibinfo{person}{Felix Berkenkamp}, \bibinfo{person}{Matteo Turchetta}, {and} \bibinfo{person}{Andreas Krause}.} \bibinfo{year}{2018}\natexlab{}.
\newblock \showarticletitle{Learning-Based Model Predictive Control for Safe Exploration}. In \bibinfo{booktitle}{\emph{57th {IEEE} Conference on Decision and Control, {CDC} 2018, Miami, FL, USA, December 17-19, 2018}}. \bibinfo{publisher}{{IEEE}}, \bibinfo{pages}{6059--6066}.
\newblock
\urldef\tempurl%
\url{https://doi.org/10.1109/CDC.2018.8619572}
\showDOI{\tempurl}


\bibitem[K{\"{o}}nighofer et~al\mbox{.}(2020)]%
        {DBLP:conf/isola/KonighoferL0B20}
\bibfield{author}{\bibinfo{person}{Bettina K{\"{o}}nighofer}, \bibinfo{person}{Florian Lorber}, \bibinfo{person}{Nils Jansen}, {and} \bibinfo{person}{Roderick Bloem}.} \bibinfo{year}{2020}\natexlab{}.
\newblock \showarticletitle{Shield Synthesis for Reinforcement Learning}. In \bibinfo{booktitle}{\emph{Leveraging Applications of Formal Methods, Verification and Validation: Verification Principles - 9th International Symposium on Leveraging Applications of Formal Methods, ISoLA 2020, Rhodes, Greece, October 20-30, 2020, Proceedings, Part {I}}} \emph{(\bibinfo{series}{Lecture Notes in Computer Science}, Vol.~\bibinfo{volume}{12476})}, \bibfield{editor}{\bibinfo{person}{Tiziana Margaria} {and} \bibinfo{person}{Bernhard Steffen}} (Eds.). \bibinfo{publisher}{Springer}, \bibinfo{pages}{290--306}.
\newblock
\urldef\tempurl%
\url{https://doi.org/10.1007/978-3-030-61362-4\_16}
\showDOI{\tempurl}


\bibitem[Liang et~al\mbox{.}(2018)]%
        {DBLP:conf/eccv/LiangWYX18}
\bibfield{author}{\bibinfo{person}{Xiaodan Liang}, \bibinfo{person}{Tairui Wang}, \bibinfo{person}{Luona Yang}, {and} \bibinfo{person}{Eric~P. Xing}.} \bibinfo{year}{2018}\natexlab{}.
\newblock \showarticletitle{{CIRL:} Controllable Imitative Reinforcement Learning for Vision-Based Self-driving}. In \bibinfo{booktitle}{\emph{Computer Vision - {ECCV} 2018 - 15th European Conference, Munich, Germany, September 8-14, 2018, Proceedings, Part {VII}}} \emph{(\bibinfo{series}{LNCS}, Vol.~\bibinfo{volume}{11211})}, \bibfield{editor}{\bibinfo{person}{Vittorio Ferrari}, \bibinfo{person}{Martial Hebert}, \bibinfo{person}{Cristian Sminchisescu}, {and} \bibinfo{person}{Yair Weiss}} (Eds.). \bibinfo{publisher}{Springer}, \bibinfo{pages}{604--620}.
\newblock
\urldef\tempurl%
\url{https://doi.org/10.1007/978-3-030-01234-2_36}
\showDOI{\tempurl}


\bibitem[Mitsch and Platzer(2016)]%
        {DBLP:journals/fmsd/MitschP16}
\bibfield{author}{\bibinfo{person}{Stefan Mitsch} {and} \bibinfo{person}{Andr{\'{e}} Platzer}.} \bibinfo{year}{2016}\natexlab{}.
\newblock \showarticletitle{{ModelPlex}: verified runtime validation of verified cyber-physical system models}.
\newblock \bibinfo{journal}{\emph{Formal Methods Syst. Des.}} \bibinfo{volume}{49}, \bibinfo{number}{1-2} (\bibinfo{year}{2016}), \bibinfo{pages}{33--74}.
\newblock
\urldef\tempurl%
\url{https://doi.org/10.1007/s10703-016-0241-z}
\showDOI{\tempurl}


\bibitem[Munroe(2011)]%
        {significant}
\bibfield{author}{\bibinfo{person}{Randall Munroe}.} \bibinfo{year}{2011}\natexlab{}.
\newblock \bibinfo{title}{Significant}.
\newblock \bibinfo{howpublished}{\url{http://xkcd.com/882/}}.
\newblock


\bibitem[Platzer(2007)]%
        {DBLP:conf/hybrid/Platzer07}
\bibfield{author}{\bibinfo{person}{Andr{\'e} Platzer}.} \bibinfo{year}{2007}\natexlab{}.
\newblock \showarticletitle{Differential Logic for Reasoning about Hybrid Systems}. In \bibinfo{booktitle}{\emph{HSCC}} \emph{(\bibinfo{series}{LNCS}, Vol.~\bibinfo{volume}{4416})}, \bibfield{editor}{\bibinfo{person}{Alberto Bemporad}, \bibinfo{person}{Antonio Bicchi}, {and} \bibinfo{person}{Giorgio Buttazzo}} (Eds.). \bibinfo{publisher}{Springer}, \bibinfo{pages}{746--749}.
\newblock
\showISBNx{978-3-540-71492-7}
\urldef\tempurl%
\url{https://doi.org/10.1007/978-3-540-71493-4_75}
\showDOI{\tempurl}


\bibitem[Platzer(2008)]%
        {DBLP:journals/jar/Platzer08}
\bibfield{author}{\bibinfo{person}{Andr{\'{e}} Platzer}.} \bibinfo{year}{2008}\natexlab{}.
\newblock \showarticletitle{Differential Dynamic Logic for Hybrid Systems}.
\newblock \bibinfo{journal}{\emph{J. Autom. Reason.}} \bibinfo{volume}{41}, \bibinfo{number}{2} (\bibinfo{year}{2008}), \bibinfo{pages}{143--189}.
\newblock
\urldef\tempurl%
\url{https://doi.org/10.1007/s10817-008-9103-8}
\showDOI{\tempurl}


\bibitem[Platzer(2017)]%
        {DBLP:journals/jar/Platzer17}
\bibfield{author}{\bibinfo{person}{Andr{\'{e}} Platzer}.} \bibinfo{year}{2017}\natexlab{}.
\newblock \showarticletitle{A Complete Uniform Substitution Calculus for Differential Dynamic Logic}.
\newblock \bibinfo{journal}{\emph{J. Autom. Reason.}} \bibinfo{volume}{59}, \bibinfo{number}{2} (\bibinfo{year}{2017}), \bibinfo{pages}{219--265}.
\newblock
\urldef\tempurl%
\url{https://doi.org/10.1007/s10817-016-9385-1}
\showDOI{\tempurl}


\bibitem[Platzer and Quesel(2008)]%
        {DBLP:conf/hybrid/PlatzerQ08}
\bibfield{author}{\bibinfo{person}{Andr{\'e} Platzer} {and} \bibinfo{person}{Jan-David Quesel}.} \bibinfo{year}{2008}\natexlab{}.
\newblock \showarticletitle{Logical Verification and Systematic Parametric Analysis in Train Control.}. In \bibinfo{booktitle}{\emph{HSCC}} \emph{(\bibinfo{series}{LNCS}, Vol.~\bibinfo{volume}{4981})}, \bibfield{editor}{\bibinfo{person}{Magnus Egerstedt} {and} \bibinfo{person}{Bud Mishra}} (Eds.). \bibinfo{publisher}{Springer}, \bibinfo{pages}{646--649}.
\newblock
\showISBNx{978-3-540-78928-4}
\urldef\tempurl%
\url{https://doi.org/10.1007/978-3-540-78929-1_55}
\showDOI{\tempurl}


\bibitem[Platzer and Quesel(2009)]%
        {platzer2009european}
\bibfield{author}{\bibinfo{person}{Andr{\'e} Platzer} {and} \bibinfo{person}{Jan-David Quesel}.} \bibinfo{year}{2009}\natexlab{}.
\newblock \showarticletitle{European Train Control System: A case study in formal verification}. In \bibinfo{booktitle}{\emph{International Conference on Formal Engineering Methods}}. Springer, \bibinfo{pages}{246--265}.
\newblock
\urldef\tempurl%
\url{https://doi.org/10.1007/978-3-642-10373-5\_13}
\showDOI{\tempurl}


\bibitem[Platzer and Tan(2018)]%
        {platzer2018differential}
\bibfield{author}{\bibinfo{person}{Andr{\'e} Platzer} {and} \bibinfo{person}{Yong~Kiam Tan}.} \bibinfo{year}{2018}\natexlab{}.
\newblock \showarticletitle{Differential equation axiomatization: The impressive power of differential ghosts}. In \bibinfo{booktitle}{\emph{Proceedings of the 33rd Annual ACM/IEEE Symposium on Logic in Computer Science}}. \bibinfo{pages}{819--828}.
\newblock
\urldef\tempurl%
\url{https://doi.org/10.1145/3209108.3209147}
\showDOI{\tempurl}


\bibitem[Pranger et~al\mbox{.}(2021)]%
        {pranger2021adaptive}
\bibfield{author}{\bibinfo{person}{Stefan Pranger}, \bibinfo{person}{Bettina K{\"o}nighofer}, \bibinfo{person}{Martin Tappler}, \bibinfo{person}{Martin Deixelberger}, \bibinfo{person}{Nils Jansen}, {and} \bibinfo{person}{Roderick Bloem}.} \bibinfo{year}{2021}\natexlab{}.
\newblock \showarticletitle{Adaptive shielding under uncertainty}. In \bibinfo{booktitle}{\emph{2021 American Control Conference (ACC)}}. IEEE, \bibinfo{pages}{3467--3474}.
\newblock
\urldef\tempurl%
\url{https://doi.org/10.23919/ACC50511.2021.9482889}
\showDOI{\tempurl}


\bibitem[Roveda et~al\mbox{.}(2020)]%
        {DBLP:journals/jirs/RovedaMFABTP20}
\bibfield{author}{\bibinfo{person}{Loris Roveda}, \bibinfo{person}{Jeyhoon Maskani}, \bibinfo{person}{Paolo Franceschi}, \bibinfo{person}{Arash Abdi}, \bibinfo{person}{Francesco Braghin}, \bibinfo{person}{Lorenzo~Molinari Tosatti}, {and} \bibinfo{person}{Nicola Pedrocchi}.} \bibinfo{year}{2020}\natexlab{}.
\newblock \showarticletitle{Model-Based Reinforcement Learning Variable Impedance Control for Human-Robot Collaboration}.
\newblock \bibinfo{journal}{\emph{J. Intell. Robotic Syst.}} \bibinfo{volume}{100}, \bibinfo{number}{2} (\bibinfo{year}{2020}), \bibinfo{pages}{417--433}.
\newblock
\urldef\tempurl%
\url{https://doi.org/10.1007/s10846-020-01183-3}
\showDOI{\tempurl}


\bibitem[Shalev{-}Shwartz et~al\mbox{.}(2016)]%
        {DBLP:journals/corr/Shalev-ShwartzS16a}
\bibfield{author}{\bibinfo{person}{Shai Shalev{-}Shwartz}, \bibinfo{person}{Shaked Shammah}, {and} \bibinfo{person}{Amnon Shashua}.} \bibinfo{year}{2016}\natexlab{}.
\newblock \showarticletitle{Safe, Multi-Agent, Reinforcement Learning for Autonomous Driving}.
\newblock \bibinfo{journal}{\emph{CoRR}}  \bibinfo{volume}{abs/1610.03295} (\bibinfo{year}{2016}).
\newblock
\showeprint[arXiv]{1610.03295}
\urldef\tempurl%
\url{http://arxiv.org/abs/1610.03295}
\showURL{%
\tempurl}


\bibitem[Smart and Kaelbling(2002)]%
        {DBLP:conf/icra/SmartK02}
\bibfield{author}{\bibinfo{person}{William~D. Smart} {and} \bibinfo{person}{Leslie~Pack Kaelbling}.} \bibinfo{year}{2002}\natexlab{}.
\newblock \showarticletitle{Effective Reinforcement Learning for Mobile Robots}. In \bibinfo{booktitle}{\emph{Proceedings of the 2002 {IEEE} International Conference on Robotics and Automation, {ICRA} 2002, May 11-15, 2002, Washington, DC, {USA}}}. \bibinfo{publisher}{{IEEE}}, \bibinfo{pages}{3404--3410}.
\newblock
\urldef\tempurl%
\url{https://doi.org/10.1109/ROBOT.2002.1014237}
\showDOI{\tempurl}


\bibitem[Sogokon et~al\mbox{.}(2022)]%
        {DBLP:journals/fmsd/SogokonMTCP22}
\bibfield{author}{\bibinfo{person}{Andrew Sogokon}, \bibinfo{person}{Stefan Mitsch}, \bibinfo{person}{Yong~Kiam Tan}, \bibinfo{person}{Katherine Cordwell}, {and} \bibinfo{person}{Andr{\'{e}} Platzer}.} \bibinfo{year}{2022}\natexlab{}.
\newblock \showarticletitle{Pegasus: Sound Continuous Invariant Generation}.
\newblock \bibinfo{journal}{\emph{Form. Methods Syst. Des.}} \bibinfo{volume}{58}, \bibinfo{number}{1} (\bibinfo{year}{2022}), \bibinfo{pages}{5--41}.
\newblock
\showISSN{0925-9856}
\urldef\tempurl%
\url{https://doi.org/10.1007/s10703-020-00355-z}
\showDOI{\tempurl}
\newblock
\shownote{Special issue for selected papers from FM'19}.


\bibitem[Sutton and Barto(2018)]%
        {sutton2018reinforcement}
\bibfield{author}{\bibinfo{person}{Richard~S Sutton} {and} \bibinfo{person}{Andrew~G Barto}.} \bibinfo{year}{2018}\natexlab{}.
\newblock \bibinfo{booktitle}{\emph{Reinforcement learning: An introduction}}.
\newblock \bibinfo{publisher}{MIT press}.
\newblock


\bibitem[Tessler et~al\mbox{.}(2019)]%
        {DBLP:conf/iclr/TesslerMM19}
\bibfield{author}{\bibinfo{person}{Chen Tessler}, \bibinfo{person}{Daniel~J. Mankowitz}, {and} \bibinfo{person}{Shie Mannor}.} \bibinfo{year}{2019}\natexlab{}.
\newblock \showarticletitle{Reward Constrained Policy Optimization}. In \bibinfo{booktitle}{\emph{7th International Conference on Learning Representations, {ICLR} 2019, New Orleans, LA, USA, May 6-9, 2019}}. \bibinfo{publisher}{OpenReview.net}.
\newblock
\urldef\tempurl%
\url{https://openreview.net/forum?id=SkfrvsA9FX}
\showURL{%
\tempurl}


\bibitem[Thumm and Althoff(2022)]%
        {DBLP:conf/icra/ThummA22}
\bibfield{author}{\bibinfo{person}{Jakob Thumm} {and} \bibinfo{person}{Matthias Althoff}.} \bibinfo{year}{2022}\natexlab{}.
\newblock \showarticletitle{Provably Safe Deep Reinforcement Learning for Robotic Manipulation in Human Environments}. In \bibinfo{booktitle}{\emph{2022 International Conference on Robotics and Automation, {ICRA} 2022, Philadelphia, PA, USA, May 23-27, 2022}}. \bibinfo{publisher}{{IEEE}}, \bibinfo{pages}{6344--6350}.
\newblock
\urldef\tempurl%
\url{https://doi.org/10.1109/ICRA46639.2022.9811698}
\showDOI{\tempurl}


\bibitem[Yang et~al\mbox{.}(2021)]%
        {DBLP:conf/aaai/YangSTS21}
\bibfield{author}{\bibinfo{person}{Qisong Yang}, \bibinfo{person}{Thiago~D. Sim{\~{a}}o}, \bibinfo{person}{Simon~H. Tindemans}, {and} \bibinfo{person}{Matthijs T.~J. Spaan}.} \bibinfo{year}{2021}\natexlab{}.
\newblock \showarticletitle{{WCSAC:} Worst-Case Soft Actor Critic for Safety-Constrained Reinforcement Learning}. In \bibinfo{booktitle}{\emph{Thirty-Fifth {AAAI} Conference on Artificial Intelligence, {AAAI} 2021, Thirty-Third Conference on Innovative Applications of Artificial Intelligence, {IAAI} 2021, The Eleventh Symposium on Educational Advances in Artificial Intelligence, {EAAI} 2021, Virtual Event, February 2-9, 2021}}. \bibinfo{publisher}{{AAAI} Press}, \bibinfo{pages}{10639--10646}.
\newblock
\urldef\tempurl%
\url{https://doi.org/10.1609/AAAI.V35I12.17272}
\showDOI{\tempurl}


\end{thebibliography}

\appendix

\section{Experimental Case Studies}\label{ap:experiments}

We implemented a compiler for our shield specification language, which is used in all our experiments and made available as part of our artifact.

\subsection{Experimental Setup}

\subsubsection{Reinforcement Learning Settings} We use the SAC~\cite{haarnoja2018soft} reinforcement learning algorithm in all our experiments, with a buffer length of 1,000,000, a learning rate of 0.003, and a discount factor of 0.99. When the proposed action is unsafe, we execute a fallback in the environment but store the proposed action in the buffer. For all meta-learning experiments, the safety budget for each episode is $10^{-7}$; otherwise, the budget for the whole training is $10^{-3}$. We run all our experiments on a machine with one Intel Xeon E5-2683 v3 56-Core Processor and \qty{126}{\giga\byte} of RAM. A typical run of 400,000 interaction steps lasts from $4$ to $6$ hours. All our figures and tables report means and standard deviations, as computed from 3 random seeds.

\subsubsection{State and Action Spaces} As defined in Definition~\ref{def:shielded-env}, the state and action spaces of shielded environments feature large and complex objects of dynamic size, which are not easily fed to neural networks. For simplicity, we represent states across experiments using a finite number of features, which are: the state of the underlying MDP, the current bound parameter values, the current step ID within the episode (normalized by the maximum episode length), the number of available observations of each type (normalized by the maximum episode length) and the remaining safety budget (normalized by the total training budget when learning in a fixed environment and by the episodic budget in the case of meta-learning). Restricted classes of inference policies are represented by neural networks on a case-by-case basis, encoding inference actions as vectors of features.

\subsubsection{Benchmarked Agents} All experiments benchmark at least three agents. The \emph{Adaptive} agent is a shielded agent that showcases our framework. The \emph{Unshielded} baseline is a vanilla SAC agent that is allowed to perform unsafe actions. The \emph{Non-adaptive} baseline does not use the inference module to update the controller monitor and therefore implements a non-adaptive shield in the style of \emph{Justified Speculative Control}~\cite{DBLP:conf/aaai/FultonP18}. To ensure a fair comparison, especially in meta-learning settings, the last two baselines are given access to an inference module and an inference policy similar to the one used by the \emph{Adaptive} agent. However, the resulting bound parameters are not used to update a controller monitor but instead simply included in the agent's state, ensuring that all controllers can benefit from the same knowledge. Numbers of crashes are reported in Table~\ref{tab:crash-stats}, while the training overhead is shown in Table~\ref{tab:overhead-stats}. From these tables, we can see that our method can achieve zero crashes with little overhead. Additionally, we show the effects of the budget for different settings in Table~\ref{tab:budget-stats}. Large safety budgets can make the agent less conservative, potentially leading to crashes when the budgets are significantly high.

\begin{table}
    \footnotesize
    \caption{Shielding overhead during training. For all case studies, we indicate (in order): the total shielding overhead as a percentage of training time, the shielding overhead per simulation step (in milliseconds), the shielding overhead due to running the inference module per simulation step (in milliseconds), the average number of aggregated measurements per simulation step and the monitor size (as the sum of the AST sizes of the generated ``\textsf{safe}'' and ``\textsf{fallback}'' programs). We present means and standard deviations across 3 random seeds. As can be seen, despite our implementation being written in unoptimized Python, shielding overhead is consistently below $15\%$. Also, shielding overhead is dominated by the cost of running the inference module, which in turn increases with the average number of aggregated measurements per time step.}
    \label{tab:overhead-stats}
    \begin{tabular}{lccccc}
        \toprule
        Case Study & Overhead  ($\%$) & Overhead (\unit{ms}) & Inference (\unit{ms})  &  Avg. Aggregation & Monitor Size \\
        \midrule
        \textsc{Sisyphean Train} & $14.30\pm0.01$ & $2.44\pm0.00$ & $2.38\pm0.01$ & $0.99\pm0.00$ & 85\\

        \midrule
        \textsc{Versatile Train} & $2.46\pm0.02$ & $0.38\pm0.00$ & $0.32\pm0.00$ & $0.07\pm0.00$ & 85\\

        \midrule
        \textsc{Crossing the River} & $3.79\pm0.80$ & $0.53\pm0.12$ & $0.48\pm0.12$ & $0.01\pm0.00$ & 84\\

        \midrule
        \textsc{Revisiting ACAS X} & $7.19\pm0.02$ & $1.02\pm0.01$ & $0.79\pm0.00$ & $0.53\pm0.01$ & 180\\

        \bottomrule
    \end{tabular}
\end{table}

\begin{table}
    \small
    \caption{Return at testing time (average value of the last 100 episodes during 10,000 evaluation steps) and number of crashes during training and testing time, for different case studies and under different settings. The number of training steps and the maximum episode length for each setting are 80,000/100; 400,000/100; 400,000/100; 400,000/50; and 400,000/40 respectively.}
    \label{tab:crash-stats}
    \begin{tabular}{lcccc}
        \toprule
        Setting & Method & Return & Crashes during training & Crashes during testing\\
        \midrule
        \multirow{4}{*}{\textsc{Sisyphean Train}} & Unshielded & $8.3\pm0.1$ & $110.7\pm21.5$ & $1.3\pm0.5$\\
 & Adaptive & $7.0\pm0.0$ & $0.0\pm0.0$ & $0.0\pm0.0$\\
 & Non-adaptive & $5.5\pm0.0$ & $0.0\pm0.0$ & $0.0\pm0.0$\\
 & Chebyshev & $5.4\pm0.0$ & $0.0\pm0.0$ & $0.0\pm0.0$\\

        \midrule
        \multirow{5}{*}{\makecell{\textsc{Versatile Train}\\($k/\sigma$ large)}} & Unshielded & $10.0\pm0.7$ & $102.7\pm45.0$ & $0.0\pm0.0$\\
 & Adaptive & $9.1\pm0.3$ & $0.0\pm0.0$ & $0.0\pm0.0$\\
 & Non-adaptive & $5.8\pm0.0$ & $0.0\pm0.0$ & $0.0\pm0.0$\\
 & Aggregate-1 & $9.6\pm0.1$ & $0.0\pm0.0$ & $0.0\pm0.0$\\
 & Aggregate-20 & $8.5\pm0.0$ & $0.0\pm0.0$ & $0.0\pm0.0$\\

        \midrule
        \multirow{5}{*}{\makecell{\textsc{Versatile Train}\\($k/\sigma$ small)}} & Unshielded & $9.5\pm1.2$ & $94.3\pm6.2$ & $0.7\pm0.5$\\
 & Adaptive & $9.1\pm0.2$ & $0.0\pm0.0$ & $0.0\pm0.0$\\
 & Non-adaptive & $5.8\pm0.1$ & $0.0\pm0.0$ & $0.0\pm0.0$\\
 & Aggregate-1 & $6.4\pm0.1$ & $0.0\pm0.0$ & $0.0\pm0.0$\\
 & Aggregate-20 & $8.7\pm0.2$ & $0.0\pm0.0$ & $0.0\pm0.0$\\

        \midrule
        \multirow{3}{*}{\textsc{Crossing the River}} & Unshielded & $3.9\pm0.4$ & $325.7\pm81.8$ & $0.7\pm0.5$\\
 & Adaptive & $4.2\pm0.7$ & $0.0\pm0.0$ & $0.0\pm0.0$\\
 & Non-adaptive & $-5.0\pm0.0$ & $0.0\pm0.0$ & $0.0\pm0.0$\\

        \midrule
        \multirow{3}{*}{\textsc{Revisiting ACAS X}} & Unshielded & $6.7\pm0.5$ & $987.0\pm66.4$ & $25.0\pm6.7$\\
 & Adaptive & $6.6\pm0.2$ & $0.0\pm0.0$ & $0.0\pm0.0$\\
 & Non-adaptive & $4.5\pm0.0$ & $0.0\pm0.0$ & $0.0\pm0.0$\\

        \bottomrule
    \end{tabular}
\end{table}

\begin{table}
  \small
  \caption{Return during testing and the number of crashes during both training and testing across all experiments, under varying safety budgets (budget per episode for meta-learning settings, and total budget for other scenarios). We present means and standard deviations across 3 random seeds. As can be seen, reducing the safety budget causes the agent to be more conservative. Setting an overly high safety budget results in crashes during training and/or testing.}
  \label{tab:budget-stats}
  \begin{tabular}{lcccc}
      \toprule
      Setting & Safety Budget & Return (testing) & Crashes (training) & Crashes (testing) \\
      \midrule
      \multirow{5}{*}{\textsc{Sisyphean Train}} & $10^{0}$ & $6.9\pm0.2$ & $0.0\pm0.0$ & $0.0\pm0.0$\\
 & $10^{-2}$ & $6.8\pm0.4$ & $0.0\pm0.0$ & $0.0\pm0.0$\\
 & $10^{-4}$ & $7.0\pm0.0$ & $0.0\pm0.0$ & $0.0\pm0.0$\\
 & $10^{-6}$ & $7.0\pm0.0$ & $0.0\pm0.0$ & $0.0\pm0.0$\\
 & $10^{-8}$ & $7.0\pm0.0$ & $0.0\pm0.0$ & $0.0\pm0.0$\\

      \midrule
      \multirow{5}{*}{\makecell{\textsc{Versatile Train}\\($k/\sigma$ large)}} & $10^{-1}$ & $9.4\pm0.2$ & $0.0\pm0.0$ & $0.0\pm0.0$\\
 & $10^{-3}$ & $9.1\pm0.2$ & $0.0\pm0.0$ & $0.0\pm0.0$\\
 & $10^{-5}$ & $9.3\pm0.1$ & $0.0\pm0.0$ & $0.0\pm0.0$\\
 & $10^{-7}$ & $9.1\pm0.3$ & $0.0\pm0.0$ & $0.0\pm0.0$\\
 & $10^{-9}$ & $9.3\pm0.2$ & $0.0\pm0.0$ & $0.0\pm0.0$\\

      \midrule
      \multirow{5}{*}{\makecell{\textsc{Versatile Train}\\($k/\sigma$ small)}} & $10^{-1}$ & $10.2\pm0.2$ & $124.0\pm38.3$ & $0.7\pm0.5$\\
 & $10^{-3}$ & $9.4\pm0.1$ & $0.3\pm0.5$ & $0.0\pm0.0$\\
 & $10^{-5}$ & $9.4\pm0.1$ & $0.0\pm0.0$ & $0.0\pm0.0$\\
 & $10^{-7}$ & $9.1\pm0.2$ & $0.0\pm0.0$ & $0.0\pm0.0$\\
 & $10^{-9}$ & $8.6\pm0.4$ & $0.0\pm0.0$ & $0.0\pm0.0$\\

      \midrule
      \multirow{5}{*}{\textsc{Crossing the River}} & $10^{-1}$ & $4.4\pm0.5$ & $0.3\pm0.5$ & $0.0\pm0.0$\\
 & $10^{-3}$ & $4.2\pm0.4$ & $0.0\pm0.0$ & $0.0\pm0.0$\\
 & $10^{-5}$ & $4.6\pm0.8$ & $0.0\pm0.0$ & $0.0\pm0.0$\\
 & $10^{-7}$ & $4.2\pm0.7$ & $0.0\pm0.0$ & $0.0\pm0.0$\\
 & $10^{-9}$ & $4.1\pm0.5$ & $0.0\pm0.0$ & $0.0\pm0.0$\\

      \midrule
      \multirow{5}{*}{\textsc{Revisiting ACAS X}} & $10^{-1}$ & $10.1\pm0.4$ & $0.0\pm0.0$ & $0.0\pm0.0$\\
 & $10^{-3}$ & $8.9\pm0.1$ & $0.0\pm0.0$ & $0.0\pm0.0$\\
 & $10^{-5}$ & $8.4\pm0.2$ & $0.0\pm0.0$ & $0.0\pm0.0$\\
 & $10^{-7}$ & $6.6\pm0.2$ & $0.0\pm0.0$ & $0.0\pm0.0$\\
 & $10^{-9}$ & $6.0\pm0.6$ & $0.0\pm0.0$ & $0.0\pm0.0$\\

      \bottomrule
  \end{tabular}
\end{table}

\subsection{Sisyphean Train}\label{ap:sisyphean-train}

\subsubsection{Environment} Assume that we have a train that is 1,000 meters away from the station with a speed of $v$, and we want it to stop within a distance range. We take one of two actions every second: accelerate with an acceleration $A$ of $\qty{4}{\meter\per\second\squared}$ or brake with an acceleration $-B$ of $\qty{-4}{\meter\per\second\squared}$. However, due to the incline of the railroad, the final acceleration includes an additional, unknown term $f(x)$, where $x$ is the position of the train. We assume that the shape of the railroad is $C\sin(\omega x + \phi)$. As a consequence, we derive:
\[
    f(x) = g\cdot\frac{C\omega\cos(\omega x + \phi)}{\sqrt{1 + C^2\omega^2\cos^2(\omega x + \phi)}}.
\]
We know that $f$ is $gC\omega^2$-Lipschitz and bounded between $\left[-gC\omega/\sqrt{1 + C^2\omega^2}, \, gC\omega/\sqrt{1 + C^2\omega^2}\right]$. At each step, the agent observes $f(x) + \eta$, where $\eta \sim \mathcal{U}(\qty{-0.3}{\meter\per\second\squared}, \qty{0.3}{\meter\per\second\squared})$.

Our goal is to stop within a distance range of $0$ to $\qty{100}{\meter}$ from the station. If we are inside this range with a speed of less than $\qty{1}{\meter\per\second}$, the episode ends with a reward of $10$; if we go over the station, the episode ends with a reward of $-10$; in other cases, we get a reward of $-0.05$ per step (i.e., a reward sequence of $-0.00$, $-0.05$, $-0.10$) to encourage achieving the goal as early as possible. All episodes stop at a maximum length of 100 control cycles.

\subsubsection{Modelling} For this example, we use a variant of the shield defined in Figure~\ref{fig:overview-train-local}, with uniform measurement noise instead of Gaussian noise, allowing us to also demonstrate the use of concentration bounds to implement the $\InvCCDF$ operator.

\subsubsection{Experimental Protocol} In this case study, we choose $k=\qty{0.0025}{\per\second\squared}$, $F=\qty{3}{\meter\per\second\squared}$, $v=\qty{30}{\meter\per\second}$, $C=\qty{0.22}{\meter}$, $\omega=\qty{0.00083}{\per\meter}$, and $\phi=\pi/2$. Because the shape of the railroad is $C\sin(\omega x + \phi)$, we assume that the station is at the highest position, and we want to go up along the railroad. Although the change rate of the railroad's height is small, our knowledge of the upper bound $F$ is very conservative. Thus, using information from the observations for inference is important.

We do not have analytical solutions for the $\InvCCDF$ function with linear combinations of uniformly distributed observations, so we need to use concentration bounds to do approximate inference. We consider two different bounds: Chebyshev and Hoeffding. From our analysis in Section~\ref{sec:eval-tail-expressions}, we expect Hoeffding bounds to work better. When not mentioned, we use them as the default method.

We benchmark our \emph{Adaptive} agent against \emph{Unshielded} and \emph{Non-adaptive} baselines. In addition, we test a variant of our adaptive agent that uses the Chebyshev inequality for inference instead of the Hoeffding inequality. For all methods, we use a hardcoded inference policy that performs aggregation every time 20 observations are available within a radius of 100 meters. Each time, a portion of the remaining safety budget is spent, which is equal to the number of steps since the last aggregation, divided by the total number of training steps. %

\subsubsection{Results} We run experiments for 80,000 steps with a maximum episode length of 100. Results are shown in Figure~\ref{fig:example1-1} and Table~\ref{tab:crash-stats}. We find that the adaptive method achieves a higher return than the non-adaptive method, while still benefitting from the same safety guarantees. The adaptive agent that uses Chebyshev bounds performs no better than the non-adaptive baseline, due to producing bounds that are too conservative to be useful.

\begin{figure}
    \includegraphics{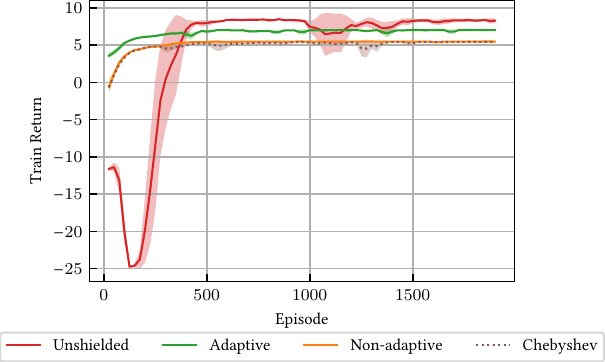}
    \caption{Returns for the \textsc{Sisyphean train} case study. \emph{Chebyshev} is a variation of \emph{Adaptive} where the Chebyshev bound is used in place of the default Hoeffding bound.}
    \label{fig:example1-1}
    \Description[]{}
\end{figure}

\subsection{Versatile Train}\label{ap:versatile-train}

\subsubsection{Environment} We use the same environment setting and shield as in the Sisyphean train case study, except that observations involve Gaussian measurement noise $\eta \sim \mathcal{N}(0, \sigma^2)$ instead of uniform noise. Additionally, we add a reward of 0.1 for non-terminal states with a non-zero remaining budget to encourage the agent only to use its budget when necessary.

\subsubsection{Experimental Protocol} To demonstrate meta-learning, we randomly sample the phase factor $\phi$ from $\mathcal{U}(0, 2\pi)$ at every episode, so as to represent different shapes of the railroad. We consider two settings: one where $k/\sigma$ is large and one where $k/\sigma$ is small. For each of them, we benchmark two adaptive agents with different hardcoded policies. The first one spends a fixed budget to aggregate all available observations every 20 steps, while the other spends a fixed budget to aggregate all available observations every single step. In addition, the \emph{Adaptive} agent \emph{learns} an inference policy that chooses at every time step whether or not to aggregate all available observations and how much budget to spend on it.

Regarding the environment parameters, we set $F=\qty{2.5}{\meter\per\second\squared}$ and $v=\qty{30}{\meter\per\second}$ for both settings. In the $k/\sigma$ large setting, we choose $k=\qty{0.002}{\per\second\squared}$, $\sigma=\qty{0.001}{\meter\per\second\squared}$, $C=\qty{0.19}{\meter}$, and $\omega=\qty{0.00080}{\per\meter}$. In the $k/\sigma$ small setting, we choose $k=\qty{0.00001}{\per\second\squared}$, $\sigma=\qty{1}{\meter\per\second\squared}$, $C=\qty{38.2}{\meter}$, and $\omega=\qty{0.0000040}{\per\meter}$. Note that the value of $k$ is determined by $C$ and $\omega$.

\subsubsection{Results} We run training for 400,000 steps, with a maximum episode length of 100. The results are shown in Figure~\ref{fig:example2-1} and Table~\ref{tab:crash-stats}. As expected, an inference policy that aggregates observations in \emph{small} batches works better in a setting with irregular tracks (large $k$) and low sensor noise (small $\sigma$), while a policy that aggregates observations in \emph{large} batches works better in settings with regular tracks and high sensor noise. However, in both cases, an inference policy can be \emph{learned} that matches the performance of its best-hardcoded counterpart.

We also perform a qualitative analysis of the policy learned by the adaptive agent, which is shown in Figure~\ref{fig:example2-2}. We show two cases, involving two different combinations of $k$ and $\sigma$. Regarding the choice of acceleration/braking, we can see that as long as the estimated braking distance is small enough for safety, the agent chooses to accelerate (unless it is near the terminal). For the inference part, in both cases, no budget is spent in the last few steps, because inference does not help when it needs to stop anyway. Also, there are some differences in the frequency of budget spending. When $k/\sigma$ is large, the change of the road can be large, which means the information we get from one position decays quickly; the noise level is small, which means we can get a good estimation without aggregating a lot of observations. All of these issues make a high inference frequency attractive, and vice versa. It is also demonstrated in the figure.

\subsection{Crossing the River}\label{ap:crossing-river}

\subsubsection{Environment} Assume that we are a robot moving in a two-dimensional plane. Since our speed is slow, we assume that we have direct control of our velocity at each control cycle. We want to cross the bridge over a river, at night, as quickly as possible and without falling into the river. However, we do not know the position of the bridge and must find it using our lamp. Once turned on, the lamp allows getting noisy measurements of the bridge position, but only when the robot is within a given radius of the bridge. Using the lamp also consumes our energy, so we only want to turn it on when necessary.

For simplicity, we assume that the bridge is a line segment at $x=0 \wedge y\in [y_b - W, y_b + W]$, whose center position $y_b$ is sampled from $\mathcal{U}(\qty{-10}{\meter}, \qty{10}{\meter})$ and whose width is $W=\qty{1}{\meter}$. The river is a line at $x=0$. Our initial position for both the $x$ and $y$ directions is sampled from $\mathcal{U}(\qty{-20}{\meter}, \qty{20}{\meter})$ (we resample if $x=0$). Initially, our velocity is 0. Every second, we can control our velocity for both directions, within a range of $[\qty{-2}{\meter\per\second}, \qty{2}{\meter\per\second}]$.

We can observe the bridge position $y_b$ from within a $\qty{5}{\meter}$ radius when our lamp is turned on, with Gaussian noise of standard deviation $\sigma \times |x|$ (i.e. we get better observations by moving closer to the river). Observations are only available every 2 seconds, due to sensor delay.

Our goal is to go to the other side of the river. If we succeed, the episode ends with a reward of 10; if we move into the river, the episode ends with a reward of -10; otherwise, we get a penalty of -0.1 at each time step, or -0.2 when the lamp is turned on. Episodes have a maximum length of 50.

\subsubsection{Modelling} We show a shield specification for this environment in Figure~\ref{fig:river-crossing}. Note that none of the details about \emph{when} observations become available need to be modeled, since this information is not soundness-relevant.

\subsubsection{Experimental Protocol} We compare the usual \emph{Unshielded}, \emph{Adaptive} and \emph{Non-adaptive} agents.

\subsubsection{Results} We run training for 400,000 steps. The results are shown in Figure~\ref{fig:example3-1} and Table~\ref{tab:crash-stats}. Unsurprisingly, the non-adaptive agent is incapable of ever crossing the bridge since it can never get enough certainty about its position. We also provide a qualitative analysis of the learned policy in Figure~\ref{fig:example3-2}. In two instances, we can see the robot approaching the river, turning on its lamp, and moving along the river until it gets some observations, after which it spends some safety budget for inference and crosses.

\begin{figure}
    \includegraphics{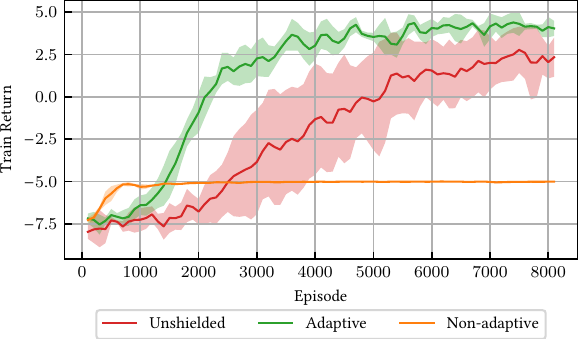}
    \caption{Returns for the \textsc{crossing the river} case study. All methods use a learned inference policy.}
    \label{fig:example3-1}
    \Description[]{}
\end{figure}

\newcommand{\rcyb}{y_\textsf{b}}
\newcommand{\rcybmin}{\underaccent{\bar}{y}_\textsf{b}}
\newcommand{\rcybmax}{{\bar y}_\textsf{b}}
\newcommand{\rcvx}{v_\textsf{x}}
\newcommand{\rcvy}{v_\textsf{y}}

\begin{figure}
  \begin{center}
  \begin{sllisting}
    \slconstant V, W, T, \sigma \\
    \slunknown \rcyb \\
    \slassume
      V > 0 \commasep
      W > 0 \commasep
      T > 0 \commasep
      \sigma > 0 \\
    \slbound
      \rcybmin : \rcybmin \le \rcyb \commasep
      \rcybmax : \rcybmax \ge \rcyb \\
    \slcontroller \\
      \slind \rcvx \slassign * \seq \rcvy \slassign * \seq \, ?(|\rcvx| \le V \land |\rcvy| \le V) \seq \\
      \slind ?(x \cdot (x + \rcvx T) > 0 \ \lor \  \rcvx = 0 \ \lor \  (\rcvx \ne 0 \land \rcybmax - W \,\le\, y + \rcvy \cdot (-x / \rcvx) \,\le\, \rcybmin + W)) \\
    \slplant
      t \slassign 0 \seq \dlode{x' = \rcvx, y' = \rcvy, t'=1 \dldom t \le T} \\
    \slsafe
      x = 0 \limply |y - \rcyb| \le W \\
    \slinvariant
      x = 0 \limply (\rcybmax - W \le y \le \rcybmin + W) \\
    \slnoise
      \eta \sim \slnormal(0, \sigma^2) \\
    \slobserve
      \omega = \rcyb - |x| \cdot \eta \\
    \slinfer
      \rcybmin, \rcybmax \slassign \slaggregate i: \omega_i \sland |x_i| \cdot \eta_i
  \end{sllisting}
  \end{center}
  \caption{A robot attempting to cross a river at night.}\label{fig:river-crossing}
  \Description[]{}
\end{figure}

\subsection{Revisiting ACAS X}\label{ap:acas}

\newcommand{\catm}{t_\textsf{m}}
\newcommand{\caAint}{A_\textsf{int}}
\newcommand{\caSigmaOne}{\sigma_\textsf{v}}
\newcommand{\caSigmaTwo}{\sigma_\textsf{h}}
\newcommand{\caEtaOne}{\eta_\textsf{v}}
\newcommand{\caEtaTwo}{\eta_\textsf{h}}
\newcommand{\caEtaThree}{\eta_\textsf{c}}
\newcommand{\caOmegaOne}{\omega_\textsf{v}}
\newcommand{\caOmegaTwo}{\omega_\textsf{h}}
\newcommand{\caOmegaThree}{\omega_\textsf{c}}
\newcommand{\cacmin}{\underaccent{\bar}{c}}
\newcommand{\cavint}{v_\textsf{int}}
\newcommand{\cahint}{h_\textsf{int}}
\newcommand{\cavintmax}{{\bar v}_\textsf{int}}
\newcommand{\cavintmin}{\underaccent{\bar}{v}_\textsf{int}}
\newcommand{\cahintmax}{{\bar h}_\textsf{int}}
\newcommand{\cahintmin}{\underaccent{\bar}{h}_\textsf{int}}
\newcommand{\cahzerointmax}{{\bar h}_{0,\textsf{int}}}
\newcommand{\cahzerointmin}{\underaccent{\bar}{h}_{0,\textsf{int}}}
\newcommand{\cahmintmax}{{\bar h}_\textsf{m,int}}
\newcommand{\cahmintmin}{\underaccent{\bar}{h}_\textsf{m,int}}
\newcommand{\cahnext}{h_\textsf{next}}
\newcommand{\cavnext}{v_\textsf{next}}
\newcommand{\catleftnext}{t_\textsf{left}}
\newcommand{\catinit}{t_0}

\begin{figure}
  \small
  \begin{center}
  \begin{sllisting}
    \renewcommand{\arraystretch}{1.2}
    \slconstant
      \catm,\, T,\, A,\, \caAint,\, R,\, V,\, H,\,
      \caSigmaOne,\, \caSigmaTwo,\, p \\
    \slunknown \cavint(*),\, \cahint(*),\, c \\
    \slassume \\
      \slind \catm \geq 0 ,\,
      T > 0 ,\,
      A > 0 ,\,
      \caAint > 0 ,\,
      R > 0 ,\,
      \caSigmaOne > 0 ,\,
      \caSigmaTwo > 0 ,\,
      0 < p < 1 ,\, \\
    \slind
      \lforall{t_1} \lforall{t_2}
        |\cahint(t_2) - \cahint(t_1) - \cavint(t_1) \cdot (t_2 - t_1)| \leq
        \caAint \cdot (t_2 - t_1) ^ 2 / 2, \\
    \slind
      \lforall{t_1} \lforall{t_2}
        |\cavint(t_2) - \cavint(t_1)| \le \caAint |t_2 - t_1|, \\
    \slind
      \lforall{t} (|\cavint(t)| \leq V \wedge |\cahint(t)| \leq H), \\
    \slind (c = 0 \lor c = 1) \land (c=1 \limply (\\
      \slind[2] (\cahint(0) > 0 \rightarrow \lforall{t} (t \leq \catm \rightarrow \cahint(\catm) \geq \cahint(t) + \cavint(t) \cdot (\catm - t))) \ \land \  \\
      \slind[2] (\cahint(0) < 0 \rightarrow \lforall{t} (t \leq \catm \rightarrow \cahint(\catm) \leq \cahint(t) + \cavint(t) \cdot (\catm - t)))))\\
    \slbound \\
      \slind \cacmin: \cacmin \le c, \, \\
      \slind
        \cavintmin: \cavintmin \le \cavint(t), \,
        \cavintmax: \cavintmax \ge \cavint(t), \,
        \cahintmin: \cahintmin \le \cahint(t), \,
        \cahintmax: \cahintmax \ge \cahint(t), \,
        \\
      \slind
        \cahzerointmin: \cahzerointmin \le \cahint(0), \,
        \cahzerointmax: \cahzerointmax \ge \cahint(0), \,
        \\
      \slind
        \cahmintmin: \cahmintmin \le \cahint(\catm), \,
        \cahmintmax: \cahmintmax \ge \cahint(\catm)
        \\
    \slcontroller \\
      \slind a \dlassign * \seq \, ?(|a| \leq A) \seq \\
      \slind \cahnext \dlassign h + v \cdot \min(T, \catm - t) + a \cdot \min(T, \catm - t) ^ 2 / 2 \seq \\
      \slind
        \cavnext \dlassign v + a \cdot \min(T, \catm - t) \seq
        \catleftnext \dlassign \max(\catm - t - T, 0) \seq \\
      \slind  ?(
        (\cahnext + \cavnext \cdot \catleftnext + A \cdot \catleftnext ^ 2 / 2 \,\geq\, \cahmintmax + R)
        \,\lor\,
        (\cahnext + \cavnext \cdot \catleftnext - A \cdot \catleftnext ^ 2 / 2 \,\leq\, \cahmintmin - R)
      )\\
    \slplant
      \catinit \dlassign t \seq \dlode{h' = v, v' = a, t' = 1 \dldom t \leq \min(\catinit + T, \catm)} \\
    \slsafe
      t = \catm \rightarrow |h - \cahint(t)| \geq R \\
    \slinvariant \\
      \slind t \geq 0 \land t \leq \catm \ \land \\
      \slind ((h + v \cdot (\catm - t) + A \cdot (\catm - t) ^ 2 / 2 \,\geq\, \cahmintmax + R) \,\lor\,
      (h + v \cdot (\catm - t) - A \cdot (\catm - t) ^ 2 / 2 \,\leq\, \cahmintmin - R)) \\
    \slnoise
      \caEtaOne \sim \slnormal(0, \caSigmaOne ^ 2) \commasep
      \caEtaTwo \sim \slnormal(0, \caSigmaTwo ^ 2) \commasep
      \caEtaThree \sim \slbernouilli(p) \\
    \slobserve
      \caOmegaOne = \cavint(t) - \caEtaOne \commasep
      \caOmegaTwo = \cahint(t) - \caEtaTwo \commasep
      \caOmegaThree = \min(1, c + \caEtaThree) \\
    \slinfer \\
      \slind \cavintmax \slassign V \seq  \cavintmin \slassign -V \seq  \cahintmax \slassign H \seq  \cahintmin \slassign -H \seq \\
      \slind  \cavintmax \slassign \slaggregate i: (\caOmegaOne)_i + \caAint \cdot |t - t_i| \sland (\caEtaOne)_i \seq \\
      \slind  \cavintmin \slassign \slaggregate i: (\caOmegaOne)_i - \caAint \cdot |t - t_i| \sland (\caEtaOne)_i \seq \\
      \slind \cahintmax \slassign \slaggregate i: (\caOmegaTwo)_i + (\cavintmax)_i \cdot (t - t_i) + \caAint \cdot (t - t_i) ^ 2 / 2 \sland (\caEtaTwo)_i \slwhen t_i \leq t \seq \\
      \slind \cahintmin \slassign \slaggregate i: (\caOmegaTwo)_i + (\cavintmin)_i \cdot (t - t_i) - \caAint \cdot (t - t_i) ^ 2 / 2 \sland (\caEtaTwo)_i \slwhen t_i \leq t \seq \\
      \slind \cahzerointmin \slassign \cahintmax + \cavintmin \cdot (0 - t) + \caAint \cdot (0 - t) ^ 2 / 2 \seq \\
      \slind \cahzerointmax \slassign \cahintmin + \cavintmax \cdot (0 - t) - \caAint \cdot (0 - t) ^ 2 / 2 \seq \\
      \slind \cacmin \slassign \slaggregate i: 0 \sland 1 - (\caEtaThree)_i  \slwhen (\caOmegaThree)_i = 1 \seq \\
      \slind \cahmintmax \slassign \cahintmax + \cavintmax \cdot (\catm - t) + \caAint \cdot (\catm - t) ^ 2 / 2 \slwhen \cacmin \leq 0 \land \cahzerointmax \geq 0 \seq \\
      \slind \cahmintmax \slassign \cahintmax + \cavintmax \cdot (\catm - t) \slwhen \cacmin > 0 \wedge \cahzerointmax < 0 \seq \\
      \slind \cahmintmin \slassign \cahintmin + \cavintmin \cdot (\catm - t) - \caAint \cdot (\catm - t) ^ 2 / 2 \slwhen \cacmin \leq 0 \land \cahzerointmin < 0 \seq \\
      \slind \cahmintmin \slassign \cahintmin + \cavintmin \cdot (\catm - t) \slwhen \cacmin > 0 \wedge \cahzerointmin > 0
  \end{sllisting}
  \end{center}
  \caption{An adaptive shield for an airborne collision avoidance system.}\label{fig:acas}
  \Description[]{}
\end{figure}

\subsubsection{Environment} Suppose we are in an aircraft, and there is an intruder moving towards us, also within the same, vertical 2D plane. We want to avoid a collision with minimal disturbance to our planned trajectory, by choosing a sequence of vertical accelerations based on sensor information. For simplicity, we assume that both aircraft have a known, constant horizontal velocity. This allows us to specify a \emph{horizontal meeting time} $\catm=\qty{40}{\second}$ (when the horizontal separation between the planes is 0). At that time, we require that the difference between the altitude of our plane and the intruder is at least $\qty{500}{\metre}$ to avoid a collision.

At time $t=0$, we assume $h=0$ and $v=0$. We control the vertical acceleration of the plane every second, and the maximum vertical acceleration $A$ in both directions is $\qty{3}{\metre\per\second\squared}$.

We denote the altitude and vertical velocity of the intruder at time $t$ as $\cahint(t)$ and $\cavint(t)$ respectively. These quantities specify the intruder's trajectory, which is unknown. In addition, we distinguish between two kinds of intruders: those that are ACAS-compliant and those that are not. Compliant intruders will never accelerate upward if their initial altitude is lower than ours and never accelerate downward if their initial altitude is higher. We can model compliance with a boolean unknown $c$ with values in $\{0, 1\}$ (1 means compliant).

We assume that for all $t$, $\cahint(t)$ is bounded between $\qty{-2000}{\metre}$ and $\qty{2000}{\metre}$, and $\cavint(t)$ is bounded between $\qty{-50}{\metre\per\second}$ and $\qty{50}{\metre\per\second}$. The maximum vertical acceleration for the intruder in both directions is also $\qty{3}{\metre\per\second\squared}$.

We get information about the intruder's position and velocity using our sensors. We can measure its altitude with noise $\mathcal{N}(0, (\qty{20}{\meter})^2)$ and its velocity with noise $\mathcal{N}(0, (\qty{2}{\meter\per\second})^2)$. Two kinds of observations can be used to infer compliance: one comes from making radio calls and the other comes from analyzing the current intruder's trajectory. Both provide independent evidence and are received at times $t=5$ and $t=10$ respectively. For simplicity, we also assume that both have a false-positive rate of $10^{-4}$ (i.e. falsely classifying a plane as compliant) and a false-negative rate (i.e. failing to classify a plane as compliant) of $0.1$. Since the safety budget for each episode is $10^{-7}$, we can only afford to claim compliance from two supportive observations.

In simulation, we assume that the initial altitude $\cahint(0)$ of the intruder is between $\qty{-500}{\meter}$ and $\qty{500}{\meter}$, and its initial vertical velocity $\cavint(0)$ is between $\qty{-2}{\meter\per\second}$ and $\qty{2}{\meter\per\second}$. The compliant variable $c$ is sampled from $\slbernouilli(0.5)$. Then, the intruder chooses a fixed acceleration for the whole simulation. If not compliant, it chooses an acceleration within $[(-V - \cavint(0)) / \catm, (V - \cavint(0)) / \catm] \cap [\qty{-3}{\metre\per\second\squared}, \qty{3}{\metre\per\second\squared}]$, where $V=\qty{50}{\metre\per\second}$ is the vertical velocity bound. If compliant, it takes the intersection of the previous bound with $[0, \infty)$ or $(-\infty, 0]$, depending on the sign of $\cahint(0)$.

The safety goal is to maintain a distance of at least $\qty{200}{\meter}$ between our plane and the intruder at the time of the meeting. We get a reward of 10 if successful, and -30 otherwise. During other time steps, we are penalized in proportion to our distance to the horizontal line $h=0$ (our initial intended trajectory) and rewarded if compliance is inferred, with a reward of $-0.2|h| + 0.2\cdot(\cacmin > 0)$.

\subsubsection{Shield Specification} We show our shield specification in Figure~\ref{fig:acas}.
The invariant dictates that one of these cases must hold: (1) if we accelerate with acceleration $A$, we will be higher than the maximum possible altitude of the intruder at the meeting time $\catm$ (with margin $R$) or (2) if we accelerate with acceleration $-A$, we will be lower than the minimum possible altitude of the intruder (with margin $R$ also). The controller ensures that the invariant will still hold at the next control cycle. Interestingly, the controller itself does \emph{not} directly depend on the lower bound $\cacmin$ for the compliance parameter. Instead, $\cacmin$ influences the controller indirectly, by affecting the inference strategy for other bound parameters such as $\cahmintmax$ that the controller depends on.

\subsubsection{Experimental Protocol} Here, we still compare the unshielded, adaptive, and non-adaptive agents. All of them rely on a learned inference policy that determines (1) how much budget should be spent on aggregating all the available compliance evidence (if any) and (2) how much budget should be spent aggregating all other available measurements (if any).

\subsubsection{Results}
We run training for 400,000 steps, with a maximum episode length of 40. Results are shown in Figure~\ref{fig:example4-1} and Table~\ref{tab:crash-stats}. The adaptive agent fares better than the non-adaptive one, and it also achieves stable rewards without crashes.

We also perform a qualitative analysis of the policy learned by the adaptive agent, which is shown in Figure~\ref{fig:example4-2}. We show two cases, involving a \emph{compliant} and \emph{non-compliant} agent. When successfully inferring compliance, our aircraft significantly reduces its uncertainty about the final intruder position and thus manages to limit the magnitude of its avoidance maneuver.

\begin{figure}
    \centering
    \includegraphics{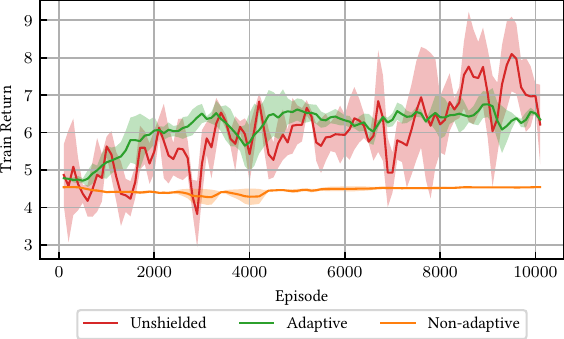}
    \caption{Returns for the \textsc{Revisiting ACAS X} case study. All methods learn an inference policy that specifies how much budget to spend on compliance observations and position/velocity observations.}
    \label{fig:example4-1}
    \Description[]{}
\end{figure}

\section{Proof Obligations for the Case Studies}
\label{ap:proof-obligations}
The proof obligations generated by our tool for each case study are shown in Figure~\ref{fig:targets-sisyphean-train}, Figure~\ref{fig:targets-crossing-the-river}, Figure~\ref{fig:targets-revisiting-ACAS-X-0}, and Figure~\ref{fig:targets-revisiting-ACAS-X}. Note that the proof obligations for the \emph{Sisyphean Train} case study and the \emph{Versatile Train} case study are the same, so we only consider the first one. Our tool automatically outputs these obligations in KeYmaera X's input language, allowing them to be directly imported for automatic or interactive proving.

All obligations are proved and all formal proofs can be found in this paper's accompanying artifact. Checking all of them using KeYmaera X currently takes about 2 minutes on a laptop. Among these 32 proof obligations, 27 can be discharged automatically or with trivial human guidance (such as providing a single term for instantiating an existential variable), while others need non-trivial guidance (marked as red).

\begin{figure}[p]
\begin{center}
\begin{sllisting}
\Init \Equiv A > 0 \, \land \, B > 0 \, \land \, T > 0 \, \land \, k > 0 \, \land \, \eta_r > 0 \, \land \, F < B \, \land \, A + F > 0 \, \land \, \\
    \slind[2] \, \lforall{x_1} \lforall{x_2} |f(x_1) - f(x_2)| \le k |x_1 - x_2| \, \land \, \lforall{x} (f(x) \le F \, \land \, f(x) \ge -A)\\
\Ctrl \Equiv y \dlassign \min(y, \TrainFxMax) \seq (?(x + v T + (A + F) \cdot T ^ 2 / 2 + (v + (A + F) \cdot T) ^ 2 / (2 (B - \min(F, y + \\
    \slind[2] \, \, k (v T + (A + F) \cdot T ^ 2 / 2) + k (v + (A + F) \cdot T) ^ 2 / (2 (B - F))))) \le 0) \seq a \dlassign A \cup a \dlassign -B)\\
\Plant \Equiv t \dlassign 0 \seq \dlode{x' = v \commasep v' = a + f(x) \commasep y' = k v \commasep t' = 1 \dldom t \le T \, \land \, v \ge 0}\\
\Inv \Equiv v \ge 0 \, \land \, x + v ^ 2 / (2 (B - \min(F, y + k v ^ 2 / (2 (B - F))))) \le 0 \, \land \, y \ge f(x)\\
\Safe \Equiv x \le 0\\
\pbmonotonicity\\
  \hspace{0.5em}\text{--}\hspace{0.4em}\TrainFxMax_1 \le \TrainFxMax_2 \limply f(x) \le \TrainFxMax_1 \limply f(x) \le \TrainFxMax_2\\
\pbmodel\\
  \hspace{0.5em}\text{--}\hspace{0.4em}\color{red}{\Init \, \land \, f(x) \le \TrainFxMax \, \land \, \Inv \limply \lboxI{\Ctrl \seq \Plant}\Inv}\\
\pbsafe\\
  \hspace{0.5em}\text{--}\hspace{0.4em}\Init \, \land \, \Inv \limply \Safe\\
\pbtotality\\
  \hspace{0.5em}\text{--}\hspace{0.4em}\Init \, \land \, f(x) \le \TrainFxMax \, \land \, \Inv \limply \ldiamond{\Ctrl}{\True}\\
\pbinference\\
  \hspace{0.5em}\text{--}\hspace{0.4em}\Init \, \land \, \Inv \, \land \, \TrainFxMax = F \limply f(x) \le \TrainFxMax\\
  \hspace{0.5em}\text{--}\hspace{0.4em}\Init \, \land \, \Inv \, \land \, \TrainFxMax = \omega_i + k |x - x_i| + \eta_i \, \land \, \omega_i = f(x_i) - \eta_i \limply f(x) \le \TrainFxMax\\
  \hspace{0.5em}\text{--}\hspace{0.4em}\Init \, \land \, \Inv \, \land \, \TrainFxMax = \TrainFxMax_i + k |x - x_i| \, \land \, f(x_i) \le \TrainFxMax_i \limply f(x) \le \TrainFxMax\\
\end{sllisting}
\end{center}
\caption{Proof obligations for the \emph{Sisyphean Train} case study. We write $\bdist{v}{b} \equiv v^2/2b$. The one obligation requiring nontrival guidance (proving the model's invariant) is shown in red. An informal proof of this formula is available in Appendix~\ref{ap:train-invariant-proof}.}
\label{fig:targets-sisyphean-train}
\Description[]{}
\end{figure}

\begin{figure}[p]
\begin{center}
\begin{sllisting}
\Init \Equiv V > 0 \, \land \, W > 0 \, \land \, T > 0 \, \land \, \sigma > 0\\
\Ctrl \Equiv \rcvx \dlassign * \seq \rcvy \dlassign * \seq l \dlassign * \seq ?(|\rcvx| \le V \, \land \, |\rcvy| \le V \, \land \, (x (x + \rcvx T) > 0 \, \lor \, \rcvx = 0 \, \land \, y + \rcvy T \ge\\
  \slind \, \, \, \, \, \, \, \, \, \, \,  \rcybmax - W \, \land \, y + \rcvy T \le \rcybmin + W \, \lor \, !\rcvx = 0 \, \land \, y + \rcvy (-x / \rcvx) \ge \rcybmax - W \, \land \, y + \rcvy (-x / \rcvx) \le \\
  \slind \, \, \, \, \, \, \, \, \, \, \, \rcybmin + W))\\
\Plant \Equiv t \dlassign 0 \seq \dlode{x' = \rcvx \commasep y' = \rcvy \commasep t' = 1 \dldom t \le T}\\
\Inv \Equiv x = 0 \limply y \ge \rcybmax - W \, \land \, y \le \rcybmin + W\\
\Safe \Equiv x = 0 \limply y \ge \rcyb - W \, \land \, y \le \rcyb + W\\
\pbmonotonicity\\
  \hspace{0.5em}\text{--}\hspace{0.4em}{\rcybmax}_1 \le {\rcybmax}_2 \, \land \, {\rcybmin}_1 \ge {\rcybmin}_2 \limply \rcyb \le {\rcybmax}_1 \, \land \, \rcyb \ge {\rcybmin}_1 \limply \rcyb \le {\rcybmax}_2 \, \land \, \rcyb \ge {\rcybmin}_2\\
\pbmodel\\
  \hspace{0.5em}\text{--}\hspace{0.4em}\Init \, \land \, \rcyb \le \rcybmax \, \land \, \rcyb \ge \rcybmin \, \land \, \Inv \limply \lboxI{\Ctrl \seq \Plant}\Inv\\
\pbsafe\\
  \hspace{0.5em}\text{--}\hspace{0.4em}\Init \, \land \, \rcyb \le \rcybmax \, \land \, \rcyb \ge \rcybmin \, \land \, \Inv \limply \Safe\\
\pbtotality\\
  \hspace{0.5em}\text{--}\hspace{0.4em}\Init \, \land \, \rcyb \le \rcybmax \, \land \, \rcyb \ge \rcybmin \, \land \, \Inv \limply \ldiamond{\Ctrl}{\True}\\
\pbinference\\
  \hspace{0.5em}\text{--}\hspace{0.4em}\Init \, \land \, \Inv \, \land \, \rcybmax = \omega_i + |x_i| \eta_i \, \land \, \omega_i = \rcyb - |x_i| \eta_i \limply \rcyb \le \rcybmax\\
  \hspace{0.5em}\text{--}\hspace{0.4em}\Init \, \land \, \Inv \, \land \, \rcybmin = \omega_i + |x_i| \eta_i \, \land \, \omega_i = \rcyb - |x_i| \eta_i \limply \rcyb \ge \rcybmin\\
\end{sllisting}
\end{center}
\caption{Proof obligations for the \emph{Crossing the River} case study. No obligation requires nontrivial guidance.}
\label{fig:targets-crossing-the-river}
\Description[]{}
\end{figure}

\begin{figure}[p]
\begin{center}
\begin{sllisting}
\Init \Equiv \catm \ge 0 \, \land \, T > 0 \, \land \, A > 0 \, \land \, \caAint > 0 \, \land \, R > 0 \, \land \, \caSigmaOne > 0 \, \land \, \caSigmaTwo > 0 \, \land \, p > 0 \, \land \, p < 1 \, \land \, \\
  \slind \, \, \, \, \, \, \, \, \, \, \, \lforall{s_1} \lforall{s_2} |\cahint(s_2) - \cahint(s_1) - \cavint(s_1) \cdot (s_2 - s_1)| \le \caAint (s_2 - s_1) ^ 2 / 2 \, \land \, \\
  \slind \, \, \, \, \, \, \, \, \, \, \, \lforall{x_1} \lforall{x_2} |\cavint(x_1) - \cavint(x_2)| \le \caAint |x_1 - x_2| \, \land \, \lforall{s} (|\cavint(s)| \le V \, \land \, |\cahint(s)| \le H) \, \land \, \\
  \slind \, \, \, \, \, \, \, \, \, \, \, (c = 0 \, \lor \, c = 1) \, \land \, (c = 1 \limply (\cahint(0) > 0 \limply \lforall{s} (s \le \catm \limply \cahint(\catm) \ge \cahint(s) + \cavint(s) \cdot\\
  \slind \, \, \, \, \, \, \, \, \, \, \,  (\catm - s))) \, \land \, (\cahint(0) < 0 \limply \lforall{s} (s \le \catm \limply \cahint(\catm) \le \cahint(s) + \cavint(s) \cdot (\catm - s))))\\
\Ctrl \Equiv a \dlassign * \seq \cahnext \dlassign h + v \min(T, \catm - t) + a \min(T, \catm - t) ^ 2 / 2 \seq \cavnext \dlassign v + a \min(T, \catm - t) \seq \\
  \slind \, \, \, \, \, \, \, \, \, \, \, \catleftnext \dlassign \max(\catm - t - T, 0) \seq ?(|a| \le A \, \land \, (\cahnext + \cavnext \catleftnext + A \catleftnext ^ 2 / 2 \ge \cahmintmax + R \, \lor \, \\
  \slind \, \, \, \, \, \, \, \, \, \, \, \cahnext + \cavnext \catleftnext - A \catleftnext ^ 2 / 2 \le \cahmintmin - R))\\
\Plant \Equiv \catinit \dlassign t \seq \dlode{h' = v \commasep v' = a \commasep t' = 1 \dldom t \le \min(\catinit + T, \catm)}\\
\Inv \Equiv t \ge 0 \, \land \, t \le \catm \, \land \, (h + v (\catm - t) + A (\catm - t) ^ 2 / 2 \ge \cahmintmax + R \, \lor \, \\
  \slind \, \, \, \, \, \, \, \, \, h + v (\catm - t) - A (\catm - t) ^ 2 / 2 \le \cahmintmin - R)\\
\Safe \Equiv t = \catm \limply |h - \cahint(t)| \ge R\\
\pbmonotonicity\\
  \hspace{0.5em}\text{--}\hspace{0.4em}{\cavintmax}_1 \le {\cavintmax}_2 \, \land \, {\cavintmin}_1 \ge {\cavintmin}_2 \, \land \, {\cahintmax}_1 \le {\cahintmax}_2 \, \land \, {\cahintmin}_1 \ge {\cahintmin}_2 \, \land \, {\cahzerointmax}_1 \le {\cahzerointmax}_2 \, \land \, \\
  \slind {\cahzerointmin}_1 \ge {\cahzerointmin}_2 \, \land \, {\cacmin}_1 \ge {\cacmin}_2 \, \land \, {\cahmintmax}_1 \le {\cahmintmax}_2 \, \land \, {\cahmintmin}_1 \ge {\cahmintmin}_2 \limply \cavint(t) \le {\cavintmax}_1 \, \land \, \\
  \slind \cavint(t) \ge {\cavintmin}_1 \, \land \, \cahint(t) \le {\cahintmax}_1 \, \land \, \cahint(t) \ge {\cahintmin}_1 \, \land \, \cahint(0) \le {\cahzerointmax}_1 \, \land \, \cahint(0) \ge {\cahzerointmin}_1 \, \land \, \\
  \slind c \ge {\cacmin}_1 \, \land \, \cahint(\catm) \le {\cahmintmax}_1 \, \land \, \cahint(\catm) \ge {\cahmintmin}_1 \limply \cavint(t) \le {\cavintmax}_2 \, \land \, \cavint(t) \ge {\cavintmin}_2 \, \land \, \\
  \slind \cahint(t) \le {\cahintmax}_2 \, \land \, \cahint(t) \ge {\cahintmin}_2 \, \land \, \cahint(0) \le {\cahzerointmax}_2 \, \land \, \cahint(0) \ge {\cahzerointmin}_2 \, \land \, c \ge {\cacmin}_2 \, \land \, \\
  \slind \cahint(\catm) \le {\cahmintmax}_2 \, \land \, \cahint(\catm) \ge {\cahmintmin}_2\\
\pbmodel\\
  \hspace{0.5em}\text{--}\hspace{0.4em}\color{red}{\Init \, \land \, \cavint(t) \le \cavintmax \, \land \, \cavint(t) \ge \cavintmin \, \land \, \cahint(t) \le \cahintmax \, \land \, \cahint(t) \ge \cahintmin \, \land \, \cahint(0) \le \cahzerointmax \, \land \,} \\
  \slind \color{red}{\cahint(0) \ge \cahzerointmin \, \land \, c \ge \cacmin \, \land \, \cahint(\catm) \le \cahmintmax \, \land \, \cahint(\catm) \ge \cahmintmin \, \land \, \Inv \limply \lboxI{\Ctrl \seq \Plant}\Inv}\\
\pbsafe\\
  \hspace{0.5em}\text{--}\hspace{0.4em}\Init \, \land \, \cahint(0) \le \cahzerointmax \, \land \, \cahint(0) \ge \cahzerointmin \, \land \, c \ge \cacmin \, \land \, \cahint(\catm) \le \cahmintmax \, \land \, \cahint(\catm) \ge \cahmintmin \, \land \, \Inv \\
  \slind \limply \Safe\\
\pbtotality\\
  \hspace{0.5em}\text{--}\hspace{0.4em}\color{red}{\Init \, \land \, \cavint(t) \le \cavintmax \, \land \, \cavint(t) \ge \cavintmin \, \land \, \cahint(t) \le \cahintmax \, \land \, \cahint(t) \ge \cahintmin \, \land \, \cahint(0) \le \cahzerointmax \, \land \,} \\
  \slind \color{red}{\cahint(0) \ge \cahzerointmin \, \land \, c \ge \cacmin \, \land \, \cahint(\catm) \le \cahmintmax \, \land \, \cahint(\catm) \ge \cahmintmin \, \land \, \Inv \limply \ldiamond{\Ctrl}{\True}}\\
\end{sllisting}
\end{center}
\caption{Proof obligations for the \emph{Revisiting ACAS X} case study (2 of 2). Obligations requiring nontrival guidance are marked in red.}
\label{fig:targets-revisiting-ACAS-X-0}
\Description[]{}
\end{figure}

\begin{figure}[p]
\begin{center}
\begin{sllisting}
\pbinference\\
  \hspace{0.5em}\text{--}\hspace{0.4em}\Init \, \land \, \Inv \, \land \, \cavintmax = V \limply \cavint(t) \le \cavintmax\\
  \hspace{0.5em}\text{--}\hspace{0.4em}\Init \, \land \, \Inv \, \land \, \cavintmin = -V \limply \cavint(t) \ge \cavintmin\\
  \hspace{0.5em}\text{--}\hspace{0.4em}\Init \, \land \, \Inv \, \land \, \cahintmax = H \limply \cahint(t) \le \cahintmax\\
  \hspace{0.5em}\text{--}\hspace{0.4em}\Init \, \land \, \Inv \, \land \, \cahintmin = -H \limply \cahint(t) \ge \cahintmin\\
  \hspace{0.5em}\text{--}\hspace{0.4em}\Init \, \land \, \Inv \, \land \, \cavintmax = {\caOmegaOne}_i + \caAint |t - t_i| + {\caEtaOne}_i \, \land \, {\caOmegaOne}_i = \cavint(t_i) - {\caEtaOne}_i \limply \cavint(t) \le \cavintmax\\
  \hspace{0.5em}\text{--}\hspace{0.4em}\Init \, \land \, \Inv \, \land \, \cavintmin = {\caOmegaOne}_i - \caAint |t - t_i| + {\caEtaOne}_i \, \land \, {\caOmegaOne}_i = \cavint(t_i) - {\caEtaOne}_i \limply \cavint(t) \ge \cavintmin\\
  \hspace{0.5em}\text{--}\hspace{0.4em}\Init \, \land \, \Inv \, \land \, \cahintmax = {\caOmegaTwo}_i + {\cavintmax}_i (t - t_i) + \caAint (t - t_i) ^ 2 / 2 + {\caEtaTwo}_i \, \land \, {\caOmegaTwo}_i = \cahint(t_i) - {\caEtaTwo}_i \, \land \, \\
  \slind \cavint(t_i) \le {\cavintmax}_i \, \land \, t_i \le t \limply \cahint(t) \le \cahintmax\\
  \hspace{0.5em}\text{--}\hspace{0.4em}\Init \, \land \, \Inv \, \land \, \cahintmin = {\caOmegaTwo}_i + {\cavintmin}_i (t - t_i) - \caAint (t - t_i) ^ 2 / 2 + {\caEtaTwo}_i \, \land \, {\caOmegaTwo}_i = \cahint(t_i) - {\caEtaTwo}_i \, \land \, \\
  \slind \cavint(t_i) \ge {\cavintmin}_i \, \land \, t_i \le t \limply \cahint(t) \ge \cahintmin\\
  \hspace{0.5em}\text{--}\hspace{0.4em}\Init \, \land \, \Inv \, \land \, \cahzerointmax = \cahintmax + \cavintmin (0 - t) + \caAint (0 - t) ^ 2 / 2 \, \land \, \cahint(t) \le \cahintmax \, \land \, \cavint(t) \ge \cavintmin \\
  \slind \limply \cahint(0) \le \cahzerointmax\\
  \hspace{0.5em}\text{--}\hspace{0.4em}\Init \, \land \, \Inv \, \land \, \cahzerointmin = \cahintmin + \cavintmax (0 - t) - \caAint (0 - t) ^ 2 / 2 \, \land \, \cahint(t) \ge \cahintmin \, \land \, \cavint(t) \le \cavintmax \\
  \slind \limply \cahint(0) \ge \cahzerointmin\\
  \hspace{0.5em}\text{--}\hspace{0.4em}\Init \, \land \, \Inv \, \land \, \cacmin = 0 + (1 - {\caEtaThree}_i) \, \land \, {\caOmegaThree}_i = \min(1, c + {\caEtaThree}_i) \, \land \, {\caOmegaThree}_i = 1 \limply c \ge \cacmin\\
  \hspace{0.5em}\text{--}\hspace{0.4em}\Init \, \land \, \Inv \, \land \, \cahmintmax = \cahintmax + \cavintmax (\catm - t) + \caAint (\catm - t) ^ 2 / 2 \, \land \, c \ge \cacmin \, \land \, \cahint(0) \le \cahzerointmax \, \land \, \\
  \slind \cahint(t) \le \cahintmax \, \land \, \cavint(t) \le \cavintmax \, \land \, (\cacmin \le 0 \, \lor \, \cahzerointmax \ge 0) \limply \cahint(\catm) \le \cahmintmax\\
  \hspace{0.5em}\text{--}\hspace{0.4em}\color{red}{\Init \, \land \, \Inv \, \land \, \cahmintmax = \cahintmax + \cavintmax (\catm - t) \, \land \, c \ge \cacmin \, \land \, \cahint(0) \le \cahzerointmax \, \land \, \cahint(t) \le \cahintmax \, \land \,} \\
  \slind \color{red}{\cavint(t) \le \cavintmax \, \land \, \cacmin > 0 \, \land \, \cahzerointmax < 0 \limply \cahint(\catm) \le \cahmintmax}\\
  \hspace{0.5em}\text{--}\hspace{0.4em}\Init \, \land \, \Inv \, \land \, \cahmintmin = \cahintmin + \cavintmin (\catm - t) - \caAint (\catm - t) ^ 2 / 2 \, \land \, c \ge \cacmin \, \land \, \cahint(0) \ge \cahzerointmin \, \land \, \\
  \slind \cahint(t) \ge \cahintmin \, \land \, \cavint(t) \ge \cavintmin \, \land \, (\cacmin \le 0 \, \lor \, \cahzerointmin \le 0) \limply \cahint(\catm) \ge \cahmintmin\\
  \hspace{0.5em}\text{--}\hspace{0.4em}\color{red}{\Init \, \land \, \Inv \, \land \, \cahmintmin = \cahintmin + \cavintmin (\catm - t) \, \land \, c \ge \cacmin \, \land \, \cahint(0) \ge \cahzerointmin \, \land \, \cahint(t) \ge \cahintmin \, \land \, }\\
  \slind \color{red}{\cavint(t) \ge \cavintmin \, \land \, \cacmin > 0 \, \land \, \cahzerointmin > 0 \limply \cahint(\catm) \ge \cahmintmin}\\
\end{sllisting}
\end{center}
\caption{Proof obligations for the \emph{Revisiting ACAS X} case study (2 of 2). Obligations requiring nontrival guidance are marked in red.}
\label{fig:targets-revisiting-ACAS-X}
\Description[]{}
\end{figure}

\clearpage

\section{Informal Proof of the Invariant for the \emph{Sisyphean Train} Case Study}\label{ap:train-invariant-proof}
We provide an informal proof of the invariant of our train shield example from Figure~\ref{fig:overview-train-local}. We use the same notation as Figure~\ref{fig:targets-sisyphean-train}. Furthermore, we denote
{\small\[Q \equiv x + vT + (A + F)T^2/2 + \bdist{v+(A+F)T}{B - \min(F, \, y + k(vT + (A + F)T^2/2) + k \cdot \bdist{v + (A+F)T}{B-F})} \le e\]} and so we have $\Ctrl \equiv (y \dlassign \min(y, \TrainFxMax) \seq (((?Q \seq a \dlassign A) \cup (a \dlassign -B)))$. We want to prove that \[\vDash \ \Init \, \land \, f(x) \le \TrainFxMax \, \land \, \Inv \limply [\Ctrl \seq \Plant] \Inv.\] Such a proof can be split into an \emph{acceleration case}: \[\vDash \ \Init \, \land \, f(x) \le \TrainFxMax \, \land \, \Inv \limply [y \dlassign \min(y, \TrainFxMax) \seq a:=-B \seq \Plant] \Inv\] and a \emph{braking case}: \[\vDash \ \Init \, \land \, f(x) \le \TrainFxMax \, \land \, \Inv \limply [y \dlassign \min(y, \TrainFxMax) \seq ?Q \seq a:=A \seq \Plant] \Inv\] Note that being in the domain constraint of the ODE, $v \ge 0$, it is trivially preserved and we focus on establishing the preservation of $y \ge f(x)$ and $x + \bdist{v}{B - \min(F, y + k \cdot \bdist{v}{B-F})} \le e$.

\subsection{Proof of the \emph{Braking Case}}
We denote $\alpha \,\equiv\, (y \dlassign \min(y, \TrainFxMax) \seq a:=-B \seq \Plant)$, and we proceed in four steps:
\begin{enumerate}
    \item $\Fvalid \Init \, \land \, f(x) \le \TrainFxMax \, \land \, \Inv \limply [\alpha] (y \ge f(x))$;
    \item $\Fvalid \Init \, \land \, f(x) \le \TrainFxMax \, \land \, \Inv \, \land \, y = y_0 \, \land \, v = v_0 \limply [\alpha] (y + k \cdot \bdist{v}{B-F} \le y_0 + k \cdot \bdist{v_0}{B-F})$;
    \item $\Fvalid \Init \, \land \, f(x) \le \TrainFxMax \, \land \, \Inv \, \land \, y = y_0 \, \land \, v = v_0 \limply [\alpha] (x + \bdist{v}{B - \min(F, y_0 + k \cdot \bdist{v_0}{B-F})} \le e)$;
    \item $\Fvalid \Init \, \land \, f(x) \le \TrainFxMax \, \land \, \Inv \limply [\alpha] (x + \bdist{v}{B - \min(F, y + k \cdot \bdist{v}{B-F})} \le e)$.
\end{enumerate}

\paragraph{Step 1} Before the execution of $\alpha$, we know that $y \ge f(x)$ and $\TrainFxMax \ge f(x)$. After the execution of $y \dlassign \min(y, \TrainFxMax)$, we still have $y \ge f(x)$. To prove that it is also preserved by the ODE, we simply use the \emph{differential invariant} rule (dI) from $\dL$\footnote{The dI rule from $\dL$ allows proving that an inequality is preserved by an ODE if its \emph{derivative} has the right sign throughout: \[
\frac{\vdash [X' \dlassign G(x)](A)' \ge (B)'}{A \ge B \vdash [X'=G(X) \, \& \, P] \, A \ge B}.
\]}. Thus, we must prove $y' \ge (f(x))'$:
\begin{itemize}
    \item[-] It is enough to prove: $k v \ge x' \times f'(x)$.
    \item[-] It is enough to prove: $k v \ge v \times f'(x)$.
    \item[-] Since $v \ge 0$, it is enough to prove that $k \ge f'(x)$, which holds since $f$ is $k$-Lipschitz.
\end{itemize}

\paragraph{Step 2} We also use dI and prove that $(y + k \cdot \bdist{v}{B-F})' \le 0$.
\begin{itemize}
    \item[-] It is enough to prove: $k v + (k v v') / (B - F) \le 0$.
    \item[-] It is enough to prove: $k v (1 + (v')/(B-F)) \le 0$.
    \item[-] It is enough to prove: $v' \le -(B-F)$.
    \item[-] It is enough to prove that $-B + f(x) \le -(B-F)$, which holds since $f$ is $F$-bounded.
\end{itemize}

\paragraph{Step 3} We also use dI and prove that $(x + \bdist{v}{B - \min(F, y_0 + k \cdot \bdist{v_0}{B-F})})' \le e$:
\begin{itemize}
    \item[-] It is enough to prove: $v' \le -B + \min(F, y_0 + k \cdot \bdist{v_0}{B-F})$.
    \item[-] It is enough to prove: $-B + f(x) \le -B + \min(F, y_0 + k \cdot \bdist{v_0}{B-F})$.
    \item[-] It is enough to prove: $f(x) \le \min(F, y_0 + k \cdot \bdist{v_0}{B-F})$.
    \item[-] It is enough to prove: $\min(F, y) \le \min(F, y_0 + k \cdot \bdist{v_0}{B-F})$, which is a consequence of step 2 since $k \cdot \bdist{v}{B-F} \ge 0$.
\end{itemize}

\paragraph{Step 4} The target is equivalent to: \[\Fvalid \Init \, \land \, f(x) \le \TrainFxMax \, \land \, \Inv \, \land \, y = y_0 \, \land \, v = v_0 \limply [\alpha] (x + \bdist{v}{B - \min(F, y + k \cdot \bdist{v}{B-F})} \le e).\]
\begin{itemize}
    \item[-] It is enough to prove: $x + \bdist{v}{B - \min(F, y + k \cdot \bdist{v}{B-F})} \le x + \bdist{v}{B - \min(F, y_0 + k \cdot \bdist{v_0}{B-F})}$, because of step 3.
    \item[-] It is enough to prove: $B - \min(F, y + k \cdot \bdist{v}{B-F}) \ge B - \min(F, y_0 + k \cdot \bdist{v_0}{B-F})$.
    \item[-] It is enough to prove: $y + k \cdot \bdist{v}{B-F} \le y_0 + k \cdot \bdist{v_0}{B-F}$, which is true because of step 2.
\end{itemize}

\subsection{Proof of the \emph{Acceleration Case}}
We denote $\alpha \equiv (y \dlassign \min(y, \TrainFxMax) \seq ?Q \seq a:=A \seq \Plant)$, and we proceed in four steps:
\begin{enumerate}
    \item $\Fvalid \Init \, \land \, f(x) \le \TrainFxMax \, \land \, \Inv \limply [\alpha] (y \ge f(x))$;
    \item $\Fvalid \Init \, \land \, f(x) \le \TrainFxMax \, \land \, \Inv \, \land \, y = y_0 \, \land \, v = v_0 \limply [\alpha] (y + k \cdot \bdist{v}{B-F} \ge y_0 + k \cdot \bdist{v_0}{B-F})$;
    \item $\Fvalid \Init \, \land \, f(x) \le \TrainFxMax \, \land \, \Inv \, \land \, v = v_0 \, \land \, y = y_0 \, \land \, v = v_0 \limply [\alpha] (v \le v_0+(A+F)t' \, \land \, x \le x_0+v_0t+((A+F)t^2)/2 \, \land \, y \le y_0+k(v_0t+((A+F)t^2)/2))$;
    \item $\Fvalid \Init \, \land \, f(x) \le \TrainFxMax \, \land \, \Inv \, \land \, Q \limply [\alpha] (x + \bdist{v}{B - \min(F, y + k \cdot \bdist{v}{B-F})} \le e)$.
\end{enumerate}

\paragraph{Step 1} Before the execution of $\alpha$, we know that $y \ge f(x)$ and $\TrainFxMax \ge f(x)$. After the execution of $y \dlassign \min(y, \TrainFxMax)$, we still have $y \ge f(x)$. To prove that it is also preserved by the ODE, we simply use dI (differential invariant). We must prove $y' \ge (f(x))'$.
\begin{itemize}
    \item[-] It is enough to prove: $k v \ge x' \times f'(x)$.
    \item[-] It is enough to prove: $k v \ge v \times f'(x)$.
    \item[-] Since $v \ge 0$, it is enough to prove that $k \ge f'(x)$, which holds since $f$ is $k$-Lipschitz.
\end{itemize}

\paragraph{Step 2} We also use dI and prove that $(y + k \cdot \bdist{v}{B-F})' \ge 0$.
\begin{itemize}
    \item[-] It is enough to prove: $k v + (k v v') / (B - F) \ge 0$.
    \item[-] It is enough to prove: $k v (1 + (v')/(B-F)) \ge 0$.
    \item[-] It is enough to prove: $v' \ge -(B-F)$.
    \item[-] It is enough to prove: $A + f(x) \ge -(B-F)$, which holds since $A+f(x) \ge 0 \ge -(B-F)$.
\end{itemize}

\paragraph{Step 3} We also use dI.
\begin{itemize}
    \item[-] For $v \le v_0+(A+F)t$, it is enough to prove that $v' \le A+F$, or $A + f(x) \le A+F$, which holds since $f$ is $F$-bounded.
    \item[-] For $x \le x_0+v_0t+((A+F)t^2)/2$, it is enough to prove that $x' \le v_0+(A+F)t$, or $v \le v_0 +(A+F)t$, which is already proved.
    \item[-] It is enough to prove: $v' \ge -(B-F)$.
    \item[-] Similarly, we can prove that $y \le y_0+k(v_0t+((A+F)t^2)/2)$.
\end{itemize}

\paragraph{Step 4} From step 3, we have upper bounds for $x$, $y$, and $v$. By montonicity, it is enough to prove $\Fvalid \Inv \, \land \, v \ge 0 \, \land \, v=v_0 \, \land \, x=x_0 \, \land \, y=y_0 \limply [\alpha] (x_0 + v_0T + (A + F)T^2/2 + \bdist{v_0+(A+F)T}{B - \min(F, \, y_0 + k(v_0T + (A + F)T^2/2) + k \cdot \bdist{v_0 + (A+F)T}{B-F})} \le e)$. This holds by definition of $?Q$.

\section{Deriving controller monitors and fallback policies}\label{ap:ctrl-monitor}

\subsection{Deriving Control Action Spaces}
\label{ap:defining-action-ctrl}
We formally demonstrate how nondeterministic $\dL$ controllers induce an action space $\CtrlActionSem{\Ctrl}$ and define $\CtrlExeSem{\Ctrl}(s, a)$ in Figure~\ref{fig:ctrl-fun-sem}. To provide a more advanced example that uses those definitions, let us consider the following controller:
\[
\Ctrl \ \equiv \ (x := * \,;\, y := * \,;\, ?x \ge y) \,\cup\, (x := 0 \,;\, ((y := * \,;\, ?y \ge 0) \,\cup\, y := -1)).
\]

Then, we have $\CtrlActionSem{\Ctrl} \equiv (\Reals \times \Reals \times \UnitSet) \uplus (\UnitSet \times ((\Reals \times \UnitSet) \uplus \UnitSet)) \!\IsomorphTo \Reals^2 \uplus \Reals \uplus \UnitSet$. Moreover, consider action $a = \SumRightInj{(\Unit, \SumLeftInj{(8, \Unit)})}$, which encodes the path: ``take the right branch, do the assignment, take the left branch, pick value $8$ for $y$ and do the test''. Then, we have: \(\CtrlExeSem{\Ctrl}((3, 7), a) = (0, 8). \)

\begin{figure}
  \begin{minipage}{0.4\textwidth}
  \begin{align*}
      \CtrlActionSem{\alpha \cup \beta} &\ \equiv \ \CtrlActionSem{\alpha} \uplus \CtrlActionSem{\beta} \\
      \CtrlActionSem{\alpha ; \beta} &\ \equiv \ \CtrlActionSem{\alpha} \times \CtrlActionSem{\beta} \\
      \CtrlActionSem{x := *} &\ \equiv \ \Reals \\
      \CtrlActionSem{x := e} &\ \equiv \ \UnitSet \\
      \CtrlActionSem{?Q} &\ \equiv \ \UnitSet \\
  \end{align*}
  \end{minipage}
  \begin{minipage}{0.5\textwidth}
  \begin{align*}
      \CtrlExeSem{\alpha \cup \beta}(s, \SumLeftInj{a}) & \ \equiv \ \CtrlExeSem{\alpha}(s, a) \\
      \CtrlExeSem{\alpha \cup \beta}(s, \SumRightInj{a}) & \ \equiv \ \CtrlExeSem{\beta}(s, a) \\
      \CtrlExeSem{\alpha ; \beta}(s, (a_1, a_2)) & \ \equiv \ \CtrlExeSem{\beta}(\CtrlExeSem{\alpha}(s, a_1), a_2) \\
      \CtrlExeSem{x \dlassign *}(s, v) & \ \equiv \ \UpdateMap{s}{x}{v} \\
      \CtrlExeSem{x \dlassign e}(s, \Unit) & \ \equiv \ \UpdateMap{s}{x}{\sem{e}(\EmptyValuation, s)} \\
      \CtrlExeSem{?Q}(s, \Unit) & \ \equiv \ s
  \end{align*}
  \end{minipage}
  \caption{A formal definition of $\CtrlActionSem{\Ctrl}$ and $\CtrlExeSem{\Ctrl}$, as introduced in Section~\ref{sec:safe-rl-via-jsc}. Note that these definitions do not feature loops and differential equations, as these are not allowed to occur in controllers.}
  \Description[]{}
  \label{fig:ctrl-fun-sem}
\end{figure}

\subsection{Deriving Controller Monitors}
\label{ap:extracting-cm}
Figure~\ref{fig:deriving-cm} provides more details on how controller monitors can be automatically extracted from nondeterministic controllers. For pedagogical reasons, we take a slightly different approach here than is used in the original ModelPlex paper~\cite{DBLP:journals/fmsd/MitschP16}, but both approaches are similar in spirit. Note that both $\CtrlMonitorSem{\alpha}$ and $\CtrlExeSem{\alpha}$ are computable whenever $\alpha$ is a hybrid program that is free of function symbols, loops, differential equations, modalities, and quantifiers.

\begin{figure}
  \begin{align*}
  \CtrlMonitorSem{\alpha \cup \beta}(s, \SumLeftInj{a}) & \ \equiv \ \CtrlMonitorSem{\alpha}(s, a) \\
  \CtrlMonitorSem{\alpha \cup \beta}(s, \SumRightInj{a}) & \ \equiv \ \CtrlMonitorSem{\beta}(s, a) \\
  \CtrlMonitorSem{\alpha \seqI \beta}(s, (a_1, a_2)) & \ \equiv \ \BoolAnd{\CtrlMonitorSem{\alpha}(s, a_1)}{\CtrlMonitorSem{\beta}(\CtrlExeSem{\alpha}(s), a_2)} \\
  \CtrlMonitorSem{x \dlassign *}(s, \Unit) & \ \equiv \ \BoolTrue \\
  \CtrlMonitorSem{x \dlassign e}(s, \Unit) & \ \equiv \ \BoolTrue \\
  \CtrlMonitorSem{?Q}(s, \Unit) & \ \equiv \ (\EmptyValuation, s) \in \sem{Q}
  \end{align*}
  \caption{Extracting a controller monitor from a nondeterministic $\dL$ controller.}
  \Description{The decomposition of controller monitors.}
  \label{fig:deriving-cm}
\end{figure}

\subsection{Existence of a Computable Fallback Policy}\label{ap:generating-fallbacks}
Our shielding algorithm requires a fallback policy $\CtrlFallbackSem{\Ctrl}$ to pick actions whenever the desired action is rejected by the controller monitor. An explicit fallback policy is typically provided by users on a case-by-case basis (see discussion in Section~\ref{sec:safe-rl-via-jsc}), but fallback actions can also be generated systematically and automatically with the help of an SMT solver. We demonstrate this point by studying an example with all features of the general case. Indeed, consider the following controller:
\[
  \alpha \ \equiv \ ((x \dlassign * \,;\, y \dlassign v \cdot x) \,\cup\, (y \dlassign w)) \seq ?(y \ge 1).
\]
Here, the action space is $\CtrlActionSem{\alpha} \equiv ((\Reals \times \UnitSet) \uplus \UnitSet) \times \UnitSet$. In order to find an action, one can consider the following modified controller where a vector of variables $\VarsVec{u}$ is used to encode nondeterministic choices:
\[
  \alpha_{\VarsVec{u}} \ \equiv \ ((?u_1 = 0 \seq x \dlassign u_2 \seq y \dlassign v \cdot x) \,\cup\, (?u_1=1 \seq y \dlassign w)) \seq\, ?(y \ge 1).
\]
In the general case, any occurrence of $x \dlassign *$ can be replaced by $x  \dlassign u$ with $u$ a fresh variable, and any occurrence of $\alpha \cup \beta$ can be replaced by $(?(u=0) \seq \alpha) \cup (?(u=1) \seq \beta)$ with $u$ a fresh variable also. Then, one can consider the formula $F_{\VarsVec{u}} \equiv \ldiamond{\alpha_{\VarsVec{u}}}{\text{true}}$. Using the axioms of $\dL$, we can rewrite $F_\VarsVec{u}$ into an equivalent, quantifier-free real arithmetic formula $G_{\VarsVec{u}}$~\cite{DBLP:conf/hybrid/Platzer07}. In our example, we get:
\[
  G_{u_1, u_2} \equiv (u_1 = 0 \land v \cdot u_2 \ge 1) \lor (u_1 = 1 \land w \ge 1).
\]
Finally, given a particular state $s$, one can replace any variable not in $\VarsVec{u}$ by its value in $s$ and then generate an action by computing a model of the resulting formula using an SMT solver. In this particular case, assuming that $s(v) = 2$ and $s(w) = 0$, we seek models of the following formula:
\[
  (u_1 = 0 \land 2 \cdot u_2 \ge 1) \lor (u_1 = 1 \land 0 \ge 1).
\]
A valuation for $\VarsVec{u}$ can be naturally mapped into an action. For example, the model $\{u_1 \mapsto 0, u_2 \mapsto 4 \}$ of the formula above is mapped to action $(\SumLeftInj{(4, \Unit)}, \Unit)$, which can be interpreted as ``take the left branch and choose $x := 4$''.

\section{Details on the Inference Strategy Language}\label{ap:inf-details}

\subsection{Analysis of the Strategy from Figure~\ref{fig:overview-train-global}}\label{ap:train-global-inf-strategy}
In this section, we analyze the inference strategy of the shield defined in Figure~\ref{fig:overview-train-global}. The first line is:
\[ \TrainThetaSlopeMin, \TrainThetaSlopeMax \slassign \slaggregate i, j: (\omega_j - \omega_i)/(u_j - u_i) \sland (\eta_j - \eta_i)/(u_j - u_i) \slwhen u_j > u_i. \]
First, note that $\TrainThetaSlopeMin, \TrainThetaSlopeMax \slassign e$ is simply a syntactic sugar for $\TrainThetaSlopeMin \slassign e \seq \TrainThetaSlopeMax \slassign e$, as defined in Figure~\ref{fig:inf-lang-syntax-semantics}. Let us thus focus on the second assignment. The corresponding proof obligation is:
\[ (\omega_i = \TrainThetaSlope u_i + \TrainThetaBias - \eta_i) \land (\omega_j = \TrainThetaSlope u_j + \TrainThetaBias - \eta_j) \land (u_j > u_i) \limplyI \TrainThetaSlope \le (\omega_j - \omega_i)/(u_j - u_i) + (\eta_j - \eta_i)/(u_j - u_i),  \]
where irrelevant assumptions have been elided. As can be seen, the guard specified by the $\textsc{when}$-clause appears as a crucial assumption. During evaluation, the inference statement fails when a positive weight is put on a pair of steps where $u_j > u_i$ does not hold (and thus $\OptNone$ is returned). The second inference assignment from Figure~\ref{fig:overview-train-global} is:
\[ \TrainThetaBiasMax \slassign \slaggregate i: \omega_i - \TrainThetaSlopeMax u_i \sland \eta_i \slwhen u_i \le 0. \]
Here, the associated proof obligation is:
\[ (\omega_i = \TrainThetaSlope u_i + \TrainThetaBias - \eta_i) \land (\TrainThetaSlopeMax \ge \TrainThetaSlope) \land (u_i \le \omega_i - \TrainThetaSlopeMax u_i + \eta_i) \limply (\TrainThetaBias \le 0). \]
Note how the definition of $\TrainThetaSlopeMax$ (i.e. $\TrainThetaSlopeMax \ge \TrainThetaSlope$) is included as an assumption since it appears in the \textsc{aggregate} expression. This behavior is formalized in the definition of $\VarDef{\cdot}$ in Figure~\ref{fig:inf-lang-obligations}.

\subsection{Evaluating Probabilistic Tail Bounds}\label{sec:eval-tail-expressions}

In this Appendix, we provide more details on how the $\InvCCDF$ operator used by the inference strategy language can be implemented. Note however that optimizing this operator in the general case is orthogonal to our main contributions and goes beyond the scope of this paper.

\paragraph{Gaussian case.} When the noise component of an \textsc{aggregate} statement is a linear combination of Gaussian variables, the resulting $\InvCCDF$ expression can be evaluated optimally. Note that the coefficients and Gaussian hyperparameters can be state-dependent since all state variables can be substituted by their concrete values before $\InvCCDF$ is evaluated. The following properties make the handling of Gaussian noise particularly straightforward:
\begin{itemize}
  \item If $X_i \sim \slnormal(\mu_i, \sigma_i^2)$ are independent and $\lambda_i \in \Reals$, then $\sum \lambda_i X_i \sim \slnormal\bigl(\sum_i \lambda_i X_i, \, \sum_i \lambda_i^2 \sigma_i^2\bigr)$.
  \item $\Prob_{X \sim \slnormal(\mu, \sigma^2)} \{ X > \delta_\Eps \} = \Eps$ where $\delta_\Eps \equiv \mu + \sqrt{2} \cdot \sigma \cdot \ErfInv(1 - 2\Eps)$ and $\ErfInv$ is the inverse of the error function defined by $\Erf(z) \equiv (2/\sqrt{\pi}) \int_0^z e^{-t^2}dt$ for $z \in \Reals$~\cite{bertsekas2008introduction}.
\end{itemize}

\paragraph{Discrete case} Discrete distributions that are expressed using a finite combination of Bernouilli variables can be generally handled by computing the associated, finite probability table. Also:
\begin{itemize}
  \item If $X \sim \slbernouilli(p)$, then $\InvCCDF(X, \Eps) = 0$ if $\Eps \ge p$ and $1$ otherwise.
  \item If $X_i \sim \slbernouilli(p_i)$ i.i.d, then $\InvCCDF((1/n)\sum_{i=1}^n X_i, \Eps) = 0$ if and only of $\Eps \ge \Prob(\forall i \ X_i = 0) = 1 - \prod_i(1-p_i)$.
\end{itemize}

\paragraph{Use of concentration inequalities} Concentration inequalities (e.g. Chebyshev, Cantelli, Hoeffding, or Chernoff bounds) can be used to implement the $\InvCCDF$ operator for arbitrary distributions, using simple features such as their mean, variance, support or moment-generating function~\cite{bertsekas2008introduction}. Such features benefit from strong compositional properties\footnote{For example, $\Expectation(X + Y) = \Expectation(X) + \Expectation(Y)$ and $\Variance(X + Y) = \Variance(X) + \Variance(Y)$ for $X$ and $Y$ independent random variables.}, making them easy to estimate in a large variety of cases. The simplest concentration bound is the Chebyshev inequality. It states that for any distribution $X$ with mean $\mu$ and variance $\sigma^2$ and for any $k > 0$, we have: \[\Prob(|X-\mu| > k \sigma) \le \frac{1}{k^2}. \]
We can therefore define $\InvCCDF(X, \Eps) = \mu + \sigma / \sqrt{\Eps}$. In particular, in a typical case where we have independent variables $X_1,\dots,X_n$ with mean $0$ and variance $\sigma^2$ along with some convex coefficients $\lambda_1,\dots,\lambda_n$,
this yields $\InvCCDF\bigl(\sum_i \lambda_i X_i, \, \Eps\bigr) = \SqrtI{\sum_i \lambda_i^2} \cdot \sigma \cdot \bigl(1/\sqrt{\Eps}\bigr)$.

The Hoeffding inequality often provides tighter bounds when aggregating independent noise variables with bounded support. It states that if $X_1, \dots, X_n$ are independent and that $a_i \le X_i \le b_i$ and $\Expectation(X_i) = \mu_i$ for all $i$, then the following holds for all $t > 0$:
\[ \Prob((X_1+\dots+X_n) - (\mu_1 + \dots + \mu_n) > t) \le \exp\biggl(-\frac{2t^2}{\sum_{i=1}^n(b_i-a_i)^2}\biggr). \] In the particular case where $X_1, \dots, X_n$ are independent with mean $0$, $a \le X_i \le b$ for all $i$ and $\lambda_1,\dots\lambda_n$ are convex coefficients, the following holds for all $t>0$: \[\Prob\biggl(\sum_i\lambda_iX_i > t\biggr) \le \exp\biggl(-\frac{2t^2}{(b-a)^2 \cdot \sum_i \lambda_i^2}\biggr).\] Thus, in this case, we can define $\InvCCDF\bigl(\sum_i \lambda_i X_i, \, \Eps\bigr) = \SqrtI{\sum_i \lambda_i^2} \cdot (b-a) \cdot \ln \bigl(1/\sqrt{\Eps}\bigr)$. This bound is exponentially better than the Chebyshev bound in terms of $\Eps$, as it features a factor of $\ln \bigl(1/\sqrt{\Eps}\bigr)$ instead of $1/\sqrt{\Eps}$. However, both bounds improve at the same rate of $1/\sqrt{n}$ as $n$ gets larger, assuming a uniform distribution $\lambda_i=1/n$ of the convex coefficients.

\section{Proofs}

\subsection{Main Safety Theorem}\label{ap:main-theorem-proof}

This section provides a proof sketch for Theorem~\ref{thm:main-safety}. Throughout this section, we fix a shield specification $\tuple{\Const, \dots}$, a compatible environment $\tuple{\tuple{S, \dots}, \dots, c, u, \ObsAvailFun, \ObsMeasurementFun}$ and the associated shielded MDP $\hat E = \tuple{\hat S, \hat A, \hat P,\dots}$. Our proof uses the notions of a \emph{consistent history} and of a \emph{consistent shielded state}.

\begin{definition}[Consistent history]
  A history $h \in \Hist$ is said to be \emph{consistent} if for all $(s, z, \LocBounds) \in h$, we have $(c \CatValuation u, \StateMapping(s)) \in \sem{\BoundConj{\LocBounds}}$ and $z \subseteq \ObsAvailFun(s)$.
\end{definition}

\begin{definition}[Consistent shielded state]
  A state $\hat s=(s,h,\GlobBounds,\EpsRem) \in \hat S$ is said to be \emph{consistent} if and only if  $h$ is consistent, $\EpsRem \ge 0$ and $(c \CatValuation u, \StateMapping(s) \CatValuation \GlobBounds) \in \sem{\Assum \land \Bound \land \Inv}$.
\end{definition}

Theorem~\ref{thm:main-safety} follows from the following lemma, using the fact that the system's invariant implies the safety condition (Definition~\ref{sec:shield-spec}, item~(\ref{obl:inv-implies-post})).

\begin{lemma}[Consistency Preservation]
 For any consistent state $\hat s=(\dots,\EpsRem)$ and action $a \in \hat A$, the next state $\hat s' = (\dots,\EpsRem') \sim \hat P(\hat s, \hat a)$ is also consistent with probability at least $1 - (\EpsRem - \EpsRem')$.
\end{lemma}
\begin{proof}
We show how the consistency of $(s, h, \GlobBounds, \EpsRem)$ is preserved throughout Algorithm~\ref{alg:shield} with suitable likelihood. From Line~\ref{line:algo-start} to Line~\ref{line:start-exec-inference} (excluded), only the observation availability component of history $h$ is changed, in a way that preserves consistency. From the soundness of inference (Theorem~\ref{thm:inf-soundness}), the soundness of the measurement function (Definition~\ref{def:compat-env}, item~(\ref{item:correct-observation})) and the monotonicity of the invariant (Section~\ref{sec:shield-spec}, item~(\ref{obl:inv-monotone})), each iteration of the loop starting at Line~\ref{line:start-exec-inference} preserves the consistency of $(s, \, h \cdot (s, \ObsAvailFun(s), \{p \mapsto v(p) : p \in \LocalParam\}), \, \{p \mapsto v(p) : p \in \GlobalParam\}, \EpsRem)$, barring an event of probability bounded by the decrease in $\EpsRem$. This ensures that consistency still holds before executing Line~\ref{line:mon-start}, with suitable probability. From Lemma~\ref{lem:ctrl-monitor} and controller totality (Section~\ref{sec:shield-spec}, item~(\ref{obl:ctrl-total})), a fallback is always guaranteed to exist on Line~\ref{sec:shield-spec}. Finally, from the model's accuracy (Definition~\ref{def:compat-env}, item~(\ref{item:correct-model})) and the invariant preservation property (Section~\ref{sec:shield-spec}), Line~\ref{line:perform-underlying} preserves consistency.
\end{proof}

\subsection{Soundness of the Inference Strategy Language}\label{ap:inf-lang-soundness-proof}

We provide a proof of Theorem~\ref{thm:inf-soundness} below.

\newcommand{\latestEnv}{w}
\newcommand{\historyEnv}{h}

\begin{proof}
  Let $\tuple{\Const,\dots,\Inference}$ a shield specification and let us assume that $\Fvalid \ProofSem{\Inference}$. Let $a \in \InfActionSem{\Inference}$ an inference action. We need to show that $\sem{\Inference}(a)$ is a list of \emph{sound} inference assignments. By the semantics of strategy composition, it is enough to consider an arbitrary assignment $p \slassign e$ in the strategy and show that it produces a sound symbolic inference assignment. Thus, let $a \in \InfActionSem{p \slassign e}$ and $(p, e', \Eps) = \sem{\Inference}(a)$. We must show that $(p, e', \Eps)$ is sound.

  To do so, by Definition~\ref{def:sound-inf-assignment}, we must fix an arbitrary history and show that updating the latest bound instantiation with $e'$ preserves the truth of $\Bound_p$. Thus, let us consider:
  \begin{itemize}
    \item some symbol valuations $c \in \Const \toI \Reals$ and $u \in \Unknown \toI \Funs$,
    \item a history size $n \in \Nats$,
    \item a state history $(s_1 \dots s_{n}) \in \ProgState^{n}$,
    \item an observation availability history $(z_1 \dots z_{n}) \in \ObsAvail^{n}$, and
    \item a bound instantiation history $(b_1 \dots b_{n}) \in \Inst^{n}$.
  \end{itemize}
  Let us also define:
  \begin{itemize}
    \item some sampled noise variables $\eta_1 \cdots \eta_{n} \sim \sem{\Noise}(c)$,
    \item corresponding observations $o_i = \{\omega \mapsto \sem{\Obs_\omega}(c \CatValuation u, s_i \CatValuation \eta_i) : \omega \in \Param \cap z_i \}$,
    \item the valuation $v = \bigl(s_n \CatValuation b_n \bigr) \cup \bigl(\BigCatValuation_{1 \le i \le n} \TagWithIndex{s_i}{i} \CatValuation \TagWithIndex{b_i}{i} \CatValuation \TagWithIndex{o_i}{i}\bigr)$ summarizing the history and
    \item the value $r = \sem{e'}(c, v)$ of the concrete bound obtained using $v$.
  \end{itemize}
  In addition, let us assume that $(c \CatValuation u, s_i) \in \sem{\Assum \land \Inv \land \BoundConj{b_i}}$ for all $i$. We must show: \[\Prob\left\{ (c \CatValuation u, s_{n}) \notin \sem{\BoundConj{\UpdateMap{b_n}{p}{r}}} \right\} \le \Eps.\]
  Let us abbreviate $\latestEnv = (c \CatValuation u, s_n \CatValuation b_n)$ the data about the current state and $\historyEnv = (c \CatValuation u, v)$ the full historical data. It is enough to show that $\Prob\{ \latestEnv \notin \sem{\BoundWith{p}{r}} \} \le \Eps$. We do so by case disjunction on $e$.
  \paragraph{Case $e = (\theta \slwhen G)$} In this case, the semantics of the inference strategy language gives us $e' = (\theta \WhenSymbBE G)$, and $\Eps = 0$. Let us assume that $\historyEnv \in \sem{G}$. We must show that $\latestEnv \in \sem{\BoundWith{p}{r}}$. By validity of the generated proof obligation, we have $\latestEnv \in \sem{A \limplyI G \limplyI \BoundWith{p}{\theta}}$ where $A \equiv \Assum \wedge (\bigwedge \VarDef{\FV(e)})$. Since $h \in \sem{G}$, we also have $\latestEnv \in \sem{G}$. In addition, we have $\latestEnv \in \sem{A}$ since $(c \CatValuation u, s_i) \in \sem{\Assum \land \Inv \land \BoundConj{b_i}}$ for all $i$. Thus, we get $\latestEnv \in \sem{\BoundWith{p}{\theta}}$. We conclude by observing that $r = \sem{\theta}(\latestEnv)$.
  \paragraph{Case $e = (\slbest i_1, \dots, i_n : \theta  \slwhen G)$} This case is very similar to the previous one.
  \paragraph{Case $e = (\slaggregate i_1, \dots, i_n: \theta_1 \sland \theta_2  \slwhen G)$} By the semantics of the inference strategy language, we have $e' = (\theta_1' + \theta_2' \WhenSymbBE G')$, where $G' = \bigwedge_{\lambda, j \in D} \FSubst{G}{i}{j}$, $\theta_1' \equiv \sum_{\lambda, j \in D} {\lambda \FSubst{\theta_1}{i}{j}}$ and $\theta_2' = \InvCCDF\bigl(\sum_{\lambda, j \in D} {\lambda \FSubst{\theta_2}{i}{j}}, \Eps \bigl)$. Let us assume that $\historyEnv \in \sem{G'}$. We must show that $\Prob(\latestEnv \notin \sem{\BoundWith{p}{r}}) \le \Eps$. By validity of the generated proof obligation and by alpha-renaming, we have: \[h \in \sem{\FSubst{A}{i}{j} \limplyI \FSubst{G}{i}{j} \limplyI \BoundWith{p}{\FSubst{(\theta_1+\theta_2)}{i}{j}}}\]
  for all $\lambda,j \in D$. In addition, since $\historyEnv \in \sem{G'}$, then $\historyEnv \in \sem{\FSubst{G}{i}{j}}$ for all $j$. Finally, $\historyEnv \in \sem{\FSubst{A}{i}{j}}$ for all $j$ by our \emph{history consistency} assumption. As a consequence, we have $h \in \sem{\BoundWith{p}{\FSubst{(\theta_1+\theta_2)}{i}{j}}}$ for all $j$ and thus $h \in \sem{\BoundWith{p}{b^*}}$ where $b^* \equiv \sum_{\lambda, j \in D} \lambda \FSubst{(\theta_1 + \theta_2)}{i}{j}$ since $\Bound_p$ is monotone in $p$ and thus convex (assuming without loss of generality that $p$ is an \emph{upper}-bound). Also, by monotonicity of $p$ again, we have $\historyEnv \in \sem{b \ge b^* \limplyI \BoundWith{p}{b}}$. Therefore, we deduce:
  \[ \Prob(\historyEnv \in \sem{\neg\BoundWith{p}{\theta_1'+\theta_2'}}) \,\le\, \Prob(h \in \sem{b^* > \theta_1' + \theta_2'}) \,=\, \Prob\Bigl(h \in \bigsem{\sum_{\lambda, j \in D} {\lambda \FSubst{\theta_2}{i}{j}} > \theta_2'}\Bigr) \,\le\, \Eps \] where the last inequality results from the definitions of $(\eta_i)_i$ and $(\omega_i)_i$ and from the definition of the $\InvCCDF$ operator.
  We established $\Prob(\historyEnv \in \sem{\neg\BoundWith{p}{\theta_1'+\theta_2'}}) \le \Eps$. In addition, since $\sem{\theta_1' + \theta_2'}(\historyEnv) = r$, we also have $\Prob(\historyEnv \in \sem{\neg\BoundWith{p}{r}}) \le \Eps$. Finally, since $\latestEnv \subseteq \historyEnv$ and $\BoundWith{p}{r}$ does not contain any indexed variable, we have $\Prob(\latestEnv \in \sem{\neg\BoundWith{p}{r}}) \le \Eps$, which concludes the proof.
\end{proof}

\end{document}